\documentclass[12pt,a4paper]{article}   	

\usepackage[utf8]{inputenc}

\usepackage[UKenglish]{isodate}

\usepackage[a4paper,margin=2cm, marginparwidth=2cm]{geometry}

\usepackage{amsmath,amssymb,amscd,amsfonts}
\numberwithin{equation}{section}

\usepackage{amsthm}

\newtheorem{theorem}{Theorem}
\newtheorem{lemma}[theorem]{Lemma}
\newtheorem{corollary}[theorem]{Corollary}
\newtheorem{definition}{Definition}

\usepackage[margin=10pt,font=small,labelfont=bf]{caption}

\usepackage{graphicx}
\graphicspath{{./figures/}}

\usepackage[dvipsnames]{xcolor} 

\usepackage[colorlinks,allcolors=blue]{hyperref}

\providecommand{\href}[2]{\texttt{#2}}
\providecommand{\url}[1]{\texttt{#1}}

\usepackage[square,semicolon,authoryear]{natbib}

\newcommand{\half}{\frac{1}{2}}
\newcommand{\bea}{\begin{eqnarray}}
\newcommand{\eea}{\end{eqnarray}}
\newcommand{\beq}{\begin{equation}}
\newcommand{\eeq}{\end{equation}}
\newcommand{\nnel}{\nonumber \\ {}}

\providecommand{\eqref}[1]{(\ref{#1})}

\newcommand{\figref}[1]{Fig.~\ref{#1}}
\newcommand{\tabref}[1]{Table~\ref{#1}}

\newcommand{\secref}[1]{Section~\ref{#1}}

\newcommand{\appref}[1]{Appendix~\ref{#1}}

\newcommand{\defref}[1]{Definition~\ref{#1}}

\newcommand{\corref}[1]{Corollary~\ref{#1}}

\newcommand{\lemref}[1]{Lemma~\ref{#1}}
\newcommand{\Lemref}[1]{Lemma~\ref{#1}}

\newcommand{\thref}[1]{Theorem~\ref{#1}}
\newcommand{\Thref}[1]{Theorem~\ref{#1}}

\newcommand{\tsedef}[1]{\textsc{#1}}

\newcommand{\tsevec}[1]{\mathbf{#1}}
\newcommand{\tsemat}[1]{{\mathbf{\textsf{#1}}}}

\newcommand{\pvec}{\tsevec{p}}

\newcommand{\uvec}{\tsevec{u}}
\newcommand{\vvec}{\tsevec{v}}
\newcommand{\wvec}{\tsevec{w}}
\newcommand{\xvec}{\tsevec{x}}
\newcommand{\yvec}{\tsevec{y}}
\newcommand{\zvec}{\tsevec{z}}

\newcommand{\zerovec}{\tsevec{0}}

\newcommand{\Amat}{\tsemat{A}}

\newcommand{\Dmat}{\tsemat{D}}

\newcommand{\Lmat}{\tsemat{L}}
\newcommand{\Mmat}{\tsemat{M}}
\newcommand{\Pmat}{\tsemat{P}}

\newcommand{\Smat}{\tsemat{S}}

\newcommand{\Tmat}{\tsemat{T}}

\newcommand{\Wmat}{\tsemat{W}}

\newcommand{\Zmat}{\tsemat{Z}}

\newcommand{\unitmat}{\hbox{\textsf{1}\kern-.25em{\textsf{I}}}}

\newcommand{\Acal}{\mathcal{A}}
\newcommand{\Bcal}{\mathcal{B}}
\newcommand{\Ccal}{\mathcal{C}}

\newcommand{\Dcal}{\mathcal{D}}
\newcommand{\Ecal}{\mathcal{E}}
\newcommand{\Fcal}{\mathcal{F}}
\newcommand{\Gcal}{\mathcal{G}}

\newcommand{\Lcal}{\mathcal{L}}
\newcommand{\Mcal}{\mathcal{M}}
\newcommand{\Ncal}{\mathcal{N}}
\newcommand{\Pcal}{\mathcal{P}}

\newcommand{\Rcal}{\mathcal{R}}
\newcommand{\Scal}{\mathcal{S}}
\newcommand{\Scalhat}{\widehat{\mathcal{S}}}
\newcommand{\Tcal}{\mathcal{T}}
\newcommand{\Ucal}{\mathcal{U}}
\newcommand{\Vcal}{\mathcal{V}}
\newcommand{\Wcal}{\mathcal{W}}
\newcommand{\Xcal}{\mathcal{X}}

\newcommand{\Zcal}{\mathcal{Z}}

\newcommand{\Cbb}{\mathbb{C}}
\newcommand{\Rbb}{\mathbb{R}}

\newcommand{\Ggraph}{\mathfrak{G}}
\newcommand{\link}{\ell}
\newcommand{\linkmap}{\Lambda}
\newcommand{\linkmaptilde}{\tilde{\Lambda}}

\newcommand{\kout}{k^\mathrm{(out)}}

\newcommand{\Ccalsymm}{\mathcal{C}^\mathrm{(symm)}}
\newcommand{\Dcalsymm}{\mathcal{D}^\mathrm{(symm)}}
\newcommand{\Lcalsymm}{\mathcal{L}^\mathrm{(symm)}}
\newcommand{\Ggraphsymm}{{\Ggraph}^\mathrm{(symm)}}

\newcommand{\Hcalwc}{\mathcal{H}^\mathrm{(wc)}}
\newcommand{\Lcalwc}{\mathcal{L}^\mathrm{(wc)}}
\newcommand{\Ncalwc}{\mathcal{N}^\mathrm{(wc)}}

\newcommand{\Hcalsc}{\mathcal{H}^\mathrm{(sc)}}
\newcommand{\Lcalsc}{\mathcal{L}^\mathrm{(sc)}}
\newcommand{\Ncalsc}{\mathcal{N}^\mathrm{(sc)}}

\newcommand{\linkrev}{\ell^\mathrm{(rev)}}

\newcommand{\Scalrev}{\mathcal{S^\mathrm{(rev)}}}
\newcommand{\Crev}{C^\mathrm{(rev)}}

\newcommand{\itervectorsymbol}{w}
\newcommand{\itervector}{\tsevec{\itervectorsymbol}}
\newcommand{\itervec}[1]{\tsevec{\itervectorsymbol}{(#1)}}

\begin{document}

\renewcommand{\thefootnote}{\fnsymbol{footnote}}

\begin{center}
 {\Large\textbf{From Signed Networks to Group Graphs}
 }
 \\[\baselineskip]
 {\large \href{http://www.imperial.ac.uk/people/t.evans}{T.S.\ Evans}\footnote{ORCID:  \href{http://orcid.org/0000-0003-3501-6486}{\texttt{0000-0003-3501-6486}} }}
 \\[0.5\baselineskip]
 \href{http://complexity.org.uk/}{Centre for Complexity Science}, and \href{http://www3.imperial.ac.uk/theoreticalphysics}{Abdus Salam Centre for Theoretical Physics},
 \\
 Imperial College London, SW7 2AZ, U.K.
 \\
 6\textsuperscript{th} August, 2026 (version 3).
\end{center}

\begin{abstract}
I define a \tsedef{group graph} which encodes the symmetry in a dynamical process on a network. 
Group graphs extend signed networks, where links are labelled with plus or minus one, by allowing link labels from any group and generalising the standard notion of balance. 
I show that for processes on a balanced group graph the time evolution is completely determined by the network topology, not by the group structure.
This unifies and extends recent findings on signed networks (Tian \& Lambiotte, 2024a) and complex networks (Tian \& Lambiotte, 2024b). 
I will also relate the results discussed here to existing work such as 
the ``group labelling'' of Edelman and Saks (1979), 
the ``group graph'' of Harary, Lindstr\"{o}m \& Zetterstr\"{o}m (1982), 
a ``voltage graph'' (Gross, 1974), 
a ``gain graph'' (Zaslavsky 1989), 
and ``group synchronisation'' (Karp et al, 2003). 
I will work with a more general case where edges need not be reciprocated and the labels of reciprocated edges need not be inverses of each other.
Finally, I will review some promising applications for network dynamics and symmetry-driven modelling including status, clusterability, edges with a zero label, weak balance, unbalanced group graphs and using monoids.
\end{abstract}

\setcounter{footnote}{0}
\renewcommand{\thefootnote}{\arabic{footnote}}


\section{Introduction}

Complex networks are designed to work in unruly problems where there is little of the elegant symmetry that underpins much of science. So group theory has played no significant role in network analysis as illustrated by the lack of entries on symmetry or group theory in the index of most textbooks on complex networks. 
For example neither \citet{WF94} nor \citet{C21} discuss group theory and when they do discuss symmetry it is the context of the very limited symmetry of seen in pairs of reciprocated directed edges or in one undirected edge.
This can be contrasted with the prominent role of regular lattices in science where ideas such as translation invariance, rotation and reflection symmetries of the lattice are intimately linked to important physical properties through group theory.

However, symmetry also appears in science in terms of some ``internal space'' that is not connected to the symmetries of space and time. For example, our understanding of fundamental particles is built on symmetries of such internal spaces. Take the three types of pion discovered in 1947 and 1949. The pions have a very similar mass yet also have some striking yet systematic differences such as their electric charges of $+1$, $0$ and $-1$ (in units of the electron charge). It turns out these patterns can be understood in terms of the group $SU(2)$ as ``flavour'' symmetries in which two different types of quantum particle, the `up' and the `down' quark, are mixed by a symmetry at each space-time point i.e.\ this is a symmetry in some internal ``up-down'' space and has nothing to do with the symmetry of space-time. Such internal symmetries are still a key part of the current best theory of fundamental physics.

So the question posed in this paper is if there is a role for symmetry and group theory in network analysis, not in terms of the symmetries of the network (the equivalent of the space-time symmetry) but in terms of symmetries in processes occurring on the network, in some ``internal space'' added on top of the network. The answer is yes and such symmetries may be encoded in terms of a structure I call a \tsedef{group graph}. 

The inspiration for my construction comes from the work on conserved processes on signed networks by \citet{TL24}.
\tsedef{Signed networks} \citep{LS52,H53,R54,CH56,N61,HNC65,D67a,S68,AR68} emerged from social network analysis for contexts where relations between actors can be identified as one of three types: positive, none, or negative. Though the precise interpretation of these three types can vary depending on the study, it requires that a network is encoded with a link (a binary relation, an edge)  carrying either a positive or negative label, alongside the absence of a link. Signed networks have continued to be used in network analysis in a wide range of settings; for instance see \citet{DM96,DM09,KLB09,TB09,ST10,SLT10,WJB19,KCN19,ADRD20,TL24}. Signed networks often merit discussion in standard textbooks, for example see section 4.4 of \citet{WF94} or section 4.2 of \citet{C21}, and are of current interest as shown by a recent review \citep{DCGFVFT25}. 

One important concept for signed networks is the idea of \tsedef{structural balance}, often referred to as just \tsedef{balance}, derived from psychology \citep{H46}. This has a simple interpretation in terms of the number of negative links in any closed walk of a signed network \citep{H53,CH56}, a concept often used in the analysis of signed networks \citep{HNC65,DM96,DM09,ST10,SLT10,KCN19,ADRD20,TL24,DCGFVFT25}. I will generalise balance to the case of group graphs.

An important goal of this paper is to show how the results for a diffusion process on a balanced signed network of \citep{TL24} and for a balanced network with diffusing complex phases on the nodes \citep{TL24a} can be generalised to any internal space defined on the nodes when there is a symmetry in the bilateral interactions in this internal space.

Another aim of this paper is to connect the properties of group graphs outlined here to previous work.
The closest in style to the approach used here is that of \citet{HLZ82} who defined the same structure and used the same name. The motivation there comes from theories of social balance as encoded as signed networks \citep{H53,CH56,HNC65}. However, \citet{HLZ82} appears to have had little impact 
(at the time of writing it had only three citations on both
``\href{https://app.dimensions.ai/details/publication/pub.1009664995}{Dimensions}''
and 
``\href{https://www.scopus.com/record/display.uri?eid=2-s2.0-0020175823}{Scopus}'') 
and was only found by this author after the first version of this paper was posted. 
The oldest work on group graphs \citep{G74} sees them as a generalisation of Kirchoff's laws for electrical circuits with constant voltage and current flows, and there the group graph construction is known as a 
\tsedef{voltage graph}\footnote{For informal summaries see the \href{https://en.wikipedia.org/wiki/Voltage_graph}{voltage graph article on Wikipedia} and the \href{https://mathworld.wolfram.com/VoltageGraph.html}{voltage graph article on Mathworld}.}.
Work on voltage graphs forms part of topological graph theory and the goals of that field may explain why this previous work has not reached texts focussing on social network analysis, such as  \citet{WF94}, network science books, such as \citet{C21}, or recent work generalising signed networks \citep{TL24,TL24a,BP24,TKSL24}. The voltage graph approach also produced another branch, and another name, 
\tsedef{gain graphs}\footnote{For an informal summary see the \href{https://en.wikipedia.org/wiki/Gain_graph}{Wikipedia article on gain graphs} 
or  \href{https://people.math.binghamton.edu/zaslav/Bsg/index.html}{The Home Page of Signed, Gain, and Biased Graphs by Thomas Zaslavsky}.} 
\citep{Z89} typically applied to biased graph theory and matroid theory. Again, this gain graph viewpoint neither connected with the older social network analysis approach of \citet{HLZ82} and neither has it been connected to recent work on signed networks so links to related results in other areas will be made at appropriate points in this paper. In a similar way, the \tsedef{group labelling} of a graph due to \citet{ES79} is another language and viewpoint for the same concepts but here framed in terms of abstract mathematics. While this does not seem to have had much impact 
(two citations on ``\href{https://app.dimensions.ai/details/publication/pub.1024713995}{Dimensions}'' and five on ``\href{https://www.scopus.com/pages/publications/0039439105}{Scopus}''), 
\citet{MS97} do link this work to Harary's work on signed networks in social science \citep{H53}. The same ideas also underlie the topic of \tsedef{group synchronisation} \citep{KEES03,GK06a,S11e,AMMSS18,AF19,S21e} where, for example, we are trying to synchronise clocks or phases of nodes in a network given an estimate of the time or phase difference between nodes at the ends of an edge.

The novelty in this paper comes from the links made between all these different approaches and that the work here is for the most general case of an arbitrary directed graph and an arbitrary non-Abelian group. The literature cited here is often developed for undirected graphs or for abelian groups, or both, but the key concepts are already in the literature though appearing in many different guises.

The outline for the paper is as follows. I start in \secref{s:background} by introducing the notation and terminology used here for networks and groups. The main definitions and properties of group graphs are then developed in \secref{s:groupgraphs} including the definition of a group graph and a balanced group graph. I apply these ideas in the context of processes on networks in \secref{s:dynamics}. I will finish with a discussion in \secref{s:discussion}. There I will also examine connections with previous work on voltage graphs, group graphs, gain graphs, group synchronisation, group labelling and so forth \citep{G74,ES79,HLZ82,Z82,Z89,KEES03,GK06a,S11e}. 
This final discussion will also include a review of some of some promising applications for network dynamics and symmetry-driven modelling using group graphs including status, clusterability, the representation of edges with a zero label or weight, weak balance, unbalanced group graphs and using monoids not graphs.

\section{Background}\label{s:background}

In this section I will set out my notation and at the same time highlight some of the basic properties of groups and networks used in the rest of the paper. In particular, I will define some terminology for directed networks that is less well known.

In general I will use braces $\{ \ldots\}$ to denote a collection of elements with no order while parentheses $( \ldots )$ indicate a sequence where the elements of the collection come with a specific order. I will always write a sequence as 
\beq 
 (x_m)_{m=1}^L = (x_m, x_{m-1}, \ldots, x_2,x_1) \, .
\label{e:seqdef}
\eeq 
That is I will place the first element $x_1$ on the far right and the last element on the far left of any expression. While this does not match the left to right order in the text, it is consistent with the right to left conventions used here for matrix/vector algebra.

\subsection{Networks}\label{s:net}

A \tsedef{network} (or \tsedef{graph}) is a set of nodes (vertices) $\Ncal$ and a set of links (edges) $\Lcal$ which are pairs of nodes. 
Most of this paper is about \tsedef{directed networks} $\Dcal=(\Ncal,\Lcal)$ (also known as ``digraphs'') where links are ordered pairs of nodes denoted as the sequence $(v,u) \in \Lcal$ to indicate a link from node $u$ to node $v$ in accordance with the convention for sequences set out in \eqref{e:seqdef}.
The link $(v,u)$ in a directed network is distinct from the link $(u,v)$. \tsedef{Self-loops} are allowed such as the directed link $(v,v)$ from node $v$ back to itself. A link $(v,u)$ is said to be \tsedef{reciprocated} if there is also an link in the opposite direction, i.e.\ if $(u,v)$ is also an link in the network.   Note that the reciprocated links need not be identical in any other ways and a pair of reciprocated links may carry different weights or labels. 

This link notation does not support multilinks, where there can be more than one distinct link from a node $u$ to a node $v$. For the majority of the time I will use the notation outlined above to simplify the text.  Most results discussed here can be generalised with little trouble networks with multilinks.

A \tsedef{walk} $\Wcal(v_L,v_0)$ from node $v_0$ to node $v_L$ on a network (directed or undirected) is a sequence
of nodes, $(v_L,v_{L-1},\ldots,v_2,v_1,v_0)=(v_n)_{n=0}^{L}$ where there is a link $\link_n$ from each node $v_{n-1}$ to the next node in the sequence $v_{n}$ so that $\ell_m = ( v_{m},v_{m-1} ) \in \Lcal$ where $m\in\{1,\ldots,L\}$ ($\ell_m = \{ v_{m},v_{m-1} \} \in \Lcal$ for an undirected network).   The \tsedef{length} $L$ of a walk is the number of nodes in the walk minus one, $L(\Wcal) = |\Wcal|-1$.
More formally a walk $\Wcal$ on a directed network satisfies 
\beq
	\Wcal
	= 
	\big(
	v_n | v_n \in \Ncal, n\in\{0,1,2,\ldots,L\} 
	\, ,  
	\text{ and } 
	\link_m = (v_{m}, v_{m-1}) \in \Lcal, \, m\in\{1,2,\ldots,L\} 
	\big) 
	\, .
	\label{e:walkdef}
\eeq
Later I will use the notation that $\Wcal(v,u)$ is some walk starting from node $u$ and ending at node $v$. In general there are many such walks but I will show later that in special cases, some properties of group graphs only depend on the initial and final nodes in a walk.


The \tsedef{links of a walk} $\Lcal(\Wcal)$ form the sequence $( \link_m )_{m=1}^{L} = ( (v_m,v_{m-1}) )_{m=1}^{L}$ of subsequent pairs of nodes in the walk\footnote{For networks with multilinks a walk must be defined in terms of the sequence of links in the walk since a sequence of two nodes does not uniquely define an edge. However, it is straightforward to rewrite the work here in terms of walks defined using edge sequences and all the key ideas still work in this context.}. The length of a walk is also the number of links in the walk.

I can join two walks to create a new walk if the second walk $\Xcal$ starts where the first walk $\Wcal$ finished. More formally, I can concatenate two walks $\Wcal = (v_n)_{n=0}^{L}$ and  $\Xcal = (u_m)_{m=0}^{M}$ to create a new walk $\Xcal \circ \Wcal$ provided the last node of the first part of the walk $\Wcal$, shown on the on the right, is identical to the first node of the second part of the walk $\Xcal$, shown on the left, i.e.\ $v_L = u_0$. That is 
\begin{subequations}
	\begin{eqnarray}
		\Xcal \circ \Wcal
		 =
		 (u_m)_{m=0}^{M} \circ (v_n)_{n=0}^{L}
		 &=&
		 (u_M, u_{M-1}, \ldots, u_1,u_0 = v_L, v_{L-1}, \ldots, v_1,v_0  )
		 \\
		 &=&
		 (y_r )_{r=0}^{M+L}
		 \quad \text{where}
		 \quad
		 y_r
		 =
		 \begin{cases}
		 v_r & 0 \leq r \leq L \\
		 u_{r-L} & L \leq r \leq M+L
		 \end{cases}
	\end{eqnarray}
	\label{e:walkconcat}
\end{subequations}
It useful to allow a \tsedef{trivial walk} $\Wcal_v$ of length zero which is the sequence of a single node $v$, so $\Wcal_v=(v)$. 
A \tsedef{closed walk} $\Ccal$ in a network (sometimes known as a \tsedef{circuit}) is a walk in which the first and last nodes are identical. More formally
\begin{multline}
	\Ccal
	= 
	\Big(v_n | v_n \in \Ncal, n\in\{0,1,2,\ldots,L\} 
	\, ,  \;
	v_0=v_L
	\, ,  
	\\
	\text{ and } 
	\;
	\link_m = (v_{m}, v_{m-1}) \in \Lcal, \, m\in\{1,2,\ldots,L\} 
	\Big) 
	\, .
	\label{e:cwalkdef}
\end{multline}
I define a \tsedef{cycle} $\Ccal^\mathrm{(cycle)}$  to be a closed walk where all the nodes are distinct except the first and last which are equal\footnote{It is redundant to include both the first and last nodes in the notation of a cycle as occurs here. However, to avoid additional notation I will simply use the notation for a walk to represent cycles.}, i.e.\
\begin{multline}
	\Ccal^\mathrm{(cycle)}
	= 
	\Big(v_n | v_n \in \Ncal, n\in\{0,1,2,\ldots,L\} 
	\, ,  
	\\
	\text{ and } 
	\;
	v_x=v_y \; \text{iff} \;  x=y \, \text{or} \, x=0, y=L \, \text{or} \, x=L, y=0 
	\, ,  
	\\
	\text{ and } 
	\;
	\link_m = (v_{m}, v_{m-1}) \in \Lcal, \, m\in\{1,2,\ldots,L\} 
	\Big) 
	\, .
	\label{e:cycledef}
\end{multline}
Note that this also ensures that every link $\ell_m \in \Lcal(\Ccal)$ in a cycle is distinct unlike for a closed walk where links and nodes can appear more than once. It is important to note that some definitions of cycle in the literature put a lower limit on the size of a cycle. Here, cycles can be of any length including zero as the trivial walk $\Wcal_v =(v)$ from node $v$ satisfies the formal definition of a cycle \eqref{e:cycledef} with length $L=0$.

I will also need to look at two special types of network. Following \citet{HNC65,DM96,DM09}, I define a \tsedef{signed directed network} to be a directed network where every link has one of two labels: a label for a ``positive'' link --- I will use $e$ --- or a label for a ``negative'' link  --- I will use $a$. These labels are almost always represented by the numbers $+1$ or $-1$ respectively and I will also explain this second notation later. Several examples of signed networks are shown in figures below. The other type of network I will mention is a \tsedef{multilayer network} \citep{KABGMP14,B18d}, or see sec.4.2 of \citet{C21}. In the version I use here, there are several layers to the network and each nodes exists on every layer but now the links may connect the copy of a node $u$ on layer $\alpha$ to a copy of a second node $v$ lying on layer $\beta$ where the layers $\alpha$ and $\beta$ need not be the same.

\subsubsection{Connectivity in directed networks}

My main focus is on directed graphs where the connectivity between nodes (``reachability'') can be asymmetric. That is there can be a walk from node $u$ to node $v$ but this does not guarantee that there is a walk from node $v$ to node $u$. 
The definition of the concept of {balance} for signed directed networks \cite{HNC65,DM96,DM09} does not use walks and cycles, but instead uses the concepts of a \tsedef{semiwalk}  and a \tsedef{semicycle}. Essentially a semiwalk moves from node to neighbouring node, so subsequent nodes are connected by at least one link. However, unlike a walk, the link can be in either direction. 
More formally 
\begin{definition}[Semiwalk] \label{d:semiwalk}
	A \tsedef{semiwalk} $\Scal$ in a directed network $\Dcal = (\Ncal,\Lcal)$ is a sequence of nodes, $(v_n)_{n=0}^L$ 
	where there is a directed link $\link_m$ from each node $v_{m-1}$ to the next node in the sequence $v_{m}$ or, if that does not exist, there is a directed link in the opposite direction. 
	More formally a semiwalk $\Wcal$ is the sequence 
	\begin{multline}
		\Scal= \big(v_n | v_n \in \Ncal, n\in\{0,1,2,\ldots,L\}  
		\, ,  
		\\
		\text{ and } \;\;
		(v_{m}, v_{m-1})   \in \Lcal 
		\; \text{or} \; 
		(v_{m-1}, v_{m})  \in \Lcal 
		\, , \;
		m\in\{1,2,\ldots,L\} 
		\big) 
		\, .
		\label{e:semiwalkdef}
	\end{multline}
\end{definition}
The \tsedef{length} $L$ of a semiwalk is the number of nodes in the walk minus one, $L(\Scal) = |\Scal|-1$. 

The \tsedef{links of a semiwalk} $\Lcal(\Scal)$ are defined to be the sequence of links $( \link_m )_{m=1}^{L}$ used in \eqref{e:semiwalkdef} so that each link $\ell_m$ is always a link in directed graph $\Dcal=(\Ncal,\Lcal)$. That is 
\begin{align}
	\Lcal(\Scal) =& \{ \ell_m | m=1,2,\ldots,L \} \, ,
	\nonumber
	\\
	&
	\ell_m 
	=
	\begin{cases}
		( v_{m},v_{m-1} ) & \text{if} \; ( v_{m}, v_{m-1} )     \in \Lcal \\
		( v_{m-1},v_{m} ) & \text{if} \; ( v_{m},v_{m-1} ) \not\in \Lcal \;\text{and}\;  ( v_{m-1},v_{m} )\in \Lcal
	\end{cases}
	\label{e:semiwalklinks}
\end{align}
It is important to stress that the consecutive pairs of nodes $(v_m,v_{m-1})$ ($m=1,2,\ldots,L$) of a semiwalk $\Scal=(v_n)_{n=0}^L$ are \emph{not} always links on the directed network $\Dcal$. That is not all semiwalks are walks of \eqref{e:walkdef} but walks are always a special type of semiwalk. 
For example, in \figref{f:symmnetex} there is just one walk of length three from node 1 to node 4, the path $(4,3,2,1)$. However, there while there is no walk starting from node $4$ other than the trivial walk $(4)$, there are semiwalks to all other nodes such as a semiwalk $(1,2,3,4)$ of length 3 from node $1$ to node $4$.

Since a semiwalk is a sequence of nodes, I can define concatenation of semiwalks in exactly the same way as I did with walks in \eqref{e:walkconcat}. That is if $\Tcal = (u_m)_{m=0}^{M}$ and  $\Scal = (v_n)_{n=0}^{L}$ are both semiwalks where the last node $v_L$ of $\Scal$ is the first node $u_0$ of $\Tcal$ then
\bea
\Tcal \circ \Scal
=
(u_m)_{m=0}^{M} \circ (v_n)_{n=0}^{L}
&=&
(y_r )_{r=0}^{M+L}
\quad \text{where} \in \Lcal 
\quad
y_r
=
\begin{cases}
	v_r & 0 \leq r \leq L \\
	u_{r-L} & L \leq r \leq M+L
\end{cases}
\, .
\label{e:swalkconcat}
\eea

A \tsedef{closed semiwalk} is a semiwalk  where the first and last nodes are the same. A \tsedef{semicycle} is a closed semiwalk where all the nodes in the semiwalk are distinct except for the first and last nodes which are equal to each other but to no other node in the semicycle. Again, not all semicycles are cycles but cycles are always a special type of semicycle. For example $(1,3,2,1)$ in \figref{f:symmnetex}.i is a cycle and therefore a semicycle but in the opposite order, $(1,2,3,1)$, is a semicycle but not a cycle in \figref{f:symmnetex}.i.
Also, I place no restriction on the length of a semicycle so a trivial walk $(v)$ of length zero is both a cycle and a semicycle, the only ones of length zero.

\begin{figure}[htb]
	\begin{center}
		\includegraphics[width=0.6\textwidth]{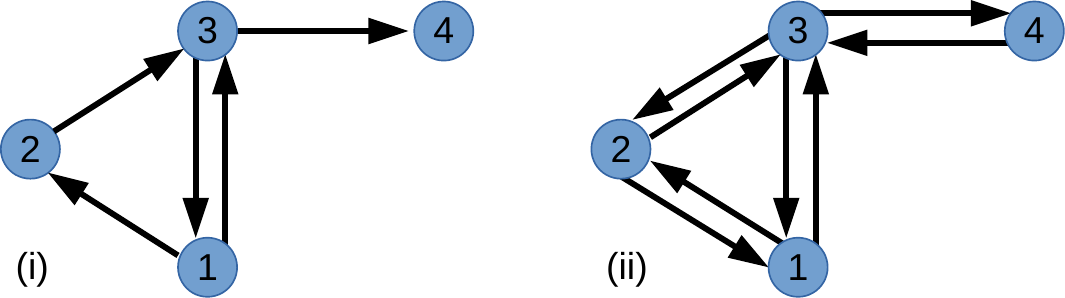}
	\end{center}
	\caption[Symmetrised network example.]{An example of (i) a directed network $\Dcal$ and (ii) its symmetrised form $\Dcalsymm$. In (i) some links, such as between nodes $1$ and $3$ are reciprocated, but most, such as between $3$ and $4$, are not. Network $\Dcal$ is not strongly connected as there are no walks from node $4$ to any other node. Rather $\Dcal$ has a single weakly-connected component, e.g.\ no walk to any other node starts from node $4$. The components of the symmeterised network $\Dcalsymm$ are always strongly connected as illustrated shown here. This strong connectivity reflects the fact that the symmetrised network are always the maximally reciprocated directed graphs so (ii) is also a directed representation of an undirected network, the undirected version of (i).}
	\label{f:symmnetex}
\end{figure}

Another way to look at semiwalks and semicycles is that they correspond to ordinary walks and cycles on the \tsedef{symmetrised directed graph} $\Dcalsymm=(\Ncal,\Lcalsymm)$ of a directed graph $\Dcal=(\Ncal,\Lcal)$, for example see p.11 and p.67 of \citet{HNC65}. The symmetrised directed graph is created from any directed graph $\Dcal$ by finding all links that are not reciprocated in $\Dcal$ and adding new links to make sure all links are reciprocated in $\Dcalsymm$. 
\begin{definition}\label{d:symm}
	The \tsedef{symmetrised directed graph} $\Dcalsymm$ of a directed graph $\Dcal=(\Ncal,\Lcal)$ 
	is defined to be the directed graph with the same nodes $\Ncal$ as the original graph $\Dcal$ but where every directed link is reciprocated.
	\begin{align}
		\Dcal=(\Ncal,\Lcal) 
		\quad    \stackrel{\mathrm{symm}}{\longmapsto} \quad
		& \Dcalsymm=(\Ncal,\Lcalsymm) 
		\nonumber
		\\
		&
		\text{where} \quad 
		\Lcalsymm
		= \big\{ (v,u ) | ( v,u )     \in \Lcal \; \text{or} \; ( u,v )     \in \Lcal \big\} 
		\, .
		\label{e:sdglinks}
	\end{align}
\end{definition}
See \figref{f:symmnetex} for an example of an unreciprocated network and its symmetrised directed graph.

By comparing the definition of allowed sequential node pairs in the semiwalk in \eqref{e:semiwalklinks} with the links in the symmetrised network \eqref{e:sdglinks}, I see that 
\begin{lemma}[Semiwalks and symmetrised networks] \label{l:swsg}
	A semiwalk $(v_n)_{n=0}^L$ on a directed network $\Dcal$ is equivalent to a walk $(v_n)_{n=0}^L$ on the symmetrised group graph $\Dcalsymm$.
\end{lemma}
This is simply because the condition for successive node pairs in a semiwalk $\Scal=(v_n)_{n=0}^L$ \eqref{e:semiwalklinks} is that there is a link between successive nodes in \emph{either} direction, that is one or both of $(v_{m},v_{m-1})$ and $(v_{m-1},v_{m})$ must be links in $\Dcal$. This is the condition \eqref{e:sdglinks} for a directed link to exist in both directions in the symmetrised version of the directed network $\Dcalsymm$ so the requirement for a walk, that $(v_{m},v_{m-1})$ always exists is guaranteed to be true in the symmetrised version $\Dcalsymm$ of the directed graph $\Dcal$.

Finally I use semiwalks to help define the connectivity of directed networks. I define a \tsedef{weakly connected component} to be a maximal subgraph $\Hcalwc=(\Ncalwc,\Lcalwc)$  where there is a \emph{semiwalk} between every pair of nodes in this subgraph. The maximal means that I cannot add any more nodes from $\Ncal$ to this subset $\Ncalwc$. The edge set is then all links $\Lcalwc \subseteq \Lcal$ containing nodes from $\Ncalwc$, i.e.\ this is an ``induced subgraph'', one induced from the node set $\Ncalwc$.
In a similar way I define a \tsedef{strongly-connected component} $\Hcalsc=(\Ncalsc,\Lcalsc)$ of a network to be one where there is always a \emph{walk} from every node to every other node in the component. Again I want the node set $\Ncalsc$ to be maximal (can not add any more nodes to this subset) and the links in $\Lcalsc$ are all those which connect nodes in this subset (an induced subgraph).

\subsection{Groups}\label{s:groups}

There are many texts on group theory (such as \citet{S98,J03}) and I will summarise the key results as a way to define my notation. 
A \tsedef{group} $\Gcal$ is a set of `elements' and a `group multiplication' law denoted by the symbol `$\times$' (formally a map $\times : \Gcal \to \Gcal$) that obey four axioms: closure ($g_1 \times g_2 \in \Gcal$), associativity ($g_1 \times (g_2 \times g_3) = (g_1 \times g_2) \times g_3$), identity ($\exists \; e \in \Gcal$ s.t.\ $e \times g = g \times e =g$) and inverse ($\exists \; (g)^{-1}\in \Gcal$ s.t.\ $g \times (g)^{-1}=(g)^{-1} \times g=e$) where these hold for any group elements $g,g_1,g_2,g_3 \in \Gcal$.  For simplicity, I sometimes write $g_1 g_2$ or $g_1 .g_2$ instead of $g_1 \times g_2$.  I will use lowercase Latin letters (such as $a$, $b$, $g$) to denote abstract elements in the group and I will always denote the identity element as $e$. 

It is important to note that elements do not always commute, that is $g_1 \times g_2 \neq g_2 \times g_1$ is not required. Where all elements in a group do commute, the group is called an \tsedef{abelian group} otherwise the group is known as a \tsedef{non-abelian group}. A \tsedef{finite group} contains a finite number of elements. 

I will call upon some basic properties of groups. For instance, it may be shown that the identity element $e$ is unique and that the inverse $g^{-1}$ of every element $g$ is also unique (the unique inverse may or may not be the same as $g$). The inverse of a product of the product of the inverse of the individual elements but in the opposite order, that is $(g_1\times g_2)^{-1} = (g_2)^{-1} \times (g_1)^{-1}$ for any group elements $g_1$ and $g_2$.

If I look at the set of group elements where they are all right multiplied by any one of the group elements, say $g_r$, then this new set of elements is just the set of group elements again, so if $\Gcal_r = \{g \times g_r | g \in \Gcal\}$ then $\Gcal_r=\Gcal$. This rearrangement theorem also applies to left multiplication by some fixed $g_r$.

Two finite groups will be used as examples in this paper. The smallest non-trivial group is the group of two elements $\{e,a\}$ which is abelian and is often denoted as $Z_2$. The group multiplication laws follow from the property of the identity element $e$ along with the rule that $a^2 \equiv a \times a = e$. The second example is the smallest non-abelian group, denoted here as $S_3$. The multiplication table is given in \tabref{t:S3}.
\begin{table}[htb]
	\begin{center}
	\begin{tabular}{c|cccccc}
	       & $e$    & $a$    & $a^2$  & $b$    & $ab$   & $a^2b$ \\ \hline
		   &        &        &        &        &        &        \\[\dimexpr-\normalbaselineskip+2pt]
	$e$    & $e$    & $a$    & $a^2$  & $b$    & $ab$   & $a^2b$ \\
	$a$    & $a$    & $a^2$  & $e$    & $ab$   & $a^2b$ & $b$    \\
	$a^2$  & $a^2$  & $e$    & $a$    & $a^2b$ & $b$    & $ab$   \\
	$b$    & $b$    & $a^2b$ & $ab$   & $e$    & $a^2$  & $a$    \\
	$ab$   & $ab$   & $b$    & $a^2b$ & $a$    & $e$    & $a^2$  \\
	$a^2b$ & $a^2b$ & $ab$   & $b$    & $a^2$  & $a$    & $e$    \\
    \end{tabular}
	\end{center}
	\caption{The Cayley table for the group $S_3$ which is the smallest non-abelian group. 
		This group can be seen in many different contexts such as the permutations of three objects or in the symmetries of an equilateral triangle. 
		The entry in a row labelled by element $g_r$ and column labelled by $g_c$ represents the result of $g_r \times g_c$, i.e.\ I left-multiply by the element labelling the row  and right multiply by the column label. 
		The group is shown in terms of three abstract elements $e,a,b$ known as \tsedef{generators} of the group as every element of the group can be expressed in terms of these three elements or their products. In the table I use the notation $ab$ to denote the product $a \times b$, $a^2 \equiv a \times a$ and so forth. It is useful to note that the elements of $S_3$ have the
		property that $a^nb=ba^{3-n}$ for $n=1,2$ which encodes the non-abelian behaviour in this group.}
	\label{t:S3}
\end{table}

\section{Group Graphs}\label{s:groupgraphs}

Group graphs are a generic theoretical construction seeking to generalise the binary labels found on signed networks to give a wider family of graphs carrying weights on their edges, ``link labels'', which come from a group. In this section I will give the basic definitions and look at some special cases that generalise the idea of balance in signed networks. I will first define these new group graphs in terms of abstract groups and then in terms of representations of groups. 
Here I will work with directed graphs looking at undirected graphs later in \secref{s:udgg}.

\subsection{Definition of a Group Graph}\label{s:groupgraphdef}

I will work with a directed graph (digraph) $\Dcal = (\Ncal,\Lcal)$.  
By working with the maximal reciprocal graph of \defref{d:symm}, undirected networks are included as a special case. The group $\Gcal$ is arbitrary and can have an infinite number of elements. However I will use examples from the two finite groups $Z_2$ and $S_3$. The first step is to define a \tsedef{group graph} in terms of group elements as labels on the links.

\begin{definition}[Group Graph] \label{d:grpgrph}
	A \tsedef{group graph} $\Ggraph=(\Dcal,\Gcal,\linkmap)$ is a directed graph $\Dcal$ where every link $\link\in\Lcal$ is labelled by a group element $g_\link = \linkmap(\link) \in \Gcal$ called a \tsedef{link label}.
\end{definition}
Rather than numerical values usually assigned as the weights of links in weighted network, in a group graph each directed link is labelled by an abstract group element. The group associated with any given group graph is implicit in the notation for the function $\linkmap$, that is formally $\linkmap$ is a map $\linkmap(\link): \Lcal \to \Gcal$ so that $\link \mapsto g_\link = \linkmap(\link)$. 
I will also use the notation $g_{vu}\equiv\linkmap((v,u))$ for the label on a link $\ell=(v,u)$ from node $u$ to node $v$. 
Some examples are shown in \figref{f:ggex}. In particular note that the $Z_2$ group graph is by definition a type of signed directed network as here the identity element of $Z_2$ labels positive edges while $Z_2$ element $a$ labels negative edges \citep{H53,HNC65,DM96,DM09}.

\begin{figure}[htb]
	\begin{center}
		\includegraphics[width=0.9\textwidth]{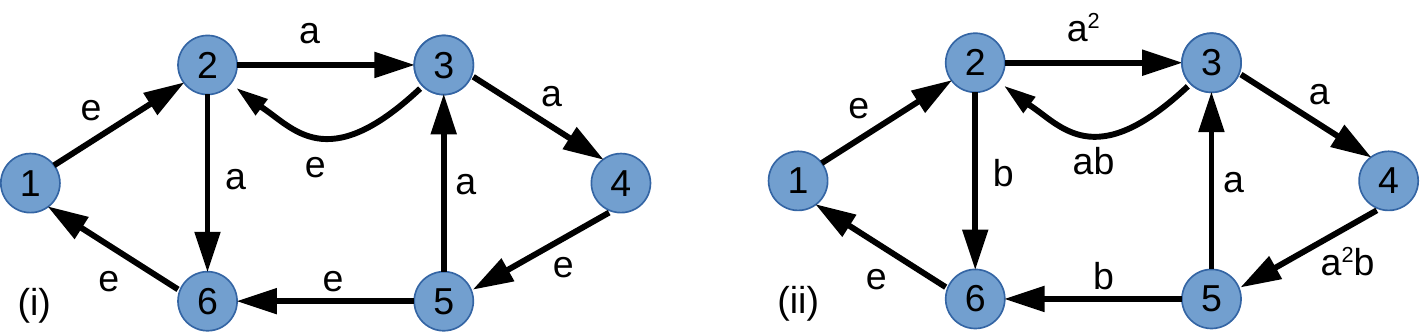}
	\end{center}
	\caption{Two examples of a group graph. In this case, both examples have the same underlying network topology, the same adjacency matrix, representing the same directed graph. The two distinct group graphs differ in the groups used to label the links. The group graph (i) on the left uses the abelian group of two elements $Z_2 = \{e,a\}$. In network (ii) on the right, the link labels come from the non-abelian group $S_3$ as described in \tabref{t:S3}. Group elements for each link are shown next to the relevant link. Both are examples of unbalanced group graphs. The example on the left is also an example of a signed network (e.g\ as defined in \citet{DM09}) where positive (negative) links are labelled by $e$ ($a$). This signed network is not balanced.}
	\label{f:ggex}
\end{figure}

\begin{definition}[The label of a walk] \label{d:walkgpelm}
	The label $\linkmap(\Wcal)$ of a walk $\Wcal=(v_n)_{n=0}^{L}$ on a group graph $\Ggraph=(\Dcal,\Gcal,\linkmap)$ is the group element given by the product of the group elements labelling the links in the walk in the order in which they appear in the walk, with the group elements from the first (last) link on the right (left). That is\footnote{Note this definition extends the domain of the function $\linkmap$ from the set of links $\Lcal$ required in the  \defref{d:grpgrph} to the set of all walks $\Wcal^\mathrm{(all)}$, which includes the set of links as the walks of length one. That is $\linkmap : \Wcal^\mathrm{(all)} \to G$, $\Wcal \mapsto g_\Wcal = \linkmap(\Wcal)$.} 
	\begin{subequations}\label{e:walklabel}
	\bea
	\linkmap(\Wcal) 
	&=& 
	\linkmap(\link_n) \times \linkmap(\link_{n-1}) \times \ldots \times \linkmap(\link_{2}) \times \linkmap(\link_1)    
	=
	g_{\link_n} g_{\link_{n-1}}  \ldots g_{\link_{2}} g_{\link_1}  \, ,  
	\\ 
	&& 
	\Wcal =(v_n)_{n=0}^{L}
	\, ,  
	\quad 
	\link_m = (v_{m},v_{m-1}) 
	\, ,
	\quad  
	m \in \{1,2,\ldots,L\}   
	\, .
	\eea
	\end{subequations}
\end{definition}
Note the order in the notation used for the elements in a walk label. That is the group element associated with the first link in the walk is placed on the far right, with the group element for each subsequent link in the walk placed to the left of group elements associated with links earlier in the walk. This order is important for non-abelian groups. While other conventions for the notation are possible, I have tried to be consistent, i.e.\ for any sequence I write the first element on the right while the last element of the sequence is on the left. For instance, my notation for links (walks of length one) is consistent with this; $(v,u)$ is the link from node $u$ to node $v$.

For example, in \figref{f:ggex} let $v_i \in \Ncal$ be the node with integer label $i$ in the figure. The walk from node $v_1$ to $v_4$ via node $v_2$ and then $v_3$ is written as the sequence $\Wcal = (v_4,v_3,v_2,v_1)$ with the first node on the right and the last node on the left. The label of this walk is $\linkmap(\Wcal)= \linkmap((v_4,v_3)) \times \linkmap((v_3,v_2)) \times \linkmap((v_2,v_1))$ and this is $\linkmap(\Wcal)= a \times a \times e = e$ for the $Z_2$ group graph (i) on the left while it is $a \times a^2 \times e = e$ for the $S_3$ example (ii) on the right. 

Finally, the definition of the label of the walk is consistent with the definition of the concatenation of two walks. That is if I have walks $\Wcal = (v_n)_{n=0}^{L}$ and  $\Xcal = (u_n)_{n=0}^{M}$ where the last node of $\Wcal$ is also the first node of $\Xcal$, i.e.\ $v_L = u_0$, then I can join these walks together to give a new walk $\Xcal \circ \Wcal$ as defined in \eqref{e:walkconcat}. What I can see is that the label $\linkmap(\Xcal \circ \Wcal)$ of the combined walk $\Xcal \circ \Wcal$ is simply the group product of the labels of the two constituent walk  in the same order (label of first part of the walk on the right, label of the last part of the walk on the right) 
\beq
 \linkmap(\Xcal \circ \Wcal) = \linkmap(\Xcal) \times \linkmap(\Wcal) \, .
 \label{e:concatwalklabel}
\eeq 
This follows from \defref{d:walkgpelm} of the label of the walk as the product of the labels of individual links (walks of length one) in the walk in the order in which the links appear in the walk.

Note that for consistency, the label $\linkmap(\Wcal_v)$ of any trivial walk $\Wcal_v$ of length zero, i.e.\ the sequence of just one node and no links, so  $\Wcal_v=(v)$, must be defined to be the identity, that is $\linkmap(\Wcal_v)=e$ for any node $v$.

\subsection{Balanced group graphs}\label{s:balance}

\subsubsection{Labels of Semiwalks}

The aim is to discuss balance on group graphs by using the properties established for labels of walks on group graphs built from the symmetrised  directed graph $\Dcalsymm$ in order to compare this to the use of semiwalks and semicycles on a group graph $\Ggraph$ built from a directed graph $\Dcal$. The main issue is how to define the label on links added when a graph is symmetrised. The idea is that whenever I add a new link to get to the symmetrised directed graph $\Dcalsymm$ I use the \emph{inverse} group element of the existing reciprocal link\footnote{The general idea is that if I move along one link from node $u$ to node $v$ then move back, from $v$ to $u$, this second step is reversing the first move.  Therefore it makes sense that in terms of the properties of the group label on reciprocated edges, I might want to look at reversing the effect of the link label on the first link, and for a group element, reversing or undoing is what the inverse of a group element does. This picture will be made more precise when looking at processes on a network in \secref{s:dynamics}.}.
The mathematical reason for this will emerge once I define balance in group graphs. 
\begin{definition}[Symmetrised Group Graph] \label{d:sgrpgrph}
	The \tsedef{symmetrised group graph} $\Ggraphsymm=(\Dcalsymm, \Gcal, \linkmap)$ of a group graph $\Ggraph=(\Dcal, \Gcal, \linkmap)$ is defined to be the group graph of the {symmetrised directed graph} $\Dcalsymm$. If a link in the original group graph $\Ggraph$ is not reciprocated, then the reciprocal link is added to $\Lcalsymm$ with label equal to the inverse of the original link. If the reciprocal of a link already exists in $\Ggraph$ then then labels are left unchanged. 
 
	That is
	\beq
	\linkmap\big((v,u)\big) =  
	\begin{cases}
		\;\; \linkmap\big((v,u)\big)            & \text{if} \; (v,u) \in \Lcal \\
		\Big[\linkmap\big((v,u)\big) \Big]^{-1} & \text{if} \; (u,v) \in \Lcal \; \text{and} \; (v,u) \not\in \Lcal \\
	\end{cases} \, .
	\eeq
\end{definition}
An example is shown in \figref{f:symmggex}.

Note that for the symmetrised group graph of a multigraph, we need to replace multiple links from the same node $u$ to the same node $v$ by a single edge.  This could only be done here if all these links had the same link label.

\begin{figure}[htb]
	\begin{center}
		\includegraphics[width=0.6\textwidth]{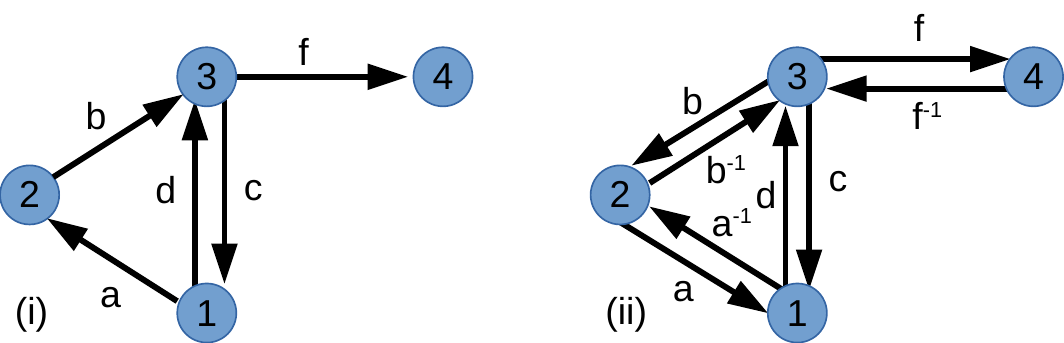}
	\end{center}
	\caption[Symmetrised group graph example.]{On the left in (i) is an example of a group graph and on the right in (ii) is the symmetrised version. Here $a,b,c,d,f$ are link labels and are elements of some group. In the original group graph (i) only the links between nodes $1$ and $3$ are reciprocated and these carry their labels into the symmetrised group graph (ii). All other links are not reciprocated so in the symmetrised group graph extra links are added to ensure every link is reciprocated. The link labels of these new links in (ii) are always the inverse group elements of the link in the original group graph (i).}
	\label{f:symmggex}
\end{figure}

\begin{definition}[Semiwalk Labels] \label{d:semiwalksgg}
	A label of a semiwalk $\Scal$ on a group graph is equal to the label of the equivalent walk on the symmetrised group graph. 
	\begin{subequations}\label{e:swalklabel}
		\bea
		\linkmap(\Scal) 
		&=& 
		\linkmap(\link_n) \times \linkmap(\link_{n-1}) \times \ldots \times \linkmap(\link_{2}) \times \linkmap(\link_1)    
		=
		g_{\link_n} g_{\link_{n-1}}  \ldots g_{\link_{2}} g_{\link_1}  \, ,  
		\\ 
		&& 
		\Scal =(v_n)_{n=0}^{L}
		\, ,  
		\quad 
		\link_m = (v_{m},v_{m-1}) \in \Lcalsymm
		\, ,
		\quad  
		m \in \{1,2,\ldots,L\}   
		\, .
		\eea
	\end{subequations}
\end{definition}
What this means is that if I look at a semiwalk $\Scal = (v_n)_{n=0}^L$ on a group graph $\Ggraph=(\Dcal, \Gcal, \linkmap)$ then the labels used in $\Lambda(\Scal)$ using \eqref{e:walklabel} come from those links $\ell_m=(v_m,v_{m-1})$ which are present in the original directed network $\Dcal$. If such a link is not in the original graph then, and only then, do I use the label derived from the reciprocated edge, i.e.\ $[\Lambda((v_{m-1},v_{m}))]^{-1}$, which has been added in to the link $\Lcalsymm$ in the symmetrised group graph. So if a link is reciprocated in the original group graph, there is no guarantee that the link labels of reciprocated links are inverses of each other. For instance, in \figref{f:ggex} the reciprocated links between nodes $2$ and $3$ do not satisfy this condition $\Lambda((v_3,v_2)) \neq [\Lambda((v_2,v_3))]^{-1}$ in either example network. 

I will often encounter semiwalks and semicycles whose label is the identity element. These may be referred to as \tsedef{identity semiwalks} and \tsedef{identity cycles}.

A useful concept is the reverse semiwalk.
\begin{definition}[Reversed semiwalks] \label{d:rsemiwalk}
	A reversed semiwalk $\Scalrev$ of a semiwalk $\Scal=(v_n)_{n=0}^L$ on a group graph $\Ggraph=(\Dcal, \Gcal, \linkmap)$ is the semiwalk defined by the sequence of nodes in reverse order, that is $\Scalrev= (v_{L-m})_{n=0}^L$ .
\end{definition}
A key point of difference between walk and semiwalks is that walks on a group graph can not always be reversed because a walk can pass through unreciprocated links. However, semiwalks may pass along unreciprocated links in either direction, so a reverse semiwalk always exists on a group graph if a semiwalk exists. More formally
\begin{lemma}[Existence of reversed semiwalks] \label{l:rsemiwalk}
	A reversed semiwalk $\Scalrev= (v_{L-m})_{n=0}^L$ of a semiwalk $\Scal=(v_n)_{n=0}^L$ on a group graph $\Ggraph=(\Dcal, \Gcal, \linkmap)$ always exists.
\end{lemma}

\subsubsection{Balance for group graphs}

I can now define balance for a group graph using semiwalks imitating the approach used for signed networks. Once defined, I will then look at the properties of balanced group graphs, aiming for \thref{t:propbgg} below which summarises the key properties.

\begin{definition}[Balanced group graph] \label{d:bggsw}
	A group graph is \tsedef{balanced} when the label of any semiwalk only depends on the initial and final nodes in the walk, i.e.\ it is independent of the route taken. That is
	\beq 
	  g_{vu} = \linkmap(\Scal(v,u)) = \linkmap(\Scal'(v,u)) \in \Gcal
	\eeq
	for any semiwalks $\Scal(v,u)$ or $\Scal'(v,u)$ from node $u$ to node $v$.
\end{definition}
Two examples of balanced group graphs are shown in \figref{f:bggex}.
The notation $\Scal(v,u)$ indicates a semiwalk that starts from node $u$ and ends at node $v$ but it can be one of many such semiwalks that may exist. This type of path independence is often a key property or an important requirement in many other physical problems. In these cases the ``effort'' of moving between two points is the same whatever path is chosen and is often encoded by the use of conservative fields and potentials as in electromagnetism. Balance in group graphs is reproducing this property but in terms of processes that are moving through a group as I will show more precisely in \secref{s:dynamics}.

\begin{figure}[htb]
	\begin{center}
		\includegraphics[width=0.9\textwidth]{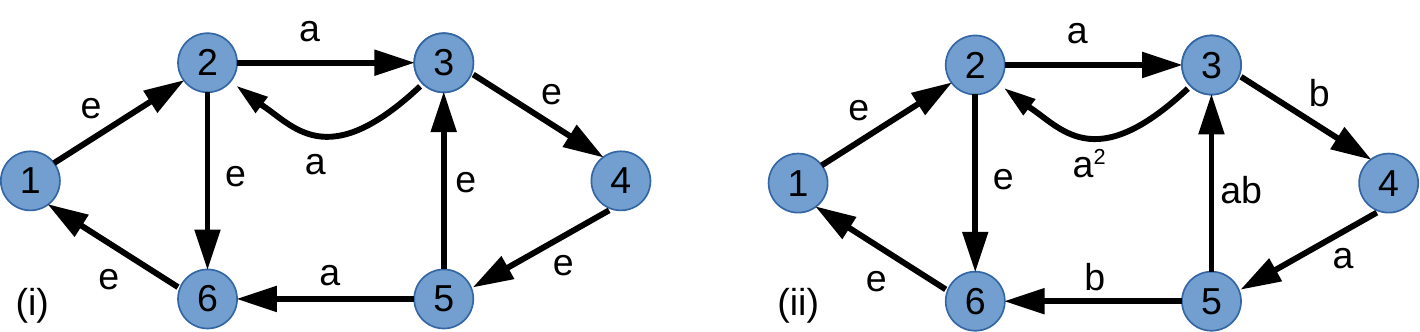}
	\end{center}
	\caption{Two examples of balanced group graphs. These use the same topology as in \figref{f:ggex} but different link labels. The group graph (i) on the left has links lying in $Z_2$ while in (ii) on the right links lie in $S_3$. The numbers inside the blue circles are used to identify nodes in examples in the text. Group elements for each link are shown next to the relevant link. The example on the left is a balanced signed network where positive (negative) links are labelled with $e$ ($a$).}
	\label{f:bggex}
\end{figure}

Balanced group graphs have a number of useful properties.
\begin{lemma}[Balance Symmetrised Group Graphs] \label{l:bsgrpgrphsw}
	A group graph is \tsedef{balanced} iff the symmetrised group graph is balanced. 
\end{lemma}
This follows from \defref{d:semiwalksgg} which defines the label of every semiwalk on the original group graph $\Ggraph$ to be the same of the equivalent a walk, so also a semiwalk, on the symmetrised group graph $\Ggraphsymm$. Since balance is defined in terms of these semiwalk labels and these are the same for the same semiwalk on $\Ggraph$ and $\Ggraphsymm$, the result for balance in \lemref{l:bsgrpgrphsw} follows.

Note the issue with the lack of definition of a symmetrised group graph in a multigraph, when there are multiple links in the same direction between the same pair of nodes but with different link labels, means that \lemref{l:bsgrpgrphsw} still works for multigraphs as it is implicit that the lemma requires the existence of a symmetrised group graph to hold.

\begin{lemma}[Closed walks in balanced group graphs] \label{l:cswbggid}
	Any closed semiwalk $\Ccal$ in a balanced graph has a label equal to the identity element. 
	\beq
	\linkmap(\Ccal) =e \, .
	\label{e:cwalkbggid}
	\eeq
\end{lemma}
This follows because for any closed semiwalk $\Ccal$ I may create longer semiwalks by going around that closed semiwalk $\Ccal$ any number of times, i.e.\ follow the closed semiwalk $\Ccal' = \Ccal \circ \Ccal \circ \ldots  \circ \Ccal$. If this closed semiwalk did not have a label equal to the identity, this would immediately generate different labels depending on the number of times the semiwalk $\Ccal'$ went round the closed semiwalk $\Ccal$. Since $\Ccal'$ is also a semiwalk between the same nodes as $\Ccal$, in a balanced group graph these two closed semiwalks $\Ccal'$ and $\Ccal$ must have the same label. The only way to satisfy all these constraints is if $\linkmap(\Ccal) =e$. More formally I have the following.
\begin{proof}
	Consider a non-trivial closed semiwalk $\Ccal = (v, \ldots, v)$, a semiwalk of length one or more starting and ending at some node $v$. I can concatenate this to give a \emph{distinct} semiwalk $\Ccal' = \Ccal\circ \Ccal$ by \eqref{e:walkconcat}. This will also start and end at node $v$. Using \eqref{e:concatwalklabel} then gives for the labels of the semiwalks that $\linkmap(\Ccal') = \linkmap( \Ccal) \times \linkmap( \Ccal)$. Since the group graph is balanced, and the semiwalks $\Ccal$ and $\Ccal'$ are between the same nodes, they must by definition have the same group label in a balanced group graph so that $\linkmap( \Ccal) =\linkmap(\Ccal')$. Given that all group elements have an inverse, I may left multiply by $(\linkmap( \Ccal))^{-1}$ to give that $\linkmap(\Ccal) =e$ as required.
	\\
	The trivial walk $\Wcal_v=(v)$ of length zero from any node $v$ is also a closed semiwalk. I have that $\Wcal_v \circ \Wcal_v = \Wcal_v$ so that this concatenation does not give a distinct closed semiwalk. However, from the definition of the label of a walk in \eqref{e:concatwalklabel}, I have that $\Lambda(\Wcal) \times \Lambda(\Wcal_v) = \Lambda(\Wcal)$. Left multiplying by $(\Lambda(\Wcal))^{-1}$ gives $\linkmap(\Wcal_v)=e$ for this remaining case.
\end{proof}


There are some other special cases of \lemref{l:cswbggid} worth noting. First, a semicycle, as defined in \eqref{e:cycledef}, is a special case of a closed semiwalk so the label of all semicycles in a balanced group graph is always the identity. The simplest non-trivial example of a cycle and semicycle is a self-loop, a link $(v,v)$ starting and ending at the same node $v$. The application of \lemref{l:cswbggid} then shows that the label of an link in a self-loop in a balanced group graph must be the identity, $\linkmap((v,v)) =e$. 

The next simplest closed semiwalks to consider are those of length two. These can be a walk which goes twice around a self-loop or these are semiwalks which visit a distinct neighbouring node and then return.  In either case, these closed semiwalks have the form $(v,u,v)$ where at least one of $(u,v)$ or $(v,u)$ must be links in the original directed graph $\Dcal$, $(u,v)\in \Lcal$ or $(v,u) \in \Lcal$ (here $u=v$ is allowed but not necessary). Again, the first and last nodes are the same so this semiwalk is a closed semiwalk. The label of this closed semiwalk in a balanced group graph is the identity by \lemref{l:cswbggid} so $\linkmap(v,u) \times \linkmap(u,v)=e$ and thus $\linkmap(v,u) = [\linkmap(u,v)]^{-1}$. That is the label of any connected node pairs in the symmetrised graph must be the inverse of each other in a balanced group graph. As I use this result later, I will summarise it as follows.
\begin{corollary}[Labels of reciprocated node pairs in balanced group graphs] \label{c:reciplabel}
	The label of any linked node pairs in a balanced group graph must be the inverse of each other. 
	\beq
	\linkmap((u,v)) = \big(\linkmap((u,v))\big)^{-1} \quad \text{if} \quad (u,v) \; \text{or} \;  (v,u) \in \Lcal \, .
	\eeq
\end{corollary}
For example, in the two examples in \figref{f:ggex}, the reciprocated links between nodes $2$ and $3$ are not labelled by a group element and its inverse so these group graphs cannot be balanced. On the other hand, the same pair of reciprocated links in the balanced groups graphs of \figref{f:bggex} do satisfy this criteria. 

The same reasoning means that multiedges in the same direction in a balanced group graph must all have the same label.

Note that equivalent work on groups graphs under the group labelling name \cite{ES79}, voltage graph name \citet{GA73}, the gain graph name \citet{Z89}  or for the topic of graph synchronisation always demand that reciprocated edges carry link labels which are inverses of one 
another\footnote{For instance, for voltage graphs this has been required since the first paper, see condition C1 on page 943 of \citet{GA73}. For a group labelled graph see equation (2) of \cite{AF19}. For gain graphs this condition is in the first paragraph of section 5 in \citet{Z89}. The condition appears not to have been relaxed as shown (at the time of writing) in the following informal summaries: \href{https://en.wikipedia.org/wiki/Voltage_graph}{voltage graph article on Wikipedia}, 
the \href{https://mathworld.wolfram.com/VoltageGraph.html}{voltage graph article on Mathworld},
and the \href{https://en.wikipedia.org/wiki/Gain_graph}{Wikipedia article on gain graphs}.}.

While necessary for a balanced group graph, this is neither a necessary restriction for a general group graph nor a desirable one. For instance, in social networks, it is quite possible for a pair of people to report links with each other in some survey but they could have different view about the positive or negative nature of that association. Here, the link labels on a group graph are unconstrained.

Before I state and prove the properties of balanced group graphs I need to establish one property of reverse semiwalks in a balanced group graph. 
\begin{lemma}[Reversed semiwalks label on balanced group graphs] \label{l:rsemiwalkbgg}
	The label $\linkmap(\Scalrev)$ of a reversed semiwalk $\Scalrev$ of a semiwalk $\Scal$ in a balanced group graph is the inverse of the label on the semiwalk $\Scal$, that is  $\linkmap(\Scalrev) = [\linkmap(\Scal)]^{-1}$.
\end{lemma}
This follows since the group elements of a semiwalk come from the group multiplication of the link labels of links in the corresponding walk on the symmetrised group graph, \defref{d:semiwalksgg}, in the order in which the links appear in that walk. The group inverse of this product is simply the link label of the reciprocated link in the symmetrised group graph but with the order of the (reciprocated) links in the reverse order since $(g_1\times g_2)^{-1} = (g_2)^{-1} \times (g_1)^{-1}$ for any group elements $g_1$ and $g_2$. The reciprocated links are guaranteed to carry the inverse label of the original link only because the symmetrised group graph is balanced so \corref{c:reciplabel} applies.
\begin{proof}
	Consider a semiwalk $\Scal=(v_n)_{n=0}^L$ which corresponds to a sequence of directed links $\Lcal(\Scal) = (\ell_m)_{m=1}^L$ where $\ell_m=(v_m, v_{m-1}) \in\Lcalsymm$ on the symmetrised group graph (these may or may be links on the original group graph). By \defref{d:semiwalksgg} I have that the label for this semiwalk is 
	\bea
	\linkmap(\Scal) 
	&=& 
	\linkmap(\link_L) \times \linkmap(\link_{L-1}) \times \ldots \times \linkmap(\link_{2}) \times \linkmap(\link_1)    
	\, .
	\label{e:swalklabel2}
	\eea
	In the same way the reverse semiwalk $\Scalrev=(v_{L-n})_{n=0}^L$  has label
	\bea
	\linkmap(\Scal) 
	&=& 
	\linkmap(\linkrev_1) \times \linkmap(\linkrev_{2}) 
	\times \ldots \times 
	\linkmap(\linkrev_{L-1}) \times \linkmap(\linkrev_{L-1})    
	\label{e:swalkrevlabel2}
	\eea 
	where $\linkrev_m=(v_{m-1}, v_{m}) \in\Lcalsymm$ is the reciprocated link of $\ell_m$ in the symmetric group graph. Since the symmetrised group graph is balanced by \lemref{l:bsgrpgrphsw} I have that $\linkmap(\linkrev_1) =[\linkmap(\link_1)]^{-1}$ by \corref{c:reciplabel}.  Thus, by using that $(g_2.g_1)^{-1}=(g_1)^{-1}.(g_2)^{-1}$ I find that
	\bea
	\linkmap(\Scal) 
	&=& 
	[\linkmap(\link_1)]^{-1} \times [\linkmap(\link_{2})]^{-1} 
	\times \ldots \times 
	[\linkmap(\link_{L-1})]^{-1} \times [\linkmap(\link_{L-1}) ]^{-1}
	\label{e:swalkrevlabel3}
	= [\linkmap(\Scal) ]^{-1} \, .
	\eea 	
\end{proof}

\subsection{Node Labels}\label{s:nodelabels}

So far I defined group graphs by assigning a group element to each link. When a group graph is balanced, it turns out there is also a natural way to label the nodes with a group element. 
Deducing the node labels given the link labels is the main goal when working with the group synchroisation approach to group graphs \citep{KEES03,GK06a,S11e,AMMSS18,AF19,S21e}.

The result for node labelling can be summarized by the following lemma. 
\begin{lemma}[Node labels in a balanced group graphs] \label{l:nodelabel}
	For a balanced group graph, the nodes may be labelled by group elements in a way that is consistent with the link labels. That is I may assign $g_v\in \Gcal$ for all $v \in \Ncal$ s.t.\ $g_v = \linkmap(\Scal(v,u)) g_u$ for any semiwalk $\Scal(v,u)$ from node $u$ to node $v$. 
\end{lemma}
One way to prove this lemma is by giving an explicit method to construct these node labels.
\begin{proof}
	I can treat each weakly-connected component independently. The algorithm given here is simply repeated for each weakly connected component. So without loss of generality I assume our network $\Dcal$ has a single weakly connected component. 
	
	Pick any one source node $r \in \Ncal^\mathrm{(source)} \subseteq \Ncal$ as the root node. This root node is assigned the group element $g_r \in \Gcal$ as the node label for $r$. This can be any element of the group.
	
	Now assign the node labels to be the group elements $g_v = g(\Scal(v,r)) g_r$ where $\Scal(v,r)$ is any semiwalk from $r$ to $v$. Such a semiwalk always exists  for all nodes $v$ by definition of a weakly-connected component and a semiwalk. The group label of any semiwalk depends only on the initial and final nodes in a balanced group graph.  Thus $g_v$ is fixed and unique given $g_r$.
\end{proof}

The choice of the root node is arbitrary. To illustrate this, let me choose an alternative root node $s$. I choose the node label for this second root node $s$ to be $g_s = \linkmap(\Scal(s,r)) g_r$ where $\Scal(s,r)$ is any semiwalk from $r$ to $s$ and at least one will always exist in a weakly connected component. Then the label for any node $v$ is $g_v = \linkmap(\Scal(v,s)) g_s$ when constructed from root node $s$ but this is identical to the label assigned by using root node $r$ since $g_v = \linkmap(\Scal(v,s)) \linkmap(\Scal(s,r)) g_r = \linkmap(\Scal(v,s) \circ \Scal(s,r)) g_r = \linkmap(\Scal(v,r)) g_r$. That is, I end up with the same set of node labels whichever source or sink node I choose as the root node.

A corollary of \lemref{l:nodelabel} is that all link labels may be written in terms of the node labels of the nodes in that link. 
\begin{corollary}[Link labels from node labels in a balanced group graph] \label{c:lldecomp}
	For a balanced group graph, the label of a link $(v,u)$ is given by $g_v.(g_u)^{-1}$ where $g_v$ and $g_u$ are the node labels of nodes $v$ and $u$.
	\beq
	\linkmap\big((v,u)\big) = g_v.(g_u)^{-1} \;\; \forall \; (v,u) \in \Lcal \, .
	\label{e:lldecomp}
	\eeq
\end{corollary}
This follows from the way the node labels were constructed in the proof of \lemref{l:nodelabel} as a link is simply a semiwalk of length one.

Another corollary of \lemref{l:nodelabel} is that for each weakly connected component of a balanced group graph, there are many sets of distinct node labels $\{g_v\}$ as there are elements in the group. This follows since I can choose any group element for any one root node which then fixes the node labels\footnote{More formally, consider the example where the root node $r$ is a source node. I left multiply all the node labels by some fixed element $g$. I know that the original node labels are $g_v = \linkmap(\Wcal(v,r)) \times g_r$ for any node $v$ so now I have new node labels $g^\prime_v = g_v \times g = \linkmap(\Wcal(v,r)) \times g_r \times g $ that is $g^\prime_v = \linkmap(\Wcal(v,r)) \times g^\prime_r$. So by taking $g$ through every elements of the group, each value of $g$ generates a unique set of node labels by the rearrangement theorem. Thus I have exactly $|G|$ versions of the node labels, one for each element in the group.} for all remaining nodes. 
However, for any given set of node labels $\{g_v\}$, not every element of the group need appear as shown in the $S_3$ example of \figref{f:bggexVL}. 

\begin{figure}[htb]
	\begin{center}
		\includegraphics[width=0.9\textwidth]{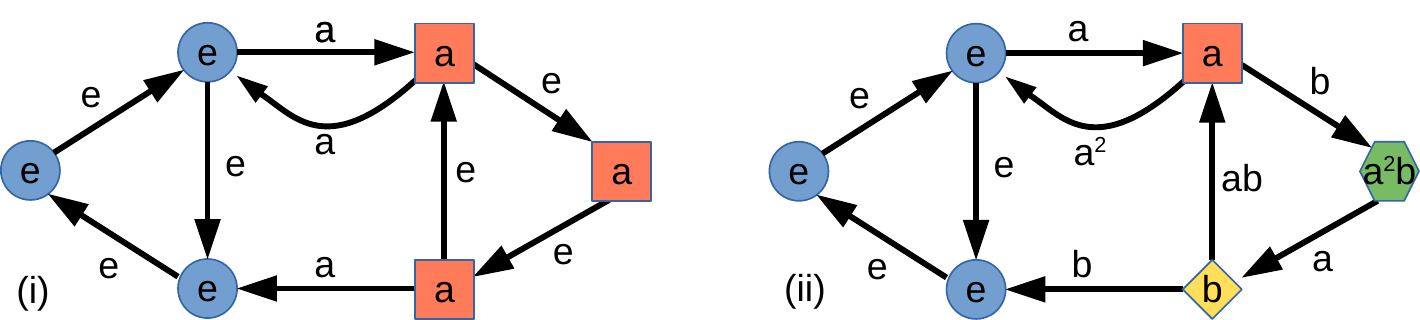}
	\end{center}
	\caption{The two balanced group graphs of \figref{f:bggex} redrawn to show the node labels $g_v \in \Gcal$ inside the node symbols. The node shape and colour indicates membership of one particular block of the node partition which respects the group structure. Note that for the $S_3$ graph on the right, there are only four blocks in the partition, $\{\Bcal_e,\Bcal_a,\Bcal_b,\Bcal_{a^2b}\}$, only four different group elements appear as node labels out of the possible six elements of $S_3$. 
	}
	\label{f:bggexVL}
\end{figure}

\begin{figure}[htb]
	\begin{center}
		\includegraphics[width=0.9\textwidth]{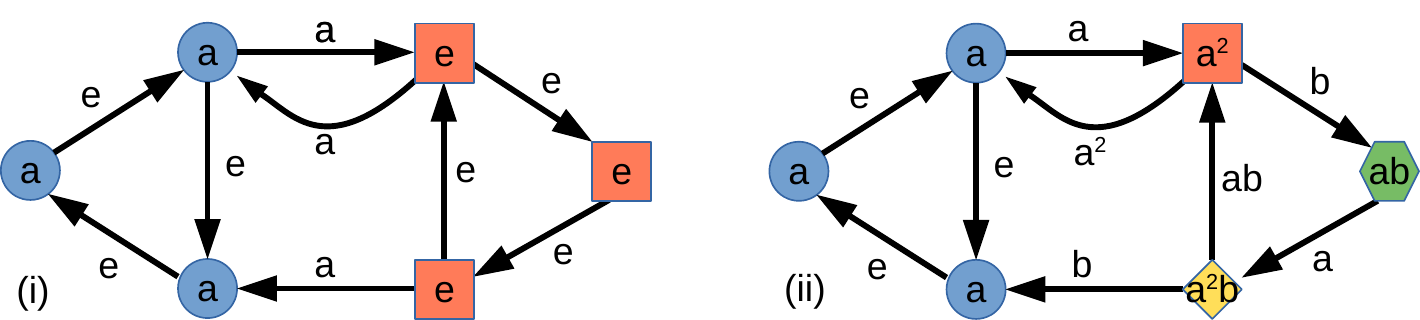}
	\end{center}
	\caption{An equivalent set of node labels for the balanced group graphs of \figref{f:bggexVL}. The node labels $g_v \in \Gcal$ inside the node symbols. The node shape and colour indicates membership of one particular block of the node partition which respects the group structure and this block structure is unchanged. 
	}
	\label{f:ggcnstntexVLrl}
\end{figure}

These node labels derived from the link labels\footnote{In fact any arbitrary assignment of a group element to each node defines a partition of the nodes. However, the node labels derived from the link labels in a balanced group graph are not simply any permutation of the group elements (of which there are $|\Gcal|!$) but one of $|\Gcal|$ specific rearrangement of the group elements. So node labels derived from link labels reflect a very specific structure in a balanced group graph.} give a partition $\Pcal$ of the nodes in any group graph. That is all nodes with the same label $g$ are placed into one block $B_g$ so that
\beq 
\Pcal = \{ \Bcal_g \, | \, g \in \Gcal, \; \Bcal_g \neq \emptyset \} \, , \quad \Bcal_g = \{ v \, | \, v_g = g \} 
\, .
\label{e:nodelabelpartition}
\eeq

What is perhaps more interesting are the following properties for \tsedef{identity semiwalks}, i.e.\ semiwalks whose label is the identity element.
\begin{lemma}[Nodes with the same node label linked by identity semiwalks.]\label{l:nltoidsw}
	If two nodes $u,v$ in a balanced group graph are in the same block $\Bcal_g$ of the node partition $\Pcal$ then the label of any semiwalk $\Scal(v,u)$ from $u$ to $v$ is the identity, $g_{vu}=\linkmap(\Scal(v,u))=e$.
\end{lemma}
This follows from \lemref{l:nodelabel} where $g_v = \linkmap(\Scal(v,u)) g_u$ but in this case I have that $g_v = g_u$ so I then see that $\linkmap(\Scal(v,u)) = e$. Note that this also works for a trivial walk of length zero where $u=v$. 

The converse is also true.
\begin{lemma}[Identity semiwalks link nodes in the same block]\label{l:indentitywalk}
	If a semiwalk $\Scal(v,u)$ runs from node $u$ to node $v$ in a balanced group graph has a label $\linkmap(\Scal(v,u))$ equal to the identity element $e$  then the two nodes are in the same block $\Bcal_g$ where $g=g_u=g_v$.
\end{lemma}
Again from \lemref{l:nodelabel} I have that $g_v = \linkmap(\Scal(v,u)) g_u$ but now $\linkmap(\Scal(v,u))=e$ and so I deduce that $g_v = g_u$. From the definition of the blocks $\Bcal_g$ of the node-label partition in \eqref{e:nodelabelpartition}, I know that these two nodes lie in the same block. 

As a special case, I can just consider node pairs connected by one directed link or a pair of reciprocated links where the link labels are the identity element. There is therefore a semiwalk of length one between the pair of connected nodes in both directions. Such nodes in a balanced group graph will only connect nodes in the same block making these blocks easy to identify. So I can consider an \tsedef{identity group graph} $\Ggraph_e$ which is the original group graph $\Ggraph$ where only identity links are retained\footnote{That is the link set $\Lcal_e$ for the {identity group graph} $\Ggraph_e = (\Gcal_e=(\Ncal,\Lcal_e),G,\linkmap)$ of any group graph $\Ggraph = (\Gcal,G,\linkmap)$ may be defined to be $\Lcal_e = \{ \ell | \ell \in \Lcal, \linkmap(\ell)=e \}$. The identity group graph $\Ggraph_e$ may be defined for balanced and unbalanced group graphs and usually more than one group graph is mapped to any one identity group graph.} with all other links removed. 

The nodes $\Acal_i$ in the $i$-th weakly connected component of the resulting identity group graph $\Ggraph_e$  would form a natural partition of the nodes of a group graph.   
In some cases these components $\{\Acal_i\}$ form a sub-partition of the original node label partition, i.e.\ it is possible that the are two distinct nodes $u$ and $v$ with the same label $g=g_u=g_v$ but which are in different components of the identity group graph, $g_u \in \Acal_i$ and $g_v \in \Acal_j$ for some $i \neq j$. This is a generalisation of the idea that some signed networks are ``clusterable'' even if they are not balanced \citep{H53,D67,KCN19}. This is because nodes can be in the same partition $\Bcal_g$ if they have a identity walk between them but the links in that walk need not be labelled by the identity. So there need not be any walk between nodes in each block $\Bcal_g$ that lies entirely on links labelled by the identity, i.e.\ nodes in one block may not be connected on the identity group graph. As components $\Acal_i$ are defined by semiwalks on identity links only, components $\Acal_i$ may correspond to smaller blocks of a partition. Put more formally $\Acal_i \subseteq \Bcal_g \subseteq \Ncal$ where $g=g_v$ for any $v \in \Acal_i$, or
\beq
\Bcal_g = \bigcup_{i \in I_g} \Ccal_i \, , \quad
I_g = \{i | g_v = g, \, v \in \Ccal_i\} \, .
\eeq 
A situation where the number of (non-empty) components $\Acal_i$ is greater than the number of non-empty blocks is shown in \figref{f:ggbexBVL}. In \figref{f:ggbexBVL}.i on the left, the node shape and colour indicates the three node blocks present: $\Bcal_e = \{v_1,v_2,v_4,v_6\}$ containing nodes labelled by $e$ (blue circles), $\Bcal_a = \{v_3\}$ containing nodes labelled by $a$ (red square) and $\Bcal_b = \{v_5\}$ containing nodes labelled by $b$ (yellow diamond). One of blocks $\Bcal_e$, is not weakly connected by any walk along identity links, links whose label is $e$. So block $\Bcal_e$ can be split into two smaller components, $\Acal_{126}=\{v_1,v_2,v_6\}$ and $\Acal_{4}=\{v_4\}$ (shown in \figref{f:ggbexBVL}.ii on the right as blue circles and green hexagons respectively) defined by nodes which are weakly connected only by walks on identity links, i.e.\ walks on the identity group graph $\Ggraph_e$ shown in \figref{f:ggbexBVL}.ii. So here I have $\Bcal_e = \Acal_{126} \cup \Acal_4$, $\Bcal_a = \Acal_3$ and $\Bcal_b = \Acal_5$. Then $\{\Acal_{126},\Acal_3,\Acal_4,\Acal_5\}$ forms a sub-partition of the node partition $\{\Bcal_e,\Bcal_a,\Bcal_b\}$. This is an example of a balanced group graph that is clusterable. Note that if the original group graph was not balanced it can still be clusterable. For instance, changing the label $ab$ on the link from node $5$ to node $4$ in \figref{f:ggbexBVL}.i to any other group element creates an unbalanced group graph but one that is still clusterable as it has the same identity group graph $\Ggraph_e$ as the group graph on the right. 
\begin{figure}[htb]
	\begin{center}
		\includegraphics[width=0.9\textwidth]{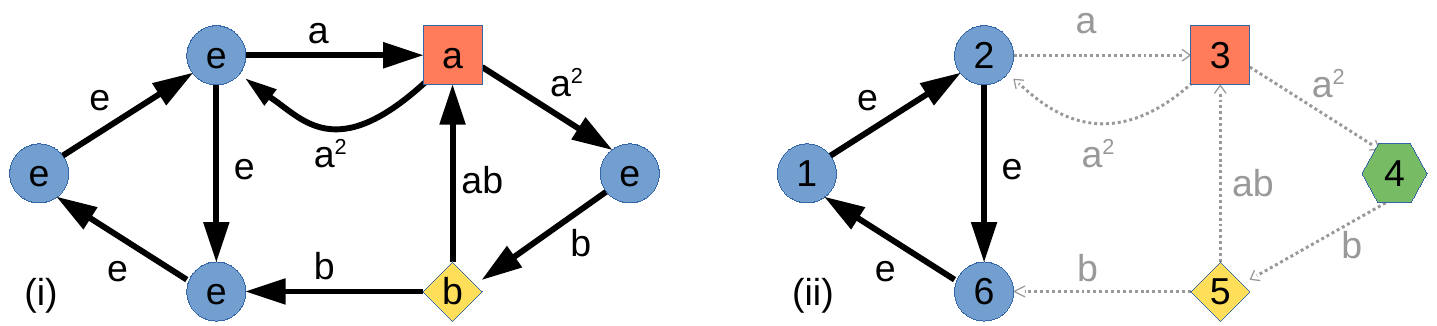}
	\end{center}
	\caption{On the left in (i) a balanced $S_3$ group graph is shown while in (ii) on the right, the corresponding identity group graph $\Ggraph_e$ is shown, that is the same group graph but where only links labelled by the identity $e$ are retained.  Links are shown as black solid arrows, thin grey dashed lines represent links that have been removed. In (i) on the left the labels inside each node symbol give the node label $g_v \in \Gcal$ while in (ii) on the right the labels inside each node symbol give the index $i$ f each node $v_i$. 	
		The node shape and colour indicates membership of one particular block of a node partition which respects the group structure. 
	}
	\label{f:ggbexBVL}
\end{figure}
A rather trivial corollary is that I can always create the identity group graph $\Ggraph_e$ from any group graph $\Ggraph$ whether it is balanced or not. The identity group graph is, trivially, a balanced group graph.

Finally, for a balanced group graph I can define the link labels of a balanced group graph from any given set of node labels, something which contains the converse of \lemref{l:nodelabel}.
\begin{lemma}[From node labels to balanced group graph]\label{l:nltocgg}
	Given a set of node labels, $\{g_v\}$, for a directed network $\Dcal=(\Ncal,\Lcal)$, where $g_v \in \Gcal$ for some group $\Gcal$ and $v \in \Ncal$,  a balanced group graph $\Ggraph=(\Dcal,\Gcal,\linkmap)$ can be defined by assigning the link labels to be $\linkmap((v,u)) = g_u \times (g_v)^{-1}$ for all links $(v,u) \in \Lcal$. 
\end{lemma}
\begin{proof}
	Suppose I have a semiwalk $\Scal = (v_n)_{n=0}^L$ from node $v_0$ to node $v_L$. The label of the walk is then given by
	\bea
	\linkmap(\Scal) 
	&=&
	\linkmap(v_L,v_{L-1}) \times \linkmap(v_{L-1},v_{L-2}) \times \ldots \times \linkmap(v_{2},v_{1}) \times \linkmap(v_{1},v_{0})    
	\\ 
	&=&
	g(v_L) g(v_{L-1})^{-1} . g(v_{L-1}) g(v_{L-2})^-1 
	\ldots 
	g(v_{2})g(v_{1})^{-1} . g(v_{1})g(v_{0})^{-1}    
	\\ 
	&=&
	g(v_L) g(v_{0})^{-1}    
	\eea
	where $g(v) \equiv g_v$. Therefore the label of any semiwalk only depends on the node labels of the source and target nodes. Hence the label of any semiwalk is the same whatever semiwalk is taken between two nodes. Hence these link labels $\linkmap((v,u)) =  g_v(g_u)^{-1}$ do define a balanced group graph by \defref{d:bggsw}.
\end{proof}
Note this contains the converse of \lemref{l:nodelabel}, namely that nodes can be labelled in a balanced group graph. However, this also shows that \emph{any} set of node labels can be used to define link labels of the form $\linkmap((v,u)) = g_v(g_u)^{-1}$ of \lemref{l:nltocgg} that then gives an associated balanced group graph. So this gives me a simple way to construct balanced group graphs. 

\begin{figure}[htb]
	\begin{center}
		\includegraphics[width=0.9\textwidth]{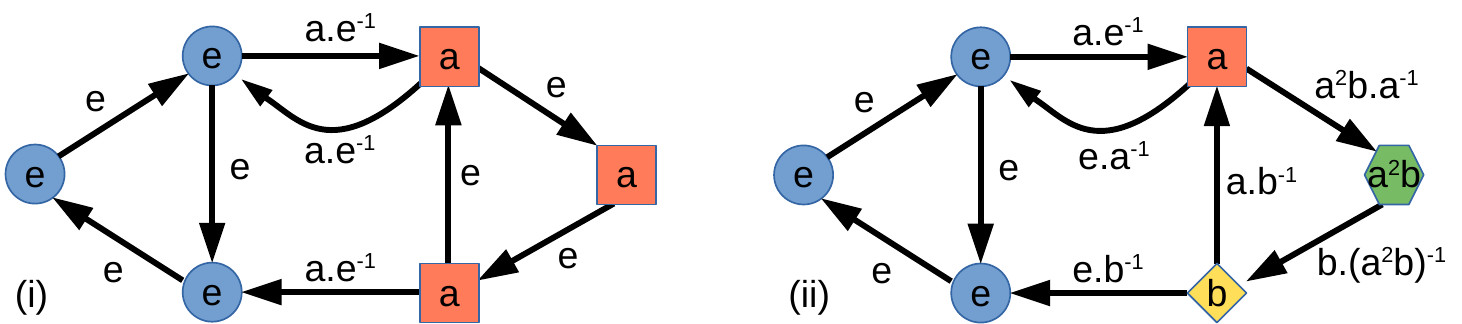}
	\end{center}
	\caption{The two balanced group graphs of \figref{f:bggex}. These have been redrawn to show the node labels $g_v \in \Gcal$ inside the node symbols and the link labels rewritten in the form $g_v(g_u)^{-1}$ for links from node $u$ to node $v$. The node shape and colour indicates membership of one particular block of the node partition which respects the group structure. Note that for the $S_3$ graph in (ii) on the right, there are only four blocks in the partition, only four different group elements appear as node labels out of the possible six elements of $S_3$. 
	}
	\label{f:ggcnstntexVL2}
\end{figure}

\subsubsection{Summary of balanced group graph properties}\label{s:summarybgg}

I have worked through the definition and properties of group graphs in one particular order, starting with one particular definition, \defref{d:bggsw}. However, for many of properties the converse is also true and so there are other starting points. So, motivated by the properties of balanced signed networks, I will now show the following theorem to summarise some of the previous results in a more concise manner. This is a direct translation of Theorem 13.2 on p 342 in \citet{HNC65} for balanced signed networks.
\begin{theorem}[Properties of balanced group graphs]\label{t:propbgg}
	The following statements are equivalent for any balanced group graph $\Ggraph=(\Dcal, \Gcal, \linkmap)$.
	\begin{enumerate}
		\item \label{i:swalk} For every pair of nodes, all semiwalks joining them have the same label.
		\item \label{i:node} The set of nodes $\Ncal$ of a group graph $\Ggraph$ can be partitioned into $|\Dcal|$ blocks (some of which may be empty) such that every node in a block $\Ncal_g$ carries a label $g \in \Gcal$ and every link $\ell$ from a node with label $g$ to a node with label $g'$ has a label  $\linkmap(\ell)=(g')^{-1}g$. 
		\item \label{i:cswalk} Every closed semiwalk $\Ccal$ of $\Ggraph$  has an identity label, $\linkmap(\Ccal)=e$.
	\end{enumerate}
\end{theorem}
Put another way, any of these three statements can be used as a definition of a balanced group graph. 
Note that we can restate \thref{t:propbgg}.\ref{i:node} in the language of \tsedef{switching transformations} \citep{AR68,Z82,Z89,CW86,CDD21} as noted in \corref{c:balswitch} of \appref{as:switching}.
One way to see that these statements are all equivalent and to relate them to the properties derived above is as follows.
\begin{proof}
	I have already shown that \thref{t:propbgg}.\ref{i:swalk} implies\footnote{That \thref{t:propbgg}.\ref{i:node} follows from \thref{t:propbgg}.\ref{i:swalk} is implicit in \lemref{l:nodelabel} which I derived from \defref{d:bggsw} in the earlier text.} \thref{t:propbgg}.\ref{i:node}. This is because  \thref{t:propbgg}.\ref{i:swalk} is equivalent to \lemref{l:nodelabel} which I derived starting from \defref{d:bggsw} which is equivalent to \thref{t:propbgg}.\ref{i:node}. 
	
	To show that \thref{t:propbgg}.\ref{i:node} implies \thref{t:propbgg}.\ref{i:cswalk} I note that I can pick any node in the closed semiwalk to write the closed semiwalk as $\Scal=(v_n)_{n=0}^L$ where $v_L=v_0$ is the first node and last node in a (closed) semiwalk.  Then I can use \lemref{l:cswbggid} which states the label of this closed semiwalk is $g(v_L) \times[g(v_0)]^{-1}$ where $g(v) \equiv g_v$ is the node lable of node $v$. So for this closed walk I have that $g(v_L)=g(v_0)$ so must be labelled by the identity element as the first and last nodes in this closed walk will always have the same node label defined in the previous part \thref{t:propbgg}.\ref{i:node} on the symmetrised graph. 
	
	To show that \thref{t:propbgg}.\ref{i:cswalk} implies \thref{t:propbgg}.\ref{i:swalk} I can consider semiwalks from any node $u$ to a node $v$. If there is only no such semiwalk or just one, then there is nothing to show (these cases are not constrained by balance). If there is more than one distinct semiwalk I can take each possible pair of such semiwalks, say $\Scal$ and $\Scal'$. By definition, the reverse of either one of these can be used to create a closed semiwalk, for instance $\Ccal = \Scalrev \circ \Scal'$. The label of this closed semiwalk is $e$ by point (\ref{i:cswalk}) so I have that 
	\beq 
	e 
	= \linkmap(\Ccal) 
	= \linkmap(\Scalrev) \times \linkmap( \Scal')
	= [\linkmap(\Scal)]^{-1} \times \linkmap( \Scal') \, .
	\eeq
	Right multiplying by $\linkmap(\Scal)$ gives me the results I require that $\linkmap(\Scal) =\linkmap( \Scal') $\,.
	
	Now I have proved this cycle of theorems, I deduce that the other dependencies required are also true. Namely that 
	\thref{t:propbgg}.\ref{i:node} implies \thref{t:propbgg}.\ref{i:swalk},
	\thref{t:propbgg}.\ref{i:cswalk} implies \thref{t:propbgg}.\ref{i:node}, and
	\thref{t:propbgg}.\ref{i:swalk} implies \thref{t:propbgg}.\ref{i:cswalk}.	
\end{proof}

In fact I can tighten these theorems to replace ``semiwalks'' by ``semipaths'' in \thref{t:propbgg}.\ref{i:swalk} and ``closed semiwalks'' by ``semicycles'' in \thref{t:propbgg}.\ref{i:cswalk} with \citet{HNC65} using ``semipaths'' and ``semicycles'' in their statement of this theorem for signed directed networks. This is simply because semiwalks are concatenations of semipaths and semicycles, so closed semipaths are simply made up of concatenations of semicycles. These results are generic to directed networks (and, in fact, follow from equivalent properties on undirected networks) and are shown in the literature such as in \citet{HNC65}.

\subsubsection{Balance, signs and the $Z_2$ group}

Balance is one of the key concepts for signed networks. Here I will highlight that balanced signed directed networks are balanced $Z_2$ group graphs. Since $Z_2$ is one of the groups where balance is possible for an undirected group graph, this result covers both directed and undirected cases.

\Thref{t:propbgg} encoded three properties of balanced group graphs that match properties known for balanced signed networks (e.g.\ see theorem 13.2 on p 342 in \citet{HNC65}). I shall use just one of these properties to make the connection between signed networks and $Z_2$ group graphs explicit but I could have used any of the three properties. 

One definition of a balanced signed network is in terms of partitions of the nodes in a signed undirected network (theorem 3 of \citet{H53}) or for directed signed networks (e.g.\ see theorem 13.2 on page 342 in \citet{HNC65})
\begin{definition}[A balanced signed network in terms of a partition]
	A balanced signed network (directed or undirected)  is one where the nodes can be partioned into two components where the links between nodes in the same component are only positive while links between those in different clusters are always negative links.
\end{definition}
This is just a particular case of \lemref{l:nltoidsw} and its converse \lemref{l:indentitywalk} for a $Z_2$ group graph in the parity irreducible representation of $Z_2$. So a balanced signed network is just a balanced group graph where the group is $Z_2$.

\subsection{Undirected group graphs}\label{s:udgg}

So far all the group graphs have been constructed from directed networks $\Dcal$. However, there are many examples of signed networks based on undirected networks. Let $\Ucal=(\Ncal,\Lcal)$ denote an undirected network where links between node $u$ and $v$ have no direction. These undirected links are a set of two nodes denoted by braces, $\{u,v\} =\{v,u\} \in \Lcal$, and are not an ordered sequence such as is denoted by $(v,u)$. 

An undirected network is equivalent to a symmetric directed network, that is $\Dcalsymm$, where every link in the undirected network $\{u,v\} \in \Lcal$ is represented by a pair of reciprocated directed links in an equivalent directed network where $(u,v) \in \Lcalsymm$ and  $(v,u) \in \Lcalsymm$. So there is no difference between walks and semiwalks on an undirected network --- every semiwalk on an undirected walk is also a walk.

The question is are there group graphs for undirected networks? The answer is yes but there is a strong limitation the examples where balance is possible.  

The definition a group graph $\Ggraph=(\Ucal, \Gcal, \linkmap)$ on an undirected network $\Ucal$ can be constructed from   \defref{d:grpgrph} of a group graph on a directed graph if the symmetric graph representation $\Dcalsymm$ of any undirected graph is used. This means that undirected links $\{u,v\}$ are given a label $\linkmap(\{u,v\}) \in \Gcal$ in some group $\Gcal$ regardless of the direction of the equivalent reciprocated directed links in $\Dcalsymm$. If there are semiwalks in $\Dcalsymm$, walks in $\Ucal$, passing along a link $\{v,u\}\in \Lcal$ then this link will contribute a factor of $g_{vu}=\linkmap(\{u,v\})$. For any semiwalk moving along this link in the opposite direction, this link will contribute the same link label as $g_{uv}=\linkmap(\{u,v\})=\linkmap(\{v,u\})=g_{vu}$ since the order of nodes in a link do not matter for an undirected network. 

The problem is that while a group graph based on an undirected network is well defined, it can not be balanced for general link labels from some group $\Gcal$. 
Consider semicycles $(v,u,v)$ formed from any undirected edge $\{u,v\} \in \Lcal$. The label on the only edge involved is always a single group element $g_{vu}=g_{uv}$ which means these semicycles always have a label of $(g_{vu})^2$. For many group elements in most groups this is not the identity. If this $(g_{vu})^2$ is not the identity for all the link labels in the group graph then that will violate one of the properties of a balanced graph that the label for any closed semiwalk is the identity, e.g.\ as summarised in \thref{t:propbgg}.\ref{i:cswalk}. 

The only groups where \emph{every} element of the group satisfies $g^2=2$ are product groups made from $Z_2$ subgroups, $(Z_2)^n$. 
So undirected group graphs based on these product groups are as unconstrained as the directed versions.
These $(Z_2)^n$ group graphs can be thought of as representing symmetric relationships between two nodes where if there is an edge present, then there are $n$ different aspects to this one relationship, each aspect is either positive or negative.  For instance, the existence of the edge might reflect significant interaction between two individuals who are represented by two nodes in a social network. The signs might indicate when various key properties of a pair of interacting nodes are compatible, the same, similar (positive) or different (negative). Properties which could include gender, political leaning, ``compatible'' Western or Chinese astrology sign, and so forth. All of these are always defined for any comparison between individuals so the multiple signs must be used to label any link. An example is shown in \figref{f:ggudZ2Z2}.

For example, looking in \citet{CH56}, one of the oldest papers on signed networks, the authors interpret ``Heider's conception of balance'' \citep{H46} in terms of \emph{two} binary relations. The first relation is denoted by ``L'' and they describe this as concerning ``attitudes, or the relation of liking or evaluating'' while ``the second type of relation refers to cognitive unit formation, that is, to such specific relations as similarity, possession, causality, proximity, or belonging''.
This second relation is denoted by ``U''. In each case the relation can be positive or negative. This gives me a representation as a $Z_2 \times Z_2$ group graph where every edge has two signs, one for the ``U'' relation and another  sign for the ``L'' relation.

\begin{figure}[htb]
	\begin{center}
		\includegraphics[width=0.9\textwidth]{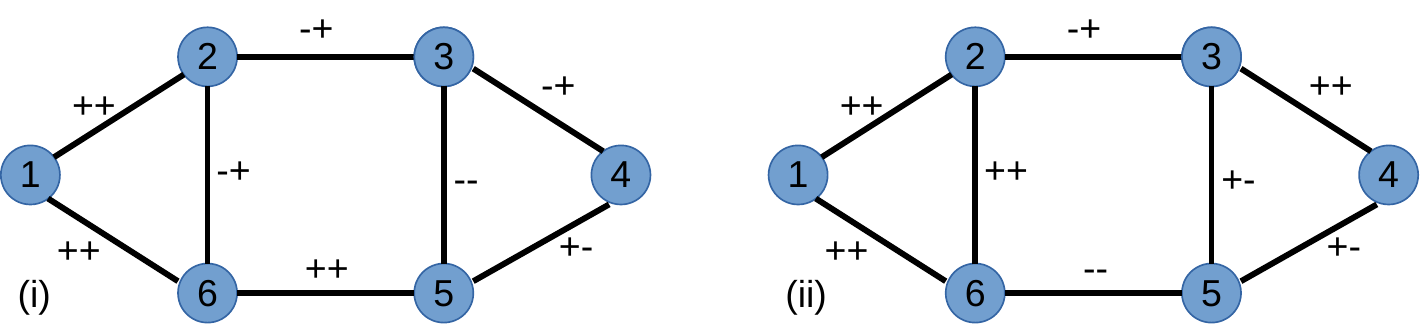}
	\end{center}
	\caption{Examples of a ``type 2'' network in the terminology of \citet{CH56}, a $Z_2 \times Z_2$ group graph. Here the existence of an edge indicates there is a relationship between two individuals (nodes). Each edge carries a label given by  an ordered pair of signs $\pm\pm$. The first sign represents compatibility under one relation (say Western astrological star sign) while the second represents compatibility under a different relation (say Chinese astrological animal). These signs are always defined for any pair of individuals so every edge has to have both signs. That is, this network is not a more general type of two-layer network in which an edge between two nodes may exist on one layer but not the other. On the left in (i) is an unbalanced example while network (ii) on the right is balanced.  
	}
	\label{f:ggudZ2Z2}
\end{figure}

However, there are also find examples of undirected and balanced group graphs based on other groups but they are more contrived and appear to be less useful. I have noted that group graphs can be constructed without using all the group elements. This fact can be used to find balanced undirected group graphs based on groups where some elements do not satisfy $g^2=2$ provided such elements are not used as link labels. Some examples of balanced undirected group graphs for $S_3$ are shown in \figref{f:ggudbal}. 
\begin{figure}[htb]
	\begin{center}
		\includegraphics[width=0.65\textwidth]{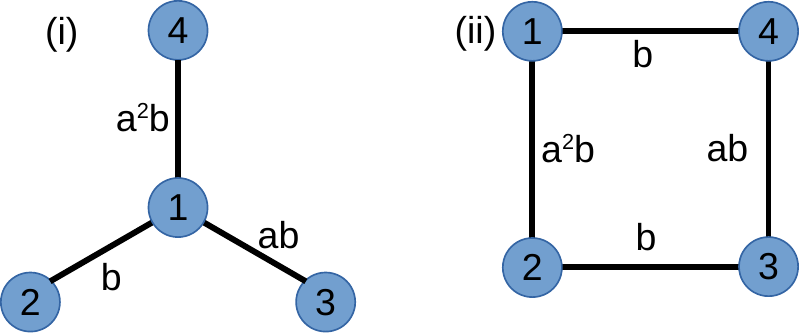}
	\end{center}
	\caption{Some examples of balanced undirected $S_3$ group graphs. The link label label must be an element where $g^2=e$ which for $S_3$ is the identity $e$ and the three reflection elements: $b$, $ab$ and $a^2b$. These four elements do not form a group on their own as seen by the need for a third group element, $a$, needed to write down these four elements of $S_3$.   
	}
	\label{f:ggudbal}
\end{figure}
It is unclear if such balanced undirected group graphs when based on a limited set of group elements are useful or not.

\section{Dynamics on group graphs}\label{s:dynamics}

One of the motivations for defining group graphs is to extend the results of \citet{TL24a} for processes on signed networks and \citet{TL24} for processes on networks with edges labelled by complex numbers. To do this, I first need to move away from abstract groups to look at representations of groups.

\subsection{The representation of groups}\label{s:grouprep}

So far I have used the labels $e,a,b$ etc to denote the elements of the group in an abstract sense. However, when looking at symmetries in certain processes on a network the group elements will appear as matrices. That is I will use a \tsedef{representation} of the symmetry group in terms of a set of $d$-by-$d$ matrices\footnote{In principle I can use other structures to represent the group but I will limit ourselves to linear processes in real- and complex vector spaces and hence I will only consider matrix representations.} 
$\{\Dmat(g)| g \in \Gcal \}$. The key property of this matrix representation is that when multiplying these matrices together they reflect the group multiplication properties so if $g_2 \times g_1 = g_3$ then $\Dmat(g_2)\Dmat(g_1)=\Dmat(g_3)$. 

In a $d$-dimensional representation, group elements $g$ are represented by $d$-by-$d$ representation matrices $\Dmat(g)$ which act on $d$-dimensional vectors that live in a $d$-dimensional vector space $\Rcal$ (usually real or complex numbers). If the matrices in the representation are all distinct, so $\Dmat(g_1) = \Dmat(g_2)$ if and only if $g_1=g_2$,  then the representation is \tsedef{faithful} otherwise the representation is \tsedef{unfaithful}.

Finally a representation is \tsedef{reducible} if there is a similarity transformation in terms of an invertible matrix $\Smat$ such that $\Smat^{-1} \Dmat(g)\Smat$ is block diagonal for all $g \in \Gcal$. Any representation where this transformation does not exist is called an \tsedef{irreducible representation} (often abbreviated as an \tsedef{irrep}). 
There is a well established machinery to find the decomposition of any given representation into its irreducible representations. Likewise the irreducible representations of most groups of interest are well known. Here, I will focus on two groups: $Z_2$ and $S_3$.

One of the best known examples of a group graph is a signed network. A signed network has two types of link, usually labelled by $+1$ and $-1$. These can be interpreted as the matrices from a one-dimensional matrix representation of $Z_2$ where $\Dmat(e) = +1$ (identity) and $\Dmat(a) = -1$. This is the only faithful irreducible representation of $Z_2$ and it is called the \tsedef{parity} irreducible representation.

A two-dimensional representation of $S_3$ is
\begin{eqnarray}
	\begin{array}{ccc}
		\Dmat(e)= 
		\begin{pmatrix}
			1 & 0 \cr\noalign{\vskip6pt}
			0 & 1 \cr
		\end{pmatrix},\hfill
		&
		\Dmat(a)   = {\textstyle{1\over2}}
		\begin{pmatrix}
			-1  & -\sqrt{3} \cr\noalign{\vskip6pt}
			\sqrt{3} &       -1  \cr
		\end{pmatrix}
		,\hfill 
		&
		\Dmat(a^2) = 
		{\textstyle{1\over2}}
		\begin{pmatrix}
			-1  & \sqrt{3} \cr\noalign{\vskip6pt}
			-\sqrt{3} &      -1  \cr
		\end{pmatrix}
		, \hfill
		\\
		\noalign{\vskip18pt}
		\Dmat(b) =
		\begin{pmatrix}
			-1 & 0 \cr\noalign{\vskip6pt} 
			0 & 1 \cr
		\end{pmatrix},
		&
		\Dmat(a b) =
		{\textstyle{1\over2}}
		\begin{pmatrix}
			1 & -\sqrt{3} \cr\noalign{\vskip6pt}
			-\sqrt{3} &       -1  \cr
		\end{pmatrix}, 
		&
		\Dmat(a^2 b) = 
		{\textstyle{1\over2}}
		\begin{pmatrix}
			1  & \sqrt{3} \cr\noalign{\vskip6pt}
			\sqrt{3} &      -1  \cr
		\end{pmatrix}
	\end{array}
	. \hfill
	\label{e:S3d2}
\end{eqnarray}
These matrices are symmetry transformations of an equilateral triangle centred at the origin in a two-dimensional plane when using orthogonal coordinates.  Matrices representing elements $b$, $ab$ and $a^2b$ are reflections (so have determinant $-1$) and these three elements form one class of $S_3$. The matrix for $a$ represents a rotation by $120^\circ$ so the element $a^2$ is here represented by a $240^\circ$ rotation. Together $a$ and $a^2$ elements form a second class of elements in $S_3$. The last class exists in all groups and that is the class containing just the identity element $e$, here represented by the two-dimensional unit matrix. These matrices form a faithful representation of $S_3$ and, apart from similarity transformations (i.e.\ changes of coordinate system), this is the only faithful irreducible representation of $S_3$ as well as being the only two-dimensional irreducible representation of $S_3$.

There are many other representations. For instance, I could use another irreducible representation of $S_3$, the unfaithful one-dimensional irreducible representation where $\Dmat(e)=\Dmat(a)=\Dmat(a^2)=+1$ and $\Dmat(b)=\Dmat(ab)=\Dmat(a^2b)=-1$. In this case all the extra structure of the group $S_3$ has been lost and this is effectively working in terms of a faithful representation of a $Z_2$ subgroup of $S_3$. In such a case it is better to think of this signed network as a faithful representation of a $Z_2$ unless the full symmetry properties of $S_3$ appear elsewhere in the problem.

\subsection{The representation of group graphs}

Groups are normally encountered as a representation so it makes sense to define a group graph in this context.
\begin{definition}[Representation of a Group Graph]
	A \tsedef{representation of a group graph} is a Group Graph $\Ggraph$ where the group elements for the links lie in one representation $D$ of the group where each link $\link$ is labelled by $D(g_\link)$.
\end{definition}


I will only consider matrix representations so I have a $d$-by-$d$ matrix $D(\linkmap(\link))$ representing the group element $\linkmap(\link)$ for every link $\link$ in the group graph. I will use the notation $\Dmat(\link) \equiv \Dmat(g_\link) \equiv \Dmat(\linkmap(\link))$ to indicate the matrix $\Dmat$ used to represent the label of a link $\ell$, where $\{\Dmat(g) | g \in \Gcal\}$ is the matrix representation of the group\footnote{Formally I am mixing two maps in my notation. First a map from the set of links to a matrix, $D: \Lcal \to \Rcal \times \Rcal$ where $\link \mapsto \Dmat(\link)$. Then I have the usual notation for a matrix representation of a group where $D: \Gcal \to \Rcal \times \Rcal$ with $g_\link \mapsto \Dmat(g_\link) \equiv \Dmat(\linkmap(\link))$. The context should make it clear which map is being used.}.

An example of $Z_2$ group graph in the parity representation, i.e.\ a traditional signed network, is shown on the left in \figref{f:ggrepex}.  
On the right of \figref{f:ggrepex} I shown an $S_3$ group graph lying in the two-dimensional representation \eqref{e:S3d2}. 

\begin{figure}[htb]
	\begin{center}
		\includegraphics[width=0.9\textwidth]{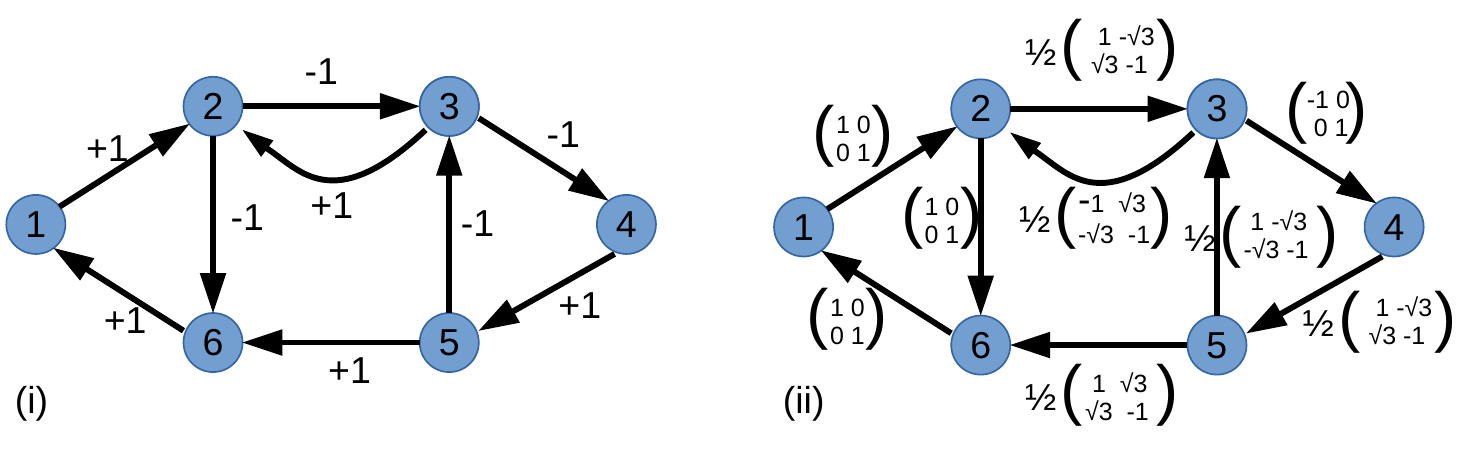}
	\end{center}
	\caption{Two examples of a matrix representation of the group graphs in \figref{f:ggex}. On the left in (i) I use the non-trivial one-dimensional irreducible representation of $Z_2$ where $D(e)=+1$ and $D(a)=-1$ to give a a matrix representation of the group graph on the left in \figref{f:ggex}. This is an unbalanced signed network as defined by \citet{DM09} where positive (negative) links are labelled with $+1$ ($-1$). In (ii) on the right is a two-dimensional matrix representation of the $S_3$ group graph shown on the right of \figref{f:ggex}. Group elements for each link are shown next to the relevant link. This matrix representation of $S_3$ is given in \eqref{e:S3d2}. 
	The numbers inside the blue circles are used to identify nodes in examples in the text.   } 
	\label{f:ggrepex}
\end{figure}

The concepts and results for group graphs developed earlier with abstract group elements all have a natural generalisation into the context of a representation of a group graph given that matrix multiplication in the representation reflects the group structure.

\subsection{The dynamics on balanced group graphs}

To show how group graphs can give useful insights, I will consider a some dynamical processes on a group graph $\Ggraph=(\Dcal,\Gcal,\linkmap)$. Here the topology described by the underlying directed network $\Dcal$ plays the role of `geographical' space, much as a regular lattice or continuous Euclidean space describe the space in which physical processes, such as diffusion of dye on a piece of paper, occur. At each node $v$ on the network $\Dcal$, at each point in this space, I will also add an ``internal space'' $\Rcal$. Each node $v$ carries a value in this space. I will assume points in this internal space can  be represented by $d$ values. So ultimately I have values $w_{vi}(t)$ for each node $v$ at time $t$ where this represents the $i$-th coordinate in the $d$-dimensional internal space $\Rcal$. For instance, perhaps there are two types of dye molecules so $w_{vi}(t)$ represents the density of dye type $i$ at node $v$ at time $t$. These values could be complex as used in quantum mechanics. So $w_{vi}(t)$ might be the amplitude of a quantum state at node $v$ at time $t$ expressed in terms of two basis states labelled by $i$, say $i = 1$ is a spin up eigenstate and $i=2$ is a spin down eigenstate.
I will look at some very general examples of how these values evolve where the evolution is given by a linear equation of the form
\bea
 \itervec{t+1}  &=& \Wmat \, \itervec{t} \, . 
 \label{e:cggrwsimple}
\eea
The matrix $\Wmat$ is an $N.d$-by-$N.d$ matrix with entries $W_{vj,ui}$ linking the value of the $i$-th coordinate of node $u$ at time $t$ to the value of the $j$-th coordinate of node $v$ at time $t+1$. In the context of a complex network, this matrix entry will be zero unless there is an edge from $u$ to $v$ in the directed network $\Dcal$.

The symmetry described by the group $\Gcal$ is in the way the dynamical process moves values through the internal space $\Rcal$, and this will be the representation space $\Rcal$ mapped by the symmetry matrices of the group. Specifically, it is assumed that the transformation matrices $\Wmat$ (where $\Wmat : \Ncal \times \Rcal \to \Ncal \times \Rcal$) can be factorised into two parts, where
\beq
 W_{vj,ui}= M_{vu} D_{ij}(v,u) \, .
 \label{e:Wfactor}
\eeq
This is not quite a direct product of two matrices because the second part, the matrix $\Dmat$, is a representation of the group element $g_{vu} = \linkmap((v,u))$, the link label of the corresponding link.

The first part $\Mmat$ of the transformation matrix $\Wmat$ in \eqref{e:Wfactor} is based solely on the topology of the network. This is a matrix $M_{vu}$ ($\Mmat: \Ncal \to \Ncal$)  which is zero when there is no link from $u$ to $v$.  The simplest example would be the adjacency matrix $\Amat$ of the underlying network $\Gcal$. This type of process on a network is the basis of eigenvalue centrality and related centrality measures. Another example would be a diffusion process where $M_{vu}= A_{vu}/s^\mathrm{(out)}_u$ where $s^\mathrm{(out)}_u = \sum_v A_{vu}$. For instance this is used to construct the PageRank centrality measure.

The second part $\Dmat(\ell)$ of the transformation matrix $\Wmat$ changes with each link $\ell$ in the network. 
Here the link labels $\Dmat(g)$ are $d$-dimensional matrices $\Dmat(g)$ acting on a $d$-dimensional vector space $\Rcal$ (usually a real or complex vector space). The matrices $\Dmat(u,v) \equiv \Dmat(g_{uv})$ represent the link labels, the group elements $g_{uv}=\linkmap((u,v)) \in \Gcal$ for an link $(u,v)$ from node $v$ to node $u$. 
So these matrices $\Dmat(\ell)$ mix the node values within the representation space alone, ($\Dmat(\ell): \Rcal \to \Rcal$).

It is worth noting that this process can be viewed as occurring on a multilayer network. In this case, nodes appear on all $d$-layers with $w_{vi}(t)$ representing the value at node $v$ on layer $i$. However, changing the basis used in the representation space $\Rcal$ will also correspond to choosing new layers in the network.

The central idea is that balance in a group graph means that the dynamics of the process is completely controlled by the topology of the network and the mixing in the representation space, between layers, plays no role. This idea is captured by the following theorem.
\begin{theorem}[Spectral equality for balanced Group Graphs]\label{t:speq}
	A linear dynamical process process $\itervec{t+1}  = \Wmat \, \itervec{t}$ in the node-representation space $\Ncal \times \Rcal$ of a balanced group graph when the factorisation of $W_{vj,ui}= M_{vu} D_{ij}(v,u)$ is present is equivalent to a linear process on the group graph in the identity representation.
\end{theorem}
\begin{proof}
Start from the simple homogeneous dynamical process \eqref{e:cggrwsimple} defined on the nodes in $\Ncal$ and on the $d$-dimensional representation space $\Rcal$. Given the factorisation of \eqref{e:Wfactor} this process can be written as 
\bea
	\itervector_{ui}({t+1})  &=& \sum_{v\in\Ncal} \sum_{j=1}^d M_{uv} D_{ij}(u,v) \itervector_{vj}({t}) \, .
	\label{e:cggrwsimple2}
\eea
Using the decomposition of link labels in terms of the node labels in a balanced group graph from \corref{c:lldecomp}, $\linkmap\big((v,u)\big) = g_v.(g_u)^{-1}$, gives
\bea
  \itervector_{ui}({t+1})  
  &=& 
  \sum_{v\in\Ncal} \sum_{j=1}^d M_{uv} \sum_{k=1}^d [\Dmat(u)]_{ik} [\Dmat^{-1}(v)]_{kj} \itervector_{vj}({t}) 
  \\
  &=& 
  \sum_{v\in\Ncal} \sum_{j=1}^d \sum_{k=1}^d [\Dmat(u)]_{ik} M_{uv}  [\Dmat^{-1}(v)]_{kj} \itervector_{vj}({t}) \, .
  \label{e:cggrwsimple3}
\eea
Left multiply by the inverse of the representation of the label of node $u$, $\Dmat^{-1}(u)$, to find that
\bea
\sum_{k=1}^d  [\Dmat(u)]_{ik} \itervector_{uk}  ({t+1})
&=& 
\sum_{v\in\Ncal} \sum_{j=1}^d M_{uv} \sum_{k=1}^d [\Dmat(u)]_{ik} [\Dmat^{-1}(v)]_{kj} \itervector_{vj}({t}) 
\\
&=& 
\sum_{v\in\Ncal} M_{uv} \sum_{j=1}^d   [\Dmat^{-1}(v)]_{ij} \itervector_{vj}({t}) \, .
\label{e:cggrwsimple4}
\eea
This then gives
\bea
x_{uj}(t+1)
&=& 
\sum_{v\in\Ncal} M_{uv} x_{vj} (t)
\label{e:xdef}
\eea
where new node $x_{ui}$ values at each node $u$ for layer $i$ are defined as
\bea
 x_{ui}(t) &=& \sum_{k=1}^d  [\Dmat(u)]_{ik} \itervector_{uk}({t})  \, .
\label{e:netrwvec}
\eea
\end{proof}
This shows that on a balanced group graph, the simple process on the nodes and in the $d$-dimensional group representation space $\Rcal$ of \eqref{e:cggrwsimple} is completely described by the same process on the network itself with each of the $d$-dimensions of the representation space decoupling and each behaving in the same way, i.e.\ equivalent to using in the trivial representation $\Dmat(g) = \unitmat$. Working in the $\xvec$ coordinates of \eqref{e:cggrwsimple3} means there is no mixing in the $j$ index of the $\Rcal$ representation space, no mixing of values in the different layers. All the properties of this process come from the topology of the network and the additional group structure adds nothing to the evolution on this network.

There is a second way to look at processes on balanced group graphs by exploiting the concept that many group representations are reducible. By using a similarity transformation, a  change of coordinates in the representation space $\Rcal$, the representation matrices $\Dmat(g)$ can be transformed into a block diagonal form. Each block is a copy of one of the irreducible representations of the group and each block represents a subspace of the full representation space, a subspace which is not mixed by the process on the nodes with other parts of the space. 

So consider a similarity transformation, a $d$-by-$d$ matrix $\Smat$, where 
\beq
\Dmat^\prime(g) =
\Smat^{-1} \Dmat(g) \Smat
\label{e:simtrans}
\eeq
In group theory, it is straightforward to show that the set of matrices $\{\Dmat^\prime(g)\}$ is a representation of the same group $G$ if $\{\Dmat(g)\}$ is a representation. Unless I am working with an irreducible representation, I can always find a transformation $\Smat$ which makes sure that all the $\Dmat^\prime$ matrices are block diagonal, with each block a separate irreducible representation.

The generic evolution equation \eqref{e:cggrwsimple2} can then be rewritten as follows
\bea
\sum_i [\Smat^{-1}]_{ki} w_{ui}({t+1})
&=& 
\sum_i [\Smat^{-1}]_{ki} \sum_{v\in\Ncal} \sum_{j=1}^d M_{uv} D_{ij}(u,v) w_{vj}({t}) 
\\
&=& 
\sum_{v\in\Ncal} \sum_{j=1}^d M_{uv} [\Smat^{-1} \Dmat(u,v) \Smat \Smat^{-1}]_{kj} w_{vj}({t}) 
\\
&=& 
\sum_{v\in\Ncal} \sum_{j=1}^d M_{uv} [\Dmat^\prime(u,v) ]_{kj}  \big( \sum_m [\Smat^{-1}]_{jm} w_{vm}({t}) \big)
\label{e:cggirrep1}
\eea
So then
\bea
y_{ui}({t+1})
&=& 
\sum_{v\in\Ncal} \sum_{j=1}^d M_{uv} D^\prime_{ij}(u,v)  y_{vj}({t}) 
\label{e:cggirrep2}
\eea
where the process is now described in terms of new variables $\yvec$ where
\bea
y_{ui}({t})
&=&
\sum_j [\Smat^{-1}]_{ij} w_{ui}({t}) \, .
\label{e:cggirrep3}
\eea
In terms of these $\yvec$ variables, I am now working in blocks linked by the symmetry and which, also by symmetry, cannot mix outside the blocks of the $\Dmat^\prime(g)$ matrices. For a general group with irreducible representations which are not one dimensional (i.e.\ non-abelian groups) such as $S_3$, such blocks are cannot be simpler than the factorisation into $d$ independent processes as given by the node label transformation to $\xvec$ in \eqref{e:xdef}. However, for abelian groups such as my $Z_2$ examples, here the irreducible representations are all one-dimensional and so this can be as informative as the node label transformation to $\xvec$ variables in \eqref{e:xdef}.

\subsection{A wider family of dynamical processes}\label{s:widerdp}

The results on dynamical systems of the form given in \eqref{e:cggrwsimple} and \eqref{e:Wfactor} showed that for balanced group graphs the group transformations had no effect on the evolution of the process which was controlled by the topology of the network alone, as encapsulated in \thref{t:speq}. I also showed how, with a restriction to the group of generalised stochastic matrices, processes based on a network Laplacian can be included in these results.  However these results, especially \thref{t:speq}, hold for a wider range of linear processes.

In many network processes the evolution equation is of the form 
\beq
\itervec{t+\Delta t}  = \itervec{t}  + \mu \Wmat \, \itervec{t}
\label{e:cggrwsimple1a}
\eeq
where $\Delta t$ is a constant setting the time scale for the process. The constant $\mu$ is a simple rescaling of the $\Wmat$ matrix. For instance, so far I have looked at discrete time updates where, without loss of generality, I set $\Delta t =1$. However if I take $\Delta t \to 0$ then the evolution equation may be written as a process in continuous time where $d \itervec{t} /dt   = \mu' \Wmat \, \itervec{t}$ with $\mu' = \mu/\Delta t$ and a finite $\mu'$ (requiring $\mu$ to scale in proportion to $\Delta t$) is now the rate constant for the process.

Sticking with the discrete version in \eqref{e:cggrwsimple1a}, 
one can repeat the previous analysis to see this extra term $\itervec{t}$ on the right-hand side does not change results. Another way to see this is to rewrite \eqref{e:cggrwsimple2} using the form \eqref{e:Wfactor} for $\Wmat$ to give 
\beq
w_{vj}(t+\Delta t) 
    =  \sum_{u \in \Ncal} \sum_{i=1}^d \delta_{vu}D_{ji}(v,v)w_{vj}(t)  
    +  \mu \sum_{u \in \Ncal} \sum_{i=1}^d  M_{vu} D_{ji}(v,u)  w_{ui}(t)
    \label{e:someeqn}
\eeq
That is this extra term when compared to \eqref{e:cggrwsimple}, the first term on the right-hand side of  \eqref{e:someeqn}, can be interpreted as equivalent to adding a self-loop of weight one and link label the identity $e$ to every node in the original balanced group graph $\Ggraph$. As self-loops must be labelled by the identity \lemref{l:cswbggid} in a balanced group graph, I am merely constructing a new balanced group graph $\Ggraph'$ where I have added these self-loops to the original directed graph $\Dcal$ and the evolution equation is again in the original form \eqref{e:cggrwsimple} with \eqref{e:Wfactor} where the matrix $\Mmat$ is replaced by $\Mmat^\prime$ given by
\beq
\Mmat^\prime =  \unitmat + \mu \Mmat \, .
\label{e:Wfactor2}
\eeq 

I can take this idea further. The \tsedef{centre} of a group $\Zcal(\Gcal)$ is the set of group elements which commute with all other group elements, so  $\Zcal(\Gcal) = \{ z | zg=gz \forall g \in G\}$. So consider an $N.d$-by-$N.d$ matrix $\Zmat = \unitmat \times \Dmat(z)$ with $Z_{vj,ui} = \delta_{uv} D_{ji}(z)$ where $\Dmat(z)$ represents any element $z$ from the centre of the group, $z \in \Zcal(\Gcal)$. If $\Wmat$ is of the form \eqref{e:Wfactor} then it follows that for any constant $\mu$
\beq
\itervec{t+1}  = \left( \sum_{z \in \Zcal} \mu_z \Zmat(z) \right) \itervec{t} + \Wmat \, \itervec{t}
\label{e:wevolz}
\quad
\Rightarrow
\quad
\xvec(t+1)  = \left( \sum_{z \in \Zcal} \mu_z \Dmat(z) \right)  \xvec(t) + \Mmat \, \xvec(t)
\eeq
where $\xvec$ is defined as before in \eqref{e:netrwvec}. This includes the common case $\itervec{t+1}  = \mu \itervec{t} + \Wmat \, \itervec{t}$ where $\Zmat$ is the unit matrix $Z_{vj,ui} = \delta_{uv} D_{ji}(e)$ since the identity is always an element of the centre of any group.

It is also straightforward to move from discrete time update \eqref{e:wevolz} to continuous time evolution as follows
\bea
\frac{\itervec{t+\Delta t} - \itervec{t} }{\Delta t }    
&=& 
\left( \sum_{z \in \Zcal} \frac{\mu_z - \delta_{ze}}{\Delta t } \Zmat(z) \right) \itervec{t} 
+ \Wmat \, \itervec{t}
\\
\Rightarrow \quad
\frac{d \xvec(t) }{d t }    
&=& 
\left( \sum_{z \in \Zcal} \mu^\prime_z \Dmat(z) \right)  \xvec(t) + \Mmat \, \xvec(t)
\eea
with some rescaling of the parameters from $\mu_z$ to $\mu^\prime_z$ as the limit $\Delta t \to 0$ is taken.

One of the simplest examples of a finite group with a non-trivial centres is the group $D_4$, the symmetries of a square in a plane which is generated by three elements $e,a,b$ and where the group can be generated by knowing that $a^4=e$, $b_2=e$ and $a^nb = b a^4-n$ for $n=1,2,3$. The centre in this case is $\Zcal = \{e,a^2\}$. If I consider a two-dimensional irreducible representation, the coordinate transformations in the plane where $a$ is a rotation by $90^\circ$ and $b$ is a reflection. In such a representation the coordinate transformations may be written as the two-dimensional identity matrix for $\Dmat(e)$ and minus the same unit matrix for $a^2$, $\Dmat(a^2) = - \unitmat$. In this case the centre elements do not add any interesting behaviour to the dynamics as the generalised form of \eqref{e:wevolz} is just adding a term $(\mu_e - \mu_{a^2})\unitmat$ to $\Mmat$ which is just a variation of the form \eqref{e:Wfactor2}.

\subsection{The network Laplacian and group graphs}\label{s:lap}

So far I have discussed network dynamics of the form $\itervec{t+1}  = \Wmat \, \itervec{t}$ given in \eqref{e:cggrwsimple} where $\Wmat$ is of the form $W_{vj,ui}= M_{vu} D_{ij}(v,u)$. 
However, the part based on network topology alone, the $\Mmat$ matrix, does not include diffusion on the group graph when it is described by the network Laplacian $\Lmat$ where $L_{vu} = \delta_{vu} \kout_v - A_{vu}$ for adjacency matrix $\Amat$ and out-degree (out-strength if weighted) $\kout_u = \sum_v A_{vu}$.  
This is studied by \citet{TKSL24} but it can be recast in the language of group graphs. 
Such processes would be written as 
\begin{align}
  w_{vj}(t+1)  
  =
  \quad & 
  w_{vj}(t+1) 
  \nnel
  &
   + \mu \left( \sum_{u \in \Ncal} \sum_{i=1}^d A_{vu} D_{ji}(v,u) w_{ui}(t) \right)
   - \mu \left( \sum_{u \in \Ncal} \sum_{i=1}^d A_{uv} D_{ij}(u,v) w_{vj}(t) \right)
     .
\label{e:gglap}
\end{align}
Here $A_{vu}$ is the adjacency matrix of the underlying network $\Dcal$, so $\Amat$ is assumed to be non-negative but not necessarily symmetric. The interpretation is that the second term on the right, multiplied by $+\mu$, represents value arriving at node $v$ from neighbouring nodes $u$. This term falls into the form \eqref{e:cggrwsimple} discussed earlier and causes no problems. It is the last term on the right of \eqref{e:gglap} which causes the problems. This last term represents value leaving node $v$ and it has the same form as second term but with indices transposed. So the flow down each edge appears twice in exactly the same form, once representing a decrease in value at the source node of an edge and a second time representing an increase in value of the target node of that edge.

Now there is a problem with \thref{t:speq} and the transformation 
\eqref{e:netrwvec} 
\bea
x_{vk}(t) &=& \sum_{j=1}^d  [\Dmat(v)]_{kj} w_{vj}(t)  \, .
\label{e:netrwvec2}
\eea
to variables $\xvec$ which removes the effect of the network topology on the dynamics.  Assuming our group graph is balanced, I left multiply \eqref{e:gglap} by the matrix $(\Dmat(v))^{-1}$ representing the node label of $v$ to find
\bea
x_{vk}(t+1)
&=&
\sum_{j=1}^d  [(\Dmat(v))^{-1})_{kj} w_{vj}(t+1)
\\
&=& 
  x_{vk}(t) 
+ \mu  \sum_{u \in \Ncal} A_{vu} )  x_{uk}(t) 
\nnel
&&
- \mu \left( \sum_{u \in \Ncal} A_{uv} 
				\sum_{i=1}^d \sum_{j=1}^d   
               [ \Dmat(u)(\Dmat(v))^{-1} ]_{ij} 
               [(\Dmat(v))^{-1}         ]_{kj}
               w_{vj}(t) \right)
               \, .
\label{e:gglap2}
\eea
The algebraic problem in the last term is that the expression is not a simple product of matrices and a vector as illustrated by the index $j$ appearing three times. Also the inverse matrix representing the node label $v$ appears twice. I can not use the properties of the representation of the inverse of a group and the structure of link and node labels in a balanced graph to simplify this last term. 

So the conclusion is that this generalisation \eqref{e:gglap} to a group graph of a dynamical process described by the network Laplacian does not satisfy the criteria of \thref{t:speq}, since \eqref{e:gglap} is not of the right form, nor the original theorem be extended to this Laplacian form by using the same algebra.

However, all is not lost. In \eqref{e:gglap2}, there is a big simplification if the column sums of all the link label matrices $\Dmat(\ell)$ in a group graph are constant. I will assume this column sum is 
always one to simplify the discussion.  What is required is that the matrices on links which appear in the group graph are all \tsedef{generalised stochastic matrices}. That is they are $d$-by-$d$ real matrices whose columns sum to one and which have an inverse but whose entries are \emph{not} restricted to lie between zero and one inclusive. In \appref{as:stochmat}, I show that these generalised stochastic matrices form a group.
So if $(\Dmat(u,v) = (\Dmat(u))^{-1}\Dmat(v)$ has column sums equal to one, then 
\bea
x_{vk}(t+1)
&=&
\sum_{j=1}^d  [(\Dmat(v))^{-1})_{kj} w_{vj}(t+1)
\\
&=& 
x_{vk}(t) 
+ \mu  \sum_{u \in \Ncal} A_{vu} )  x_{uk}(t) 
- \mu \left( \sum_{u \in \Ncal} A_{uv} 
           \sum_{j=1}^d   [(\Dmat(v))^{-1}         ]_{kj} w_{vj}(t) 
     \right)
\\
&=& 
x_{vk}(t) 
+ \mu  \sum_{u \in \Ncal} A_{vu} )  x_{uk}(t) 
- \mu \left( \sum_{u \in \Ncal} A_{uv}  x_{vk}(t) 
\right)
\eea
which gives
\bea
x_{vk}(t+1)
&=& 
x_{vk}(t) 
- \mu  \sum_{u \in \Ncal} \left(\kout_v \delta_{uv}  - A_{vu} \right)  x_{uk}(t) \, .
\label{e:gglap3}
\eea
That is, if the link labels are invertible generalised stochastic matrices (so their columns sum to one) then the dynamics of a balanced group graph reduce to those without any group structure but which is controlled simply by the network Laplacian and the network topology. So the simple result of \thref{t:speq} has been generalised to processes described by a network Laplacian only if the matrices used to represent the group elements labelling the links of the group graph come from the group of invertible generalised stochastic matrices (see \appref{as:stochmat} for definition of this group). 

Note that most problems would require `physical' stochastic matrices, that is one where the columns sum to one \emph{and} all entries lie between zero and one inclusive, but where invertibility is not required. In general, a set of such stochastic matrices does not satisfy the criteria for a group.  The main problem is that the inverse of a typical stochastic matrix may not exist or, if it does, then it is not a stochastic matrix. For example, consider the set $\Scal(2)$ of two-dimensional stochastic matrices $\Smat(a,b)$ which can be used to define physical diffusion processes. That is if acting on a vector $\pvec$ where $0 \leq p_i \leq 1$ and $\sum_i p_i =1$ then $\pvec' = \Smat(a,b) \pvec$ satisfies the same criteria.
\beq
\Scal(2) = \{ \Smat(a,b) | 0 \leq a,b \leq 1 \} \, , \quad
\Smat(a,b) =
 \begin{pmatrix}
	    a & (1-b) \\
	(1-a) & b
 \end{pmatrix}
 \, .
\eeq
These matrices form a representation of a \tsedef{monoid} not a group as they satisfy all the axioms of a group except for the existence of an inverse in $\Scal(2)$ for every element. 
For instance closure is expressed by $\Smat(a_1,b_1)\Smat(a_2,b_2)=\Smat(a,b)$ with $a=a_1a_2+(1-b_1)(1-a_2)$ and $b=b_1b_2+(1-a_1)(1-b_2)$ (a proof of closure is in \appref{as:stochmat}).
The inverse matrix of an element in $\Scal$ may not exist or, even if it does, the inverse matrix may not be in the set $\Scal$ i.e.\ it is not a `physical' stochastic matrix since
\beq
\big( \Smat(a,b) \big)^{-1} =
\frac{1}{1-a-b}
  \begin{pmatrix}
	   -b & (1-b) \\
	(1-a) & -a
 \end{pmatrix}
 \, .
\eeq
The column sums are one so this satisfies that part of the requirement to be a stochastic matrix. However, 
the inverse does not exist at all for any elements where $1=a+b$. If $a+b<1$ and $0 \leq a,b \leq 1$ then the inverse matrix exists but it has negative values on the diagonal and these are not in the set of physical two-dimensional stochastic matrices $\Scal(2)$.

The trick here is to realise that I need the inverse matrix representing every link label to exist as it is used in some of the proofs about balanced group graphs\footnote{Alternatively, one might consider thinking about consistent group graphs of \appref{as:cnsstgg} not balanced group graphs.}. However, these inverse elements need not be a link label on any of the links in a balanced group graph, so they do not appear in the expressions for the dynamical process. A group graph need not have every element of the group appearing on a link in the group graph. So one way to proceed is to work with the group defined under matrix multiplication by the set of $d$-by-$d$ invertible generalised stochastic matrices matrices $\Scalhat(d)$ where matrix entries are real, every matrix column sums to one and all matrices are invertible. This is shown to be a group in \appref{as:stochmat}.

So to summarise, dynamical processes on a group graph when described by a network Laplacian do not satisfy our earlier theorem unless the link label matrices $\Dmat(\ell)$ come from the group of invertible generalised stochastic matrices  $\Scalhat(d)$ defined in \appref{as:stochmat}. For physical purposes, it seems likely that link labels will come from a subset of invertible stochastic matrices (those with entries between zero and one) but the group graph must be analysed in terms of the full group of invertible generalised stochastic matrices.

\subsection{An example of dynamics on a balanced signed network}

To illustrate the connection between signed graphs and $Z_2$ group graphs I will look at an example based on \figref{f:bggex} but using a representation that makes contact with the work of \citet{TL24} on balanced signed networks.

The group $Z_2$ is an abelian group of two elements: the identity $e$ and $a$ where $a^2=e$. The only faithful irreducible representation is $D(e)=+1$ and $D(a)=-1$. So a $Z_2$ group graph is equivalent to a \tsedef{signed network} \citep{H53,CH56,HNC65,WF94,DM09,ST10,SLT10,D17b,KCN19,C21,TL24a,TL24}. 
However, in \citet{TL24} a diffusion process on two layers is used so here I will work in a two-dimensional representation space $\Rcal$. 

So the process will be defined in terms of the values $w_{vi}(t)$ at node $v$ on layer $i$ at time $t$ where $\itervec{t+1}  = \Wmat \, \itervec{t}$ as defined in \eqref{e:cggrwsimple} and \eqref{e:Wfactor}. The topology of the connections between nodes (regardless of layers) is given in \figref{f:bggex}. I will look at a diffusion process on this network where the transition matrix for the process is $T_{vu} = A_{vu} / k^\mathrm{(out)}_u$ with adjacency matrix $\Mmat$ and the out-degree of a node $v$ given by $k^\mathrm{(out)}_v=\sum_u A_{vu}$. For the $Z_2$ network shown in \figref{f:bggex} the matrices are
\beq
\Amat =
\begin{pmatrix}
	0 & 0 & 0 & 0 & 0 & 1 \\
	1 & 0 & 1 & 0 & 0 & 0 \\
	0 & 1 & 0 & 0 & 1 & 0 \\
	0 & 0 & 1 & 0 & 0 & 0 \\
	0 & 0 & 0 & 1 & 0 & 0 \\
	0 & 1 & 0 & 0 & 1 & 0 \\
\end{pmatrix}
\, , \quad
\Tmat =
\begin{pmatrix}
	0 & 0   & 0   & 0 & 0   & 1 \\
	1 & 0   & 1/2 & 0 & 0   & 0 \\
	0 & 1/2 & 0   & 0 & 1/2 & 0 \\
	0 & 0   & 1/2 & 0 & 0   & 0 \\
	0 & 0   & 0   & 1 & 0   & 0 \\
	0 & 1/2 & 0   & 0 & 1/2 & 0 \\
\end{pmatrix}
\, .
\label{e:z2extopology}
\eeq

Values can move between the two layers in one of two ways, depending on the label of the link \figref{f:bggex}. For positive links (labelled by $e$), values will remain on the same layer as they pass along a link. For negative links (labelled by $a$), values will switch layers. This is a two-dimensional representation of a balanced group graph where
\beq
\Dmat(e) =
\begin{pmatrix}
	1 & 0 \\
	0 & 1
\end{pmatrix}
\, ,
\quad
\Dmat(a) =
\begin{pmatrix}
	0 & 1 \\
	1 & 0
\end{pmatrix}
\, .
\label{e:z2rep2d}
\eeq

For the example in \figref{f:bggex} in the representation given here, the value at time $t$ at node $2$ on the first layer, say $A$ ($i=1$) is divided equally between nodes $3$ and $6$ for time $(t+1)$ but it is passed to the value of node $3$ on the second layer, say $B$ ($i=2$), since link $(3,2)$ is negative while the value from node $2$ on layer $A$ contributes to the value of node $6$ on the same layer $A$. 


Using the node labels from \figref{f:bggexVL} and using the transformation from $\itervector$ node values to $\xvec$ node values in \eqref{e:netrwvec}, the new node variables are
\begin{align}
	x_{1\alpha} (t+1) &= w_{1A}(t)\, , \; &
	x_{2\alpha} (t+1) &= w_{2A}(t)\, , \; &
	x_{3\alpha} (t+1) &= w_{3B}(t) \, , \;&
	\nnel
	x_{4\alpha} (t+1) &= w_{4B}(t) \, , \; &
	x_{5\alpha} (t+1) &= w_{5B}(t) \, , \; &
	x_{6\alpha} (t+1) &= w_{6A}(t) \, , \; &
	\label{e:z2xalpha}
\end{align}
on a new layer labelled $\alpha$ and the remaining variables lie on a second new layer labelled by $\beta$ where
\begin{align}
	x_{1\beta} (t+1) &= w_{1B}(t)\, , \; &
	x_{2\beta} (t+1) &= w_{2B}(t)\, , \; &
	x_{3\beta} (t+1) &= w_{3A}(t) \, , \; &
	\nnel
	x_{4\beta} (t+1) &= w_{4A}(t) \, , \; &
	x_{5\beta} (t+1) &= w_{5A}(t) \, , \; &	
	x_{6\beta} (t+1) &= w_{6A}(t) \, . \; &
	\label{e:z2xbeta}
\end{align}

Note that new labels $\alpha$ and $\beta$ are needed to indicate the new layers in terms of these $\xvec$ variables.
As \eqref{e:netrwvec} shows, the diffusion process on the balanced $Z_2$ group graph shown in \figref{f:bggex} in the representation defined in \eqref{e:z2rep2d}, occurs independently on two new layers, $\alpha$ and $\beta$ and in both cases is completely controlled by the topology of the underlying graph described in \eqref{e:z2extopology}. 
Of course, this factorisation into new layers based on node label transformations of \eqref{e:netrwvec} is not so trivial with larger groups or networks.

However, I can also use this example to show how to use a similarity transformation \eqref{e:simtrans} to rewrite everything in terms of the irreducible representations, in terms of blocks which are not mixed as demanded by the symmetry of the problem. In this case, a one useful similarity transformation is 
\beq
\Smat = \frac{1}{\sqrt{2}} 
\begin{pmatrix}
	1 & 1 \\
	1 & -1
\end{pmatrix}
\eeq
since
\beq
\Dmat^\prime(e) =
\Smat^{-1} \Dmat(e) \Smat^{-1} 
=
\begin{pmatrix}
	1 & 0 \\
	0 & 1
\end{pmatrix}
\, ,
\quad
\Dmat^\prime(a) =
\Smat^{-1} \Dmat(a) \Smat^{-1} 
=
\Dmat(a) =
\begin{pmatrix}
	1 & 0 \\
	0 & -1
\end{pmatrix}
\, .
\eeq
This shows that the two-dimensional representation \eqref{e:z2rep2d} is equivalent to two one-dimensional representations: the $D^\prime_{11}(g)$ entries are the identity representation and the $D^\prime_{22}(g)$ entries are the parity irreducible representation. For general groups there are powerful tools which allow me to find these block diagonal forms but it is relatively easy to guess the answer for this simple $Z_2$ case. 
This similarity transformation is leaves the process defined in terms of a different set of variables $y_{ui}(t)$ as defined in \eqref{e:cggirrep3} with variables defined in terms of two new layers different from those used with the $x_{ui}(t)$ variables of \eqref{e:z2xalpha} and \eqref{e:z2xbeta}. The new variables on the first new layer $\alpha$ are simply the sum of the values on the original two layers $A$ and $B$, that is 
\beq
y_{i\alpha} (t+1) = \frac{1}{\sqrt{2}} \big( w_{iA}(t) + w_{iA}(t) \big) \, .
\eeq
The values on the second layer $\beta$ are simply the difference of the values on the two original layers but with an additional minus sign for the nodes labelled by $a$. The node labelling is not unique so I will use the node labels shown for in \figref{f:bggexVL} for the balanced $Z_2$ group graph. This gives the values on the new second layer $\beta$ as
\begin{align}
	y_{1\beta} (t) &= \frac{1}{\sqrt{2}} \big( + w_{1A}(t) - w_{2B}(t) \big) \, , \; &
	y_{2\beta} (t) &= \frac{1}{\sqrt{2}} \big( + w_{2A}(t) - w_{2B}(t) \big) \, , \; &
	\nnel
	y_{3\beta} (t) &= \frac{1}{\sqrt{2}} \big( - w_{3A}(t) + w_{3B}(t) \big) \, , \; &
	y_{4\beta} (t) &= \frac{1}{\sqrt{2}} \big( - w_{4A}(t) + w_{4B}(t) \big) \, , \; &
	\nnel
	y_{5\beta} (t) &= \frac{1}{\sqrt{2}} \big( - w_{5A}(t) + w_{5B}(t) \big) \, , \; &
	y_{6\beta} (t) &= \frac{1}{\sqrt{2}} \big( + w_{6A}(t) - w_{6B}(t) \big) \, . \; &
\end{align}
Again, the dynamics of these two $y$ variables for each node, one for layer $\alpha$ and one for layer $\beta$, is completely independent of the other set of $y$ variables and the evolution is completely controlled by the topology of the network. This model and the transformations to $y$ variables are exactly as discussed in \citet{TL24} for a balanced signed network.

Finally, note that the example posed in \citet{TL24} uses a diffusion process based on the Laplacian matrix not on the PageRank matrix as used here in \eqref{e:z2xbeta}. However, the matrices used as link labels \eqref{e:z2rep2d} are not general stochastic matrices but are a special subset of unitary invertible stochastic matrices which form a group in their own right. The fact the columns of the link label matrices sum to one and they are invertible mean that the results of \secref{s:lap} apply and the analysis shown here will still work for diffusion described by a Laplacian.

\section{Discussion}\label{s:discussion}

\subsubsection*{The generalisation of signed networks and balance}

In this paper I have generalised the concept of a signed network to what I call (following \citet{HLZ82}) a \tsedef{group graph}, a structure which combines a network to describe a topology that has no symmetry, along with a symmetry group that describes features in the processes occurring on the nodes and links of that network. In the language of this paper, a signed network is a $Z_2$ group graph.

The definition of a group graph given here is very similar to, but not exactly the same as, the existing concepts of 
a ``current graph'' \citep{GA73}, 
a ``voltage graph'' \citep{G74,Z82}, 
``group labelling'' of \citet{ES79}, 
the ``group graph'' of \citet{HLZ82}, 
a ``gain graph'' \citep{Z89}, 
a ``G-gain graph'' \citep{CDD21}, 
and the ``group synchronisation'' problem \citep{KEES03,GK06a,S11e,AF19}. 
The main difference is that the definition used in this paper has no restrictions on the labels of the links of a group graph including no restrictions on the directions, reciprocity or undirected nature of links. 
Typically, the literature works with a graph which is undirected but it behaves as if it was what I call a symmetric directed graph (maximally reciprocated) where reciprocated edges are always balanced, that is they carry link labels which are inverses of each other. 
This can be seen in many places such as in condition C1 on page 943 of \citet{GA73} on voltage graphs, in equation (2) of \cite{AF19} for group labelled graph, and in the first paragraph of section 5 in \citet{Z89} for gain graphs\footnote{At the time of writing, this can also be seen in informal summaries such as 
the \href{https://en.wikipedia.org/wiki/Voltage_graph}{voltage graph article on Wikipedia}, 
the \href{https://mathworld.wolfram.com/VoltageGraph.html}{voltage graph article on Mathworld},
and the \href{https://en.wikipedia.org/wiki/Gain_graph}{Wikipedia article on gain graphs}.}.

The language and approach used here and in the earlier work can be quite different. In particular, the use of adjacency matrices taking values from an algebra defined over a group in \citet{CDD21}, rather than the usual real or complex numbers, is a different but striking approach to the generalisation of networks with binary values (simple networks), real non-negative values (weighted networks), negative and positive reals (signed graphs), complex values and so forth.
This approach to adjacency matrices is summarised in \appref{as:adjmat} though this representation is not suitable for networks with multiedges.

Some of the restrictions seen in earlier work only require minor adjustments to extend them to the general case considered here (e.g.\ allowing for self-loops) but some of these extensions, for instance for networks with directed edge or multiedges, are important in real applications. 
The most important generalisation contained here is that I want to allow for cases where there are reciprocated but unbalanced edges. 
For example, there is no reason why two people connected in a social network should have the same opinion of each other, i.e.\ reciprocated edges need not be balanced. 
A node may estimate the time difference between its clock and that of a neighbouring node. However that estimate will not agree with the neighbour's estimate of the time difference unless the two neighbours engage in further exchanges to agree on a single consistent estimate of the time difference in their clocks.
If the focus is on balance then  balanced reciprocated edges is a necessary restriction for any balanced group graph by  \corref{c:reciplabel}. So one place where this work differs from earlier work is to allow for \corref{c:reciplabel} to be false, a case of unbalanced group graphs not covered in the literature.

If one focuses on balanced group graphs then most of the differences between the work here and earlier work are relatively minor or become irrelevant.  That is for a balanced group graph reciprocated edges must carry inverse group elements of each other as expressed in \corref{c:reciplabel}. Most of these earlier works  discuss balance \citep{G74,HLZ82,Z82,Z89,CDD21}  (all using this same term) so all of the definitions and results given here regarding balance can be found in a very similar form if very different language somewhere in these earlier works, though not all lemmas are found in every previous work.

In terms of processes on a network, the earlier work has not made the connection between balanced graphs and processes on graphs.
Here I have shown explicitly how to generalise the results of \citet{TL24} for signed undirected networks (i.e.\ for a $Z_2$ group graph) and \citet{TL24a} for undirected networks with complex weights (i.e.\ a $U(1)$ group graph) to a context with arbitrary groups on a group graph. Namely, that for balanced group graphs then for many typical processes on a group graph the group structure does not effect the dynamics and it is purely the topological structure of the graph that controls the dynamics. I showed that these results hold for an arbitrary linear process described by $\itervec{t+1}= \Wmat \, \itervec{t}$ of \eqref{e:cggrwsimple} provided I could factorise the transition matrix $\Wmat$ into the form where $W_{vj,ui} = M_{vu}D_{ji}(g_{vu})$ of \eqref{e:Wfactor}. Here the matrix $\Mmat$ is an $N$-by-$N$ matrix ($N$ is the number of nodes) and $M_{vu}$ is zero when there is no link from $u$ to $v$ so $\Mmat$ is only controlled by the topology of the network. 
Examples of $\Mmat$ include the adjacency matrix and the PageRank matrix.
I also found there is a way to use this formalism for a generalisation of a diffusion processes represented by the network Laplacian, as seen in \citep{TKSL24}, but then it may be more useful to think of the link labels as a representation of a monoid or semigroup rather than a group.

The results here regarding the trivial nature of processes on balanced networks, i.e.\ we can transform into coordinates where the link labels are all the identity element, are related to existing results. A balanced graph is `switching equivalent' (see \defref{d:switching}) to a group graph with identity labels on every link, something discussed for signed networks in \citet{AR68} and \citet{Z82}, and for gain graphs in \citet{Z89}. 
In the same way, lemma 5.1 of \citet{CDD21} states that, at least for a ``finite dimensional faithful unitary representation'', the spectrum the G-gain graph consists of copies of the spectrum of underlying graph.

I have shown in \secref{s:udgg} that this work can be applied to undirected networks but then enforcing balance becomes a stronger constraint. When the group is a product group of the type $(Z_2)^n$, i.e.\ where there are $n$ signs for $n$ different binary relations associated with every link, then all the work on balanced group graphs applies without change. Should the undirected group graph have labels from another group, then only labels from a subset of elements where $g^2=e$ can be used as link labels. 

Finally, working with some processes on networks also poses problems for the imposition of balance on a group graph. In particular, if the process is diffusion described by the Laplacian matrix, then \secref{s:lap} shows that the results for balanced group graphs only apply if the link labels are restricted to a set of invertible stochastic matrices so that the group is the group of generalised invertible stochastic matrices (see \appref{as:stochmat} for details on these matrices).

\subsubsection*{Non-unitary group representations}

The examples so far have worked with unitary representations of groups where $\Dmat (g) \Dmat^\dagger (g) = \unitmat$ for any group element $g \in \Gcal$. Indeed, all representations of finite groups are unitary representations.  However, in this work I do not use the unitary property and the results hold for non-unitary representations.

The simplest example of a group graph using a non-unitary representation is the group of positive real numbers $\{ \exp(\alpha ) \, | \, \alpha \in \Rbb\}$ under multiplication, a representation of the non-compact $U(1)$ group. 
The link labels $\Dmat(\Lambda(\ell)) = \exp(\alpha_\ell)$ now represent an amplification or suppression of the values flowing along the network in a process described in \eqref{e:cggrwsimple}. However, my results still apply, that is for a balanced group graph I need to have the label of the path of any closed semiwalk equal to the identity, at is $\Lambda(C) =1$ for all closed semiwalks $C$. So in this non-unitary case, balance means that any closed semiwalk must balance the amplification and damping factors as you go round a cycle. For instance, the link labels on reciprocal links must be $\exp(\alpha)$ and $\exp(-\alpha)$ for some $\alpha \in \Rbb$. 

Alternatively, we can work with the exponents $\alpha_\ell\in \Rbb$ as the labels, the representation $\Dmat(\ell)$ of the group elements $\Lambda(\ell)$, on each link.  Then the group `multiplication' is simply addition of real numbers. This is the representation we might use if we were looking at synchronisation of clocks at each node but all each node has is the information on the time difference $\alpha_\ell$ with neighbours in an undirected network. Balance is when each node can set a single time that is consistent with all the time differences with its neighbours. This was the original motivation for the ``group synchronisation'' approach to group graphs \citep{KEES03,GK06a,S11e}. 

\subsubsection*{Weighted signed graphs as voltage and as status}

The example of link labels drawn from the group of real numbers under addition, $\Gcal = (\Rbb,+)$, also gives us the connection to the example that explains why these group graphs are widely known as voltage graphs in the literature \citep{G74,Z82}.  
In this interpretation, the node labels $g_v$ are the voltages of a node in electrical circuit. 
The label $g_{vu}$  of a directed link $(v,u)$ from node $u$ to node $v$ is the difference in the voltages $g_v$ ($g_u$) of node $v$ ($u$) so $g_{vu} = g_v \times (g_u)^{-1} = g(v)-g(u)$.
For any walk, the label of the walk is simply equal to the sum of the voltage differences of links in the walk. 

When the currents and voltages of an electrical circuit are constant in time, then Kirchoff's laws apply. 
This includes the fact that the sum of voltage differences round any closed loop in the circuit has to be zero, i.e.\ the label of any closed semiwalk is the identity. So only a balanced group graph with labels drawn from the group of real numbers under addition obeys Kirchoff's laws and so only such balanced group graphs can represent a physical electrical circuit. 

An unbalanced voltage graph does not correspond to any static real physical situation. A circuit with voltages described by a set of voltage differences corresponding to an unbalanced group graph would not be stable and static.

However, there is a second interpretation of a group graph using the group of real numbers under addition, $\Gcal = (\Rbb,+)$ where unbalanced group graphs can exist and the underlying directed graph is not necessarily symmetrised. This is when we interpret the group labels in terms of status of individuals in social networks \citep{LHK10}. In this interpretation, a directed link $(v,u)$ from node $u$ to node $v$ with weight $g_{vu}$ is the difference in status, so $g_{vu} = g(v)\times(g(u))^{-1} = g(v)-g(u)$. Here a positive status difference $g_{vu}$ is the number of promotions needed by $u$ if they are to reach the same level as $v$ in the organisation. A negative status value would indicate that $u$ is above $v$ in the formal hierarchy.

If these individual levels, the node labels $g_v$, are defined by the formal hierarchy in the organisation, then this would be a balanced group graph, for example see \figref{f:status}. 
However, perhaps a social scientist has been interviewing members of an organisation and has asked how $u$ perceives the status of colleges they work with. 
Then it would be possible for $v$ to be aware of $u$ and so there can be a link between the two nodes even if $u$ was not under the direct or indirect control of $v$. 
In any case, the answers on the survey are unlikely to be perfectly aligned with the formal hierarchical structure and so the group graph of status need not be balanced as illustrated in \figref{f:status}.

Of particular note is that there is an important difference between the lack of an edge and an edge with value zero. In standard network analysis, a zero in the adjacency, an edge weight of value zero, is equivalent to having no edge. Here however, zero is the identity element of the group $e$ and indicates a significant social relationship but one where where there is no difference in the status of the individuals, quite different from a case where there is no working relationship, e.g.\ see \figref{f:status}. This ability to distinguish between no edge and an edge with zero weight is another key feature of group graphs that is missing from traditional network theory (see also \appref{as:adjmat}). 

\begin{figure}[htb]
	\begin{center}
			\includegraphics[width=0.8\textwidth]{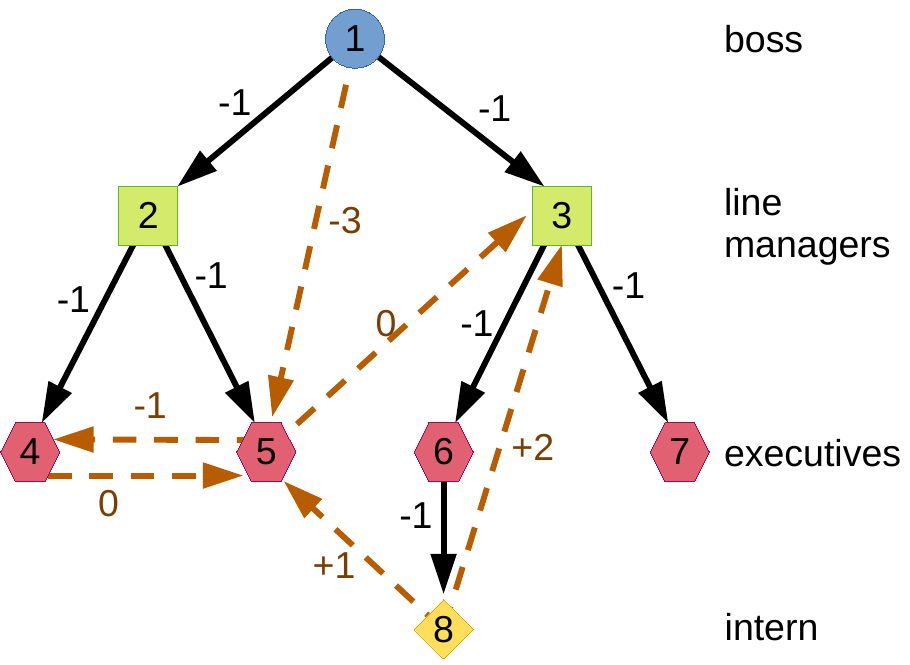}
		\end{center}
	\caption{An example of a weighted signed network representing status and perceived status in an organisation. The nodes are people in an organisation. The values on each link give the difference in status between two individuals. The solid black arrows show the formal hierarchy where links are from each supervisor to the people reporting to them.  These values are always $-1$ to indicate this is stepping down one level in the formal hierarchy. The dark orange dashed lines show perceived status differences as reported in a survey. This is a type of group graph where the link labels are elements of the abelian group defined by integers $\alpha_\ell$, the weights of each link $\ell$. The group `multiplication' rule is the addition operation on integers. Taken alone, the formal hierarchy shown with black arrows is a perfectly balanced group graph. However, if we combine that with the dark orange dashed links reported by the survey, we now have a complicated unbalanced group graph with cycles and semicycles. Note how an edge with zero weight between two individuals has some meaning which is different from the lack of an edge. A zero means there is a working relationship but it is between two people where one perceives the other to be of equal status such as how node 5 perceives node 3. Of course, the opinion of one person on another about relative status is not always reciprocated as the nodes $4$ and $5$ show.  }
	\label{f:status}
\end{figure}

\subsubsection*{Different types of negative edges}

Another use for the group graph construction is that it provides a more nuanced view of negative edges. Positive relationships in this view are always mapped to the identity. However, other elements of a group allow me to have different types of negative edge. For instance suppose I have $n$ competing blocks in a social network, friendship cliques in a school class, power blocks within a government. Suppose each of these blocks $\Bcal_m \subset \Ncal$, where $m \in \{0,1,\ldots,(n-1)\}$, is a block of a partition of the nodes in our social network. I assign a group element $a^m$ as the node label for nodes in the $m$-th block $\Bcal_m$. If I choose $a^n=e$ then our group is $Z_n$ where the elements of the group are $\{a^m | m \in \{0,1,\ldots,(n-1)\} \}$. The meaning of a link with link label $a^m$ is that it is from a node in block $p$ to a node in block $q=|m+p|_n$ so the extra information in the link labels of the negative links is that it records a more sophisticated type of inter-block relationship than allowed by the simple positive/negative labels used in signed networks. After all, the negative relationship between rival blocks may not always of the same type so why label them all as the same negative edge? An example is shown in \figref{f:weakbalance}. 
One faithful matrix representation of $Z_n$ is where $D(a^m) = \exp(2\pi i m/n)$ so I could also regard this construction in terms of opposing blocks of nodes as an example of a network with complex edge weights. 

In this case, a lack of balance could occur when people are assigning neighbours to blocks, as indicated by the link labels, in a manner that is not consistent with the block assignments made by others.

\begin{figure}[htb]
	\begin{center}
		\includegraphics[width=0.7\textwidth]{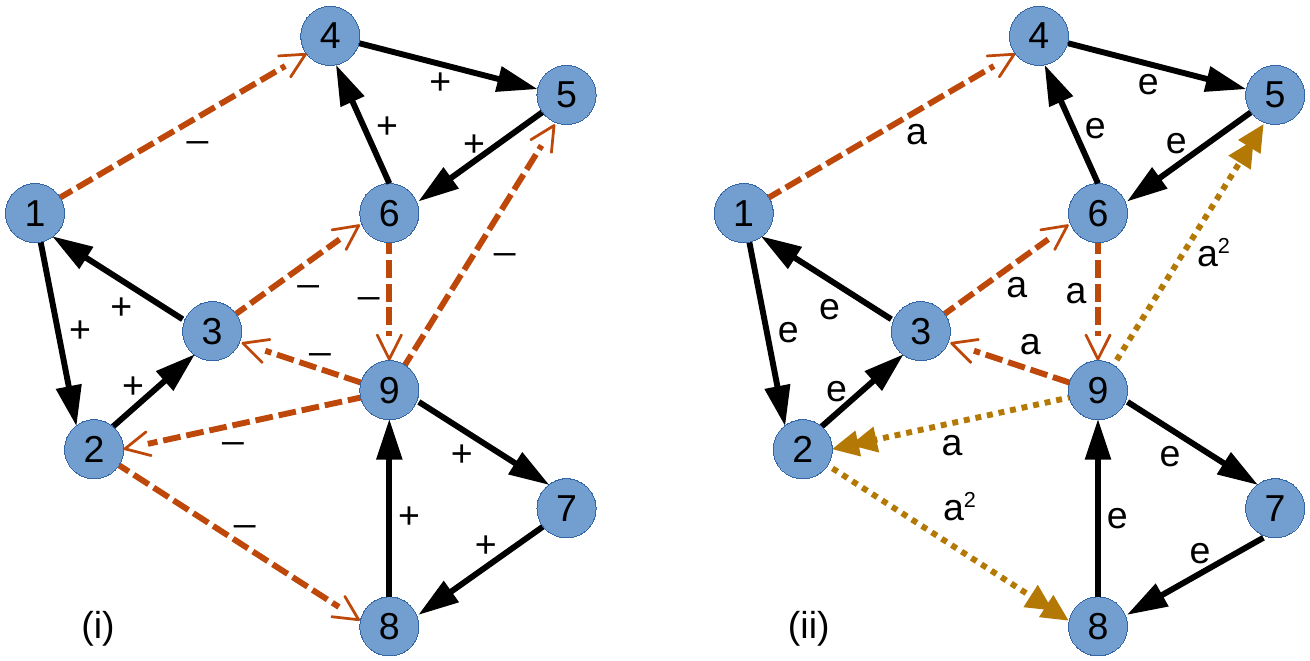}
	\end{center}
	\caption{This network has three blocks $\Bcal_m$ of nodes where the links between members of a block are always positive as denoted by $+1$ in (i) on the left and $e$ in (ii) on the right. The blocks are $\Bcal_0 = \{1,2,3\} $,  $\Bcal_1 =\{4,5,6\}$ and   $\Bcal_2 =\{7,8,9\}$. Relationships between nodes in different blocks are negative. In the signed graph (i) on the left, these have to be labelled with the same negative label $-1$. As a result, I see that I have an unbalanced signed graph in (i) because of the triangle of negative edges between nodes $3$, $6$ and $9$ in the centre. This type of signed network is sometimes said to be ``weakly balanced'' since in terms of social relations this might be considered to be a stable configuration. On the right in (ii) I show the same relations as a $Z_3$ group graph. Now I use two different labels $a$ and $a^2$ for the negative links while the positive links are labelled by $e$. For instance, node $9$ has three types of outgoing links as it has neighbours which are friends such as node $7$ as well as two different types of enemy neighbours, nodes $2$ and $5$. Here $\{e,a,a^2\}$ are elements of the group $Z_3$ with $a^3=e$ so I have a balanced $Z_3$ group graph in (ii). For instance, I may derive node labels $a^m$ for nodes in block $\Bcal_m$ in (ii) based on the link labels and the converse, these node labels may be used to give the link labels shown. The numbers inside the blue circles are used to identify nodes.}
	\label{f:weakbalance}
\end{figure}

This idea to use group graphs when there are several opposing blocks links up with the idea of \tsedef{clusterability} in signed networks and weak structural balance \citep{D67a} (see also \citet{GGLSS24} for more recent references). In this approach a signed network is described as being \tsedef{weakly balanced} as long as no closed walk has exactly one negative edge. If this is so, then the the signed network can be partitioned into a number of blocks, not just two blocks as allowed in a balanced network. As we've seen, I can always construct a group graph by assigning the same node labels to all nodes in the same clusters while maintaining distinct labels for nodes in different clusters. So it is possible to encode these cases with group graphs, e.g.\ using $Z_n$ group if there are $n$ blocks.

\subsubsection*{Beyond balanced group graphs}


When looking at signed networks, a common observation is that real systems are rarely balanced, even if there are pressures which push networks towards social balance, for example see \citet{ST10,SLT10}. 
It therefore makes sense to ask about \tsedef{unbalanced group graphs}, group graphs which are not balanced. Some properties of unbalanced group graphs are discussed in \appref{as:ubgg}.

The label for closed walks need not be the identity as it is for balanced group graphs. The label for a closed semiwalk that contains the same points can vary depending on which node is chosen as the start and end point. However, the labels given to a closed walk will always come from one conjugacy class of the group.

It is of particular interest to discuss networks that are `almost' balanced, to develop measures that show how close a network is to balance e.g.\ \cite{KCN19}. This can highlight the most efficient way to makes changes to a network in order to achieve balance or it can inform a model that tries to predict how a network will evolve over time.

The issue of measuring the lack of balance in a group graph is at the heart of the group synchronisation approach \citep{KEES03,GK06a,S11e,AF19}. In group synchronisation the aim is to find the labels $g_v$ of the nodes $v$. The link labels should be those of a balanced group graph, that is ${g}_v.({g}_u)^{-1}$ for link $(v,u)$ as given by \corref{c:lldecomp}. However, all one has are inaccurate noisy measurements $g_{vu}$ of the link label so the link labels do not give a balanced graph. The aim is to find a set of node labels $\hat{g}_v$, that will define a balanced group graph and will be a close approximation to the actual hidden values. These estimates $\hat{g}_v$ of the actual node labels ${g}_v$ are used to define link labels $\hat{g}_v.(\hat{g}_u)^{-1})$ in a balanced group graph using \corref{c:lldecomp} and then we optimise these estimated node values $\hat{g}_v$ to find ones which give link labels that are as ``close'' as possible to the measured values. 

So a major issue in group synchronisation is to find a quantitative measure of which  shows, for a given set of estimated node labels, how close the induced link labels are to the nosy measurements.  Essentially we are measuring the lack of balance in the group graph defined by the measured link labels. For instance, if the node labels $g_v$ are times of clocks which we have to find, and the link labels $g_{vu}$ are estimated time differences between nodes at the head and tail of any edge, we have a group graph based on the infinite abelian group $(\Rbb,+)$ and we could simply try to minimise $\sum_{(v,u)\in\Ecal} ( g_{vu} - \hat{g}_v.(\hat{g}_u)^{-1})^2$ to find the best estimates $\hat{g}_v$ of the node labels $g_v$ \citep{GK06a}.

For group graphs in a unitary representation, another way to measure imbalance in a group graph is to take inspiration from the ``\href{https://en.wikipedia.org/wiki/Wilson_action}{Wilson action}'' \citep{W74} used in the context of lattice gauge theories \citep{R12}. So I define a loss function $L$ where 
\beq
	L(\Ggraph,D) 
	= 
	\sum_{C \in \Ccal} 
	\mathrm{Re} 
	\Big[ \mathrm{Tr}\big(\Dmat(\Lambda(C))\big) \Big] 
	= 
	\half 
	\sum_{C \in \Ccal} 
	\left( 
	\Big[ \mathrm{Tr}\big(\Dmat(\Lambda(C))\big) \Big] 
	+
	\Big[ \mathrm{Tr}\big(\Dmat(\Lambda(\Crev))\big) \Big] 
	\right)\, .
	\label{e:wilson}
\eeq
Here I am summing over a set of set of closed semiwalks $\Ccal$. 
A useful reminder that we need to treat undirected networks as symmetric directed networks comes in the second form of the Wilson action where we can ensure a real result by summing over every cycle $C$ and its reversed partner $\Crev$ as defined in \defref{d:rsemiwalk}.

We can understand how this loss function $L$ measures the lack of balance as follows. The trace of a matrix is equal to the sum of the eigenvalues $\lambda_n$ of that matrix and there are $d$ eigenvalues for a $d \times d$ unitary matrix. All the eigenvalues of a unitary matrix have modulus one, $\lambda_n=\exp(i\theta_n)$. 
Thus the modulus of each trace can not be bigger than $d$ which happens only when all eigenvalues are $1$, i.e.\ for the unit matrix, the representation of the identity element. Hence I know that $L \leq d |\Ccal|$ where I have equality only if $\Lambda(C)=e$ for all closed semiwalks $C$ as $\Dmat(e)$ is always the unit matrix. 

It is sometimes helpful to rewrite this in the form
\beq
	 H(\Ggraph,D) 
	 = 
	 1 - 
	 \frac{1}{d |\mathcal{C}|} 
	 \sum_{C \in \Ccal} 
	 \mathrm{Re} \Big[ \mathrm{Tr} \big(\Dmat(\Lambda(C))\big)  \Big] 
	 \geq 0 
	 \, .
	 \label{e:Hdef}
\eeq   
For a group graph in a unitary representation, this $H$ has a minimum value of zero which is only achieved if the group graph is balanced. Thus $H$ provides a measure of the distance a group graph is from balance.

When the Wilson action is used in the context of lattice gauge theories, the network is invariably a hypercubic lattice (an undirected network). The obvious choice of basis there is to use closed walks of length four that takes us around each of the two-dimensional squares which are the faces of the hypercubes making up the lattice. More formally, these form a minimal cycle basis for this regular lattice.

For a general directed network there are many possible sets of closed semiwalks $\Ccal$ I could use but one sensible choice would be find a minimal set of basis semicycles $\Ccalsymm_\mathrm{b}$ out of which all other closed semiwalks could be built through concatenation. 

One necessary part of this basis set $\Ccalsymm_\mathrm{b}$ would be based on subset $\Ccalsymm_\mathrm{b2}$ of all semicycles of length two which are of the form $(u,v,u)$ where $u \neq v$, that is the set
\beq
 \Ccalsymm_\mathrm{b2} = 
 \big\{ 
 (u,v,u) | u,v, \in \Ncal, \;\; u \neq v, \;\; (u,v) \in \Lcal \; \text{or}\; (v,u) \in \Lcal   
 \big\} \, .
\eeq
The link labels of node pairs with only one edge in the original directed network $\Dcal$ are defined so they always satisfy balance so they give the maximum contribution of $d$ to the loss function $L(\Ggraph,D)$. It is the reciprocated links in the original network $\Dcal$ which can contribute to a lack of balance, e.g.\ in nodes $2$ and $3$ in \figref{f:ggex}, so these closed semicycles of length two can not be ignored. 

In a similar way, the set of self-loops $\Ccalsymm_\mathrm{b1} = \{ (u,u) | u \in \Ncal, (u,u) \in \Lcal \}$, semicycles of length one, must also be included. 

However for semicycles of length three or more, I can call upon traditional cycle analysis (unlike this paper, most other work defines cycles to be of length three or more). So for these closed semiwalks a sensible choice is to use a minimum directed cycle basis in the symmetrised directed graph $\Dcalsymm$ of \defref{d:symm} to reduce the influence of any one link label. This is because I can make any closed semiwalk of length three or more out of a combination of these cycles from the minimum cycle basis (for example see \citet{D05,BGV09,MM09a,VEE21} for information and references on closed walks and closed walk bases).
Note that you do not need to include a directed semicycle and its reverse, which is also a semicycle.

For instance, for a signed graph with $\pm1$ weights, i.e.\ a $Z_2$ group graph in the parity irrep, $L$ of \eqref{e:wilson} is simply equal to the number of closed walks $|\Ccal^{(+)}|$ with an even number of negative links (balanced closed walks),  minus the number of closed walks $|\Ccal^{(+)}|$ with an odd number of negative links (unbalanced closed walks), so $L=|\Ccal^{(+)}|- |\Ccal^{(-)}|=|\Ccal|- 2 |\Ccal^{(-)}| \leq |\Ccal^{(+)}|$ with equality only if I have a balanced signed network.

\subsubsection*{Beyond groups}

Since signed graphs are an example of a $Z_2$ group graph, it has been natural to look at how to use symmetry groups to generalise results for signed networks. However, for many processes on the graph, the type of unitary mixing often associated with groups may not always be appropriate. I have mentioned non-unitary representations associated with some infinite groups as one option. 

However, another direction would be to weaken the constraints imposed by the group structure by looking at monoids. Monoids are algebraic structures that share the same axioms as groups except that every element of a monoid does not have to have an inverse. A good example of a monoid would be a set of stochastic matrices with entries between zero and one inclusive and whose columns sum to one. One of the examples here of the $Z_2$ group used the two-dimensional representation of \eqref{e:z2rep2d}. These are also stochastic matrices but they are special cases where they are also unitary. I could also consider a link label $\Dmat(b)$ in a two-dimensional representation where every entry is $1/2$. This matrix $\Dmat(b)$ is not unitary, it has no inverse, but it is stochastic. This $\Dmat(b)$ matrix represents a process where values from each layer is shared equally between the two layers when they are transmitted along a link labelled by element $b$ in one time step. These three matrices, $\Dmat(e)$, $\Dmat(a)$ and $\Dmat(b)$, form a two-dimensional representation of a monoid with multiplication table
\beq
\begin{tabular}{c|ccc}
	       & $e$    & $a$    & $b$    \\ \hline
	       &        &        &        \\[\dimexpr-\normalbaselineskip+2pt]
	$e$    & $e$    & $a$    & $b$    \\
	$a$    & $a$    & $e$    & $b$    \\
	$b$    & $b$    & $b$    & $b$    \\
\end{tabular}
\eeq
These types of process described by stochastic matrices which are not in general reversible, hence the lack of a unitary representation, but they are often useful when representing processes occurring on a network. So monoid graphs rather than group graphs might be of interest.

We can also consider looking at structures where neither inverse nor identity elements are required, and so we just keep the associative property of a binary operation on a set of objects.  
This structure is known as a \tsedef{semi-group}, such as the set of real numbers between zero and one exclusive under multiplication $((0,1),\times)$. 
These have been discussed in the context of group synchronisation \citep{AF19}.


\subsection*{Acknowledgements}

I would also like to thank the researcher who highlighted the work on voltage graphs and gain graphs after this work was presented as part of a satellite meeting ``FRIENDS'' on signed networks at NetSci 2025. 

\subsection*{Competing interests}

The author declares that there are no competing interests.

\subsection*{Funding information} 

No funding has been received for this work.

%
%
%

\clearpage


\begin{center}
	\Large\textbf{Appendices}
\end{center}
\addcontentsline{toc}{section}{Appendices}

\appendix

\numberwithin{figure}{section}
\numberwithin{table}{section}
\numberwithin{theorem}{section}
\numberwithin{definition}{section}

\renewcommand{\thesection}{\Alph{section}}
\setcounter{section}{0}
\setcounter{definition}{0}
\setcounter{theorem}{0}
\renewcommand{\theHsection}{\Alph{section}}
\renewcommand{\theHequation}{\theHsection.\arabic{equation}}
\renewcommand{\theHfigure}{\theHsection.\arabic{figure}}
\renewcommand{\theHtable}{\theHsection.\arabic{table}}
\renewcommand{\theHtheorem}{\theHsection\arabic{theorem}}




\section{Consistent group graphs}\label{as:cnsstgg}

\subsection{Walks, consistency and group graphs}\label{as:wcgg}

When dealing with groups, the most interesting structures are those that are consistent with the group multiplication properties. A matrix representation $D$ of any abstract group is a good example of a structure where, by definition, the matrix multiplication respects the structure of the abstract group. 
This idea was used for group graphs to define balance but the network structure used was a \emph{semiwalk} not the more usual \emph{walk} or path used in network analysis. Here I will look at what happens if we define a type of group graph where the group labels on walks are consistent with the group structure. 

\begin{definition}[A consistent group graph] \label{d:consistentgg}
	A group graph is \tsedef{consistent} when the group element of any walk only depends on the initial and final nodes in the walk, i.e.\ it is independent of the route taken. 
\end{definition}

A property that follows immediately from the definition of a consistent group graph in \defref{d:consistentgg} is that there is a unique group element $g_{vu}=\linkmap(W(v,u))$ for any node pair where there is a walk $W(v,u)$ from $u$ to $v$.  Links are a special case of this.

Clearly, walks are a special case of semiwalks so any group graph that is balanced is also consistent but the converse is not guaranteed. Thus the two examples of balanced group graphs in \figref{f:bggex} are also consistent group graphs. 

It requires an explicit check but the unbalanced group graphs of \figref{f:ggex} turn out to be \tsedef{inconsistent} group graphs (i.e.\ they are not consistent group graphs).
In the $Z_2$ group graph  \figref{f:ggex}.i, the walk $(6,5,4,3,2)$ from node $2$ to node $6$ has a walk label of $e.e.a.a = e$ but the walk of length one between the same two nodes, $(6,2)$, is labelled by different $Z_2$ element, $a$. The same walks in the $S_3$ graph in \figref{f:ggex}.ii give a walk label of $b.a^2b.a.a^2=a$ for the walk $(6,5,4,3,2)$ while the walk $(6,2)$ has a label given by the $S_3$ element $b$.

On the other hand, the two examples in \figref{af:ggcnstntex} are both consistent but unbalanced group graphs.  

\begin{figure}[htb]
	\begin{center}
		\includegraphics[width=0.9\textwidth]{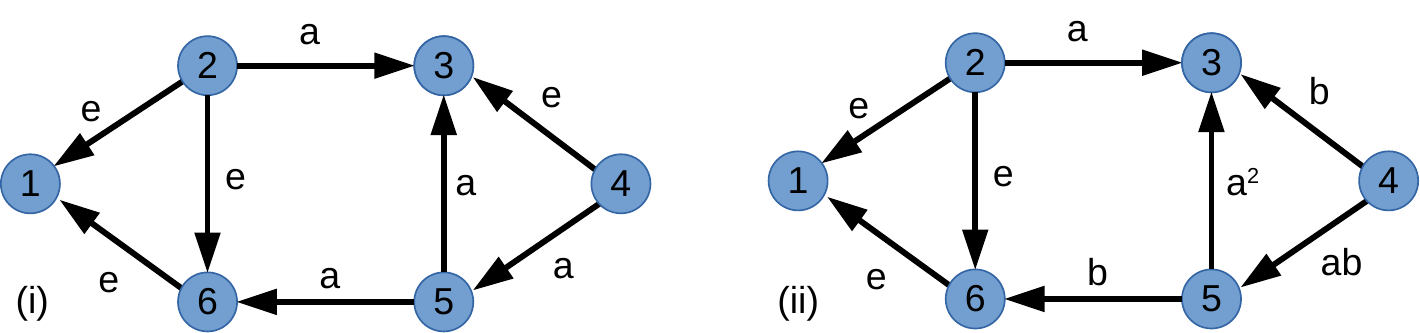}
	\end{center}
	\caption{Two examples of consistent but unbalanced group graphs. The group graph (i) on the left is has links lying in $Z_2$ 
		so this is also a signed network where positive (negative) links are labelled with $e$ ($a$). In group graph (ii) on the right links lie in $S_3$. The examples here are both consistent but unbalanced group graphs. The lack of balance can be seen in the odd number of links containing $a$ in (i) or $b$ in (ii). The numbers inside the blue circles are used to identify nodes in examples in the text. Group elements for each link are shown next to the relevant link. Both graphs are unrooted, see \defref{d:rootcomp}, since there is no path between nodes $2$ and $4$.}
	\label{af:ggcnstntex}
\end{figure}

The definition of the label of a walk was given in \defref{d:walkgpelm}. Some of the properties for semiwalks on a balanced group graph have a corresponding result for walks on a consistent group graph.

\begin{lemma}[Closed walks in consistent group graphs] \label{al:closedwalkid}
	Any closed walk $\Ccal$ in a consistent group graph has a label equal to the identity element. 
	\beq
	\linkmap(\Ccal) =e \, .
	\label{ae:closedwalkid}
	\eeq
\end{lemma}
This lemma is the equivalent of \lemref{l:cswbggid} for semiwalks on balanced group graphs and both lemmas may be proved in the same way. 
\begin{proof}
	Consider a non-trivial closed walk $\Ccal = (v, \ldots, v)$, a walk starting and ending at some node $v$ of length one or more. We can concatenate this to give a distinct closed walk $\Ccal' = \Ccal\circ \Ccal$ by \eqref{e:walkconcat}. This will also start and end at node $v$ so $\Ccal'$ is also a closed walk. Using \eqref{e:concatwalklabel} then gives for the labels of the walks that $\linkmap(\Ccal') = \linkmap( \Ccal) \times \linkmap( \Ccal)$. Since the graph is consistent, and the walks $\Ccal$ and $\Ccal'$ are between the same nodes, they must have the same group label in a consistent group graph so that $\linkmap( \Ccal) =\linkmap(\Ccal')$. Given that all group elements have an inverse, we may left multiply by $(\linkmap( \Ccal))^{-1}$ to give that $\linkmap(\Ccal) =e$ as required.
	\\
	For a trivial walk $\Wcal_v=(v)$ of length zero from any node $v$ we have that $\Wcal_v \circ \Wcal_v = \Wcal_v$. From the definition of the label of a walk in \eqref{e:concatwalklabel} we have that $\Lambda(\Wcal) \times \Lambda(\Wcal_v) = \Lambda(\Wcal)$. Left multiplying by $(\Lambda(\Wcal))^{-1}$ gives $\linkmap(\Wcal_v)=e$.
\end{proof}

There are special cases of \lemref{al:closedwalkid} worth noting. Any cycle and so any self loop are labelled by the identity in a consistent group graph. Likewise, the label of any reciprocated links must be the inverse of each other in a consistent group graph. 
\begin{lemma}[Labels of reciprocated links] \label{al:reciplabel}
	The label of any reciprocated links in a consistent group graph must be the inverse of each other. 
	\beq
	\linkmap((u,v)) = \big(\linkmap((u,v))\big)^{-1} \quad \text{if} \quad (u,v) , (v,u) \in \Lcal \, .
	\eeq
\end{lemma}
The converse of this lemma is not true. That is just because all reciprocated links carry labels which are inverse of each other does not mean a group graph is consistent.

It is important to consider how the definition of a consistent group graph, \defref{d:consistentgg}, works when there are some node pairs in the network where there is only one walk between those nodes. In such a case, requiring the group graph to be consistent provides no constraint on the group labels of the links in these paths. The most extreme application of this case is a directed tree, such as shown in \figref{af:closedwalkcnstcyex}.iii, where there will never be an alternative to any walk passing through any one link. Thus, any assignment of group labels to the links in \figref{af:closedwalkcnstcyex}.iii always gives a consistent group graph. 

However, there are cases where there is only one \emph{path} between nodes (walks made from distinct nodes) but multiple walks. For example, the case of a graph made up of a single directed closed walk, such as in \figref{af:closedwalkcnstcyex}.i, only has one path between any pair of nodes.  There are though multiple walks as a walk can go around the cycles of the three nodes many times before finishing the walk at the target node.  As there are multiple walks in this case between any node pair, the definition of a consistent group graph will put constraints on the group labels as shown by the lemmas on closed walks such as \lemref{al:closedwalkid}.

From this it also follows that the converse of \lemref{al:closedwalkid} is not true. That is, if we can show that all closed walks, including cycles, have a label equal to the identity, this is not sufficient to show that a group graph is consistent. A counter example is provided in \figref{af:closedwalkcnstcyex}.ii where there are no cycles yet consistency does constrain the group labels and we only have consistency if $ba=c$.

\begin{figure}[htb]
	\begin{center}
		\includegraphics[width=0.45\textwidth]{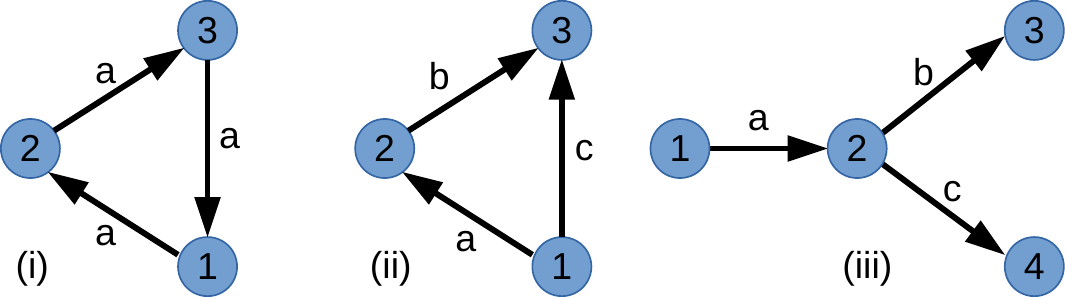}
	\end{center}
	\caption{Examples for \lemref{al:closedwalkid}. On the left, network (i) shows a simple directed closed walk which is strongly connected but with no reciprocated links. There is also only one path between any pair of nodes. The walks are the basic path concatenated with any number of the closed walks. If in (i) $a^3=e$ (e.g.\ for $S_3$) then this is a consistent group graph otherwise it is not (e.g.\ when the link labels come from the group $Z_2$). In the centre, network (ii) is a directed graph with no closed walks and which has no reciprocated links. The two walks from node $1$ to node $3$ are consistent only if $b \times a=c$. There are no closed walks in this example (ii) so \lemref{al:closedwalkid} puts no constraints on the labels showing that the converse of \lemref{al:closedwalkid} is not true. 
		A directed tree is shown on the right in network (iii) in which there is only one walk between any pair of connected nodes and there are no closed walks. This is a consistent and balanced group graph whatever group elements label the links and for any group. 
	}
	\label{af:closedwalkcnstcyex}
\end{figure}

Finally, given that semiwalks on group graphs are equivalent to walks on the symmetrised group graph, we can give an alternative definition of a balanced group graph.
\begin{definition}[Balanced Group Graph] \label{d:bgrpgrphsgg}
	A \tsedef{balanced group graph} is a group graph which has a consistent symmetrised group graph.
\end{definition}

Since all walks are also semiwalks on the group graph, but not vice versa, it is clear that balance imposes more constraints on the link labels of a group graph than consistency. It follows that a balanced group graph is always a consistent group graph, but the converse does not always hold as \figref{af:ggcnstntex} shows. 
	
	\subsection{Node Labels}\label{as:nodelabels}
	
	
	Balance in a group graph is a powerful constraint and it confers, amongst other properties, the ability to give nodes a well defined and consistent label along side the labels on the links. The consistency property of group graphs given in \defref{d:consistentgg} is a weaker constraint than balance so it should be of no surprise that it is not possible to assign node labels in the same way as with a balanced group graph. If, however, the underlying graph has a more constrained topology, then it turns out that along with consistency is enough to allow node labels to be assigned. This is the topic of this section. For simplicity, I will work with a network with a single weakly connected component.

	Working with directed graphs, what is required is that for each component there is at least one node that is connected to every other node by a walk.
	\begin{definition}[Source and sink nodes]\label{d:sourcesink}
		A \tsedef{source node} of a weakly connected component to be a node where there is a walk \emph{from} the source node to every other node in the component. Likewise, a \tsedef{sink node} of a weakly connected component is one where there is a walk from every other node in the component \emph{to} the sink node.  
	\end{definition}
	For a network with a single weakly connected component, the set of source nodes $\Ncal^{(source)}$ and sink nodes $\Ncal^{(sink)} $ may be defined formally as
	\beq
	\Ncal^{(source)} 
	=
	\big\{
	v\, | \, \forall u \in \Ncal \; \exists \; \Wcal(u,v) 
	\big\} 
	\, , \quad
	\Ncal^{(sink)} 
	=
	\big\{
	v\, | \, \forall u \in \Ncal \; \exists \; \Wcal(v,u)
	\big\} 
	\label{e:sourcesinkdef}
	\eeq
	A reminder that $\Wcal(v,u)$ is a walk that starts from node $u$ and ends at node $v$. The notation implies there is at least one such walk but there can also be many such walks.
	
	\begin{definition}[Rooted component and graph]\label{d:rootcomp}
		A weakly connected component with at least one source or at least one sink node is a \tsedef{rooted component}. A \tsedef{rooted graph} is one where all weakly connected components are rooted.
	\end{definition}
	In a strongly connected component, every node is both a source node and a sink node. 
	
	Group graphs have a group element assigned as a label to each link. When a group graph is consistent, then for any rooted component there is also a natural way to label the nodes with a group element. This can be summarized by the following lemma. 
	\begin{lemma}[Node labels in a consistent group graph] \label{al:nodelabelrcgg}
		For a consistent rooted group graph, the nodes may be labelled by group elements in a way that is consistent with the link labels. That is I may assign $g_v\in \Gcal$ for all $v \in \Ncal$ s.t.\ $g_v = \linkmap(\Wcal(v,u)) g_u$ for any walk $\Wcal(v,u)$ from node $u$ to node $v$. 
	\end{lemma}
	Note this includes the special case that for any single link $(v,u)$ (a walk of length one) with link label $g_{vu}\in \Gcal$, I have that the relationship between the labels of the tail node $u$ and the head node $v$ is $g_v = g_{vu} g_u$. I can prove this lemma by giving an explicit method to construct these node labels.
	\begin{proof}
		We can treat each weakly-connected component independently. The algorithm given here is simply repeated for each weakly connected component. So without loss of generality we assume our network $\Dcal$ has a single weakly connected component. 
		
		Pick any one source node $r \in \Ncal^\mathrm{(source)} \subseteq \Ncal$ as the root node. This root node is assigned the group element $g_r \in \Gcal$ as the node label for $r$. This can be any element of the group.
		Now assign the node labels to be the group elements $g_v = g(\Wcal(v,r)) g_r$ where $\Wcal(v,r)$ is any walk from $r$ to $v$. Such a walk always exists  for all nodes $v$ by definition \eqref{e:sourcesinkdef} of a source node. The group label of any walk depends only on the initial and final nodes in a consistent group graph.  Thus $g_v$ is fixed and unique given $g_r$ provided we have at least one source node.
		
		If there is no source node, then there is at least one sink node since the graph is rooted. Choose any sink node $r \in \Ncal^\mathrm{(sink)} \subseteq \Ncal$ as the root node. The node label for $r$ is assigned to be the group element $g_r \in \Gcal$ which can be any element of the group.
		Now assign the node labels to be the group elements such that $g_v = g_r g(\Wcal(r,v)) $ where $\Wcal(r,v)$ is any walk from $v$ to $r$. Such a walk always exists for a sink node for all nodes $v$, see \eqref{e:sourcesinkdef}. The group label of any walk depends only on the initial and final nodes in a consistent group graph.  Thus $g_v$ is fixed and unique given $g_r$ provided we have at least one sink node.
	\end{proof}
	
	So we have now shown that a rooted consistent group graph has a set of node labels which are consistent with the link labels.  At this point we can call upon the results for balanced group graphs to deduce that we have in fact created a balanced group graph. 
	\begin{lemma}[Rooted consistent group graphs are balanced] \label{al:rcggbal}
		A consistent rooted group graph is a balanced group graph. 
	\end{lemma}
	In general, consistent groups graphs are not balanced, as \figref{af:ggcnstntex} shows, but the rooted property is sufficient to ensure that consistency in the walks implies balance defined in terms of the semiwalks. The properties of consistent rooted group graphs then follow from the earlier analysis of balanced group graphs.

	\section{Unbalanced group graphs}\label{as:ubgg}
	
	One property of unbalanced group graphs uses the concept of conjugacy from group theory. Two group elements, say  $g_1,g_2 \in \Gcal$ are said to be \tsedef{conjugate} if there exists a third group element $h\in\Gcal$ such that $g_1 = h\times g_2 \times h^{-1}$.  This then defines \tsedef{conjugacy classes}. That is a partition of the set of elements into blocks known as conjugacy classes where each class contains all the elements that are conjugate to each other. For instance the classes for $Z_2$ are $C_e = \{e\}$ and $C_a = \{a\}$ while the classes for $S_3$ are $C_e = \{e\}$,  $C_a = \{a,a^2\}$ $C_b = \{b,ab,a^2b\}$.

	\begin{lemma}[Cycles Lie in Conjugacy classes]
		Each cycle in a group graph, consistent or not (balanced or unbalanced), lies in a conjugacy class of the group.
	\end{lemma}
	If you define the product of group elements round a cycle it will be different for each start/end-point but only difference by a similarity transformation so the group elements are equivalent up to this transformation i.e.\ they always lie in one conjugacy class.
	\begin{proof}
		Suppose we pick a node $v_0$ as the start (and end) of our cycle $\Ccal$ and we denote this walk as $\Ccal_0=(v_n)_{n=0}^{L}$. Then the group element $g(\Ccal_0)$ of this walk/cycle $\Ccal_0$ is the product of the group elements labelling the links in the walk in the order in which they appear in the cycle, see \defref{d:walkgpelm}. Remembering that the group elements are ordered from the first link on the right, last on the left, we have that
		\bea
		g(\Ccal_0) 
		&=& 
		g_L \times g_{L-1} \times \ldots \times g_{2} \times g_1  \, ,  
		\, ,  
		\quad 
		\link_n = (v_{n},v_{n-1}) 
		\, ,
		\quad  
		g_n = g(\link_n) 
		\, .
		\eea
		We can repeat this for the same cycle $\Ccal$ but now starting (and ending) from $v_1$. We denote this walk as $\Ccal_1 = (u_n)_{n=0}^{L}$ where $u_n = v_{n+1}$ for $n<L$, $u_L=u_0=v_1$.
		The group element with this version of the cycle is 
		\bea
		g(\Ccal_1) 
		&=& 
		g_1 \times g_{L}  \times g_{L-1} \times \ldots \times g_{2}  
		\\
		&=& 
		g_1 \times g(\Ccal_0) \times (g_{1})^{-1}  \, ,  
		\eea
		That is the two walks representing the cycle starting from neighbouring points differ by a similarity transformation. We can repeat this for all neighbouring nodes in the cycle to see that $g(\Ccal_n)  \sim g(\Ccal_{n+1})$ and from transitivity of the equivalence relation $\sim$ defined by similarity transformations, we see that every distinct walk representing a cycle can be a different group element but they must all be equivalent. This is the definition of a conjugacy class so each distinct cycle is characterised by one of the conjugacy classes of the group.  
	\end{proof}

	\section{The switching transformation}\label{as:switching}

	The \tsedef{switching transformation} \citep{AR68,Z82,Z89,CW86,CDD21} transforms one group graph $(\Dcal,G,\linkmap)$ into a second $(\Dcal,G,\linkmaptilde)$ with different link labels in such a way that it preserves the property of balance. More generally, this transformation preserves the dynamics, the time evolution, of the type of processes considered later since the network topology is unchanged. However, this does change other properties of group graphs and signed networks which are used for physical properties of networks in some practical examples. So it remains unclear how useful these switching transformations are in practice.

	\begin{definition}[Switching]\label{d:switching}
		A \tsedef{switching transformation} is a map from one group graph $(\Dcal,G,\linkmap)$ to a second $(\Dcal,G,\linkmaptilde)$ where $\linkmaptilde((v,u)) := \gamma(v)\linkmap((v,u)) (\gamma(u))^{-1}$ for each link $(v,u)$.
		Here $\gamma(v)$ is any set of group elements associated with each node $v$, so $\gamma$ is any arbitrary map $\gamma:\Ncal \to G$. The group graphs $(\Dcal,G,\linkmap)$ and $(\Dcal,G,\linkmaptilde)$ are said to be \tsedef{switching equivalent}.
	\end{definition}
	If $\gamma(v)=g$ is some fixed group element $g$ for all nodes then this is using the ``group conjugation'' transformation (see \appref{as:ubgg}). If we have a balanced group graph then we can use node labels $g(v)$ as the node switching maps $\gamma(v)$, a case we will consider below. However, the node map  $\gamma$ for switching transformations is much more general as it can associate \emph{any} group element to \emph{any} node. Thus there are a large number ($|\Gcal|^N$ where $N={|\Ncal|}$) of such switching transformations. 
	
	\begin{lemma}[Semiwalk switching]\label{l:swalkswitch}
		If two group graphs $\Ggraph=(\Dcal,G,\linkmap)$ and $\tilde{\Ggraph}=(\Dcal,G,\linkmaptilde)$ are switching equivalent then then the labels $\linkmap(\Scal)$ and $\linkmaptilde(\Scal)$ for any semiwalk $\Scal$ running from node $s$ to node $t$ are also switching related, i.e.\ $\linkmaptilde(\Scal) := \gamma(t)\linkmap(\Scal) (\gamma(s))^{-1}$.
	\end{lemma}
	Since walks are examples of semiwalks, this lemma also holds for walks. 
	The proof works in the same way as that for \lemref{l:nltocgg}.
	\begin{proof}
		Suppose I have a semiwalk $\Scal = (v_n)_{n=0}^L$ from node $v_0$ to node $v_L$. Let the switching map be $\gamma(v) \equiv \gamma_v \in G$ for each node $v$. Then label of the semiwalk in the group graph $(\Dcal,G,\linkmap)$ is 
		\bea
		\linkmap(\Scal) 
		&=&
		\linkmap(v_L,v_{L-1}) \times 
		\linkmap(v_{L-1},v_{L-2}) \times 
		\ldots \times 
		\linkmap(v_{2},v_{1}) \times 
		\linkmap(v_{1},v_{0})    
		\eea
		Inserting $e = \gamma(v_j)\big(\gamma(v_j)\big)^{-1}$ between the link labels for the links starting/ending with node $v_j$ gives
		\bea
		\gamma(v_n) \times  
		\linkmap(\Scal) \times 
		\big(\gamma(v_0)\big)^{-1}
		&=&
		\gamma(v_n) \times 
		\linkmap(v_L,v_{L-1}) \times 
		\gamma(v_{L-1})\big(\gamma(v_{L-1})\big)^{-1} 
		\nnel
		&&
		\times
		\linkmap(v_{L-1},v_{L-2}) \times  
		\gamma(v_{L-2})\big(\gamma(v_{L-2})\big)^{-1} 
		\nnel
		&&
		\times
		\ldots \times 
		\linkmap(v_{2},v_{1}) \times  
		\gamma(v_1)\big(\gamma(v_1)\big)^{-1} 
		\nnel
		&&
		\quad \quad \quad \quad
		\times \linkmap(v_{1},v_{0}) \times
		\big(\gamma(v_0)\big)^{-1}
		\\
		&=&
		\linkmaptilde(v_L,v_{L-1}) \times 
		\linkmaptilde(v_{L-1},v_{L-2}) \times 
		\nnel
		&&
		\ldots \times 
		\linkmaptilde(v_{2},v_{1}) \times 
		\linkmaptilde(v_{1},v_{0}) 
		\\
		&=&
		\linkmaptilde(\Scal) 
		\eea
	\end{proof}
	From \lemref{l:swalkswitch} we can decuce that the balance property is preserved under any switching transformation
	\begin{lemma}[Balance and switching]\label{l:sbalancewitch}
		If two group graphs $\Ggraph=(\Dcal,G,\linkmap)$ and $\tilde{\Ggraph}=(\Dcal,G,\linkmaptilde)$ are switching equivalent then they are either both balanced or they are both unbalanced.
	\end{lemma}
	\begin{proof}
		Suppose two semiwalks $\Scal$ and $\Scal'$ from the same node $s$ to the same node $t$ have the same label  $\linkmap(\Scal)=\linkmap(\Scal')$ in the original group graph $\Ggraph=(\Dcal,G,\linkmap)$. 
		Then the switched label $\linkmaptilde(\Scal)$ of the first semiwalk $\Scal$ becomes 
		$\linkmaptilde(\Scal) = \gamma(t) \linkmap(\Scal)(\gamma(s))^{-1}$. In the same way, 	
		the switched label $\linkmaptilde(\Scal')$ of the second semiwalk $\Scal'$ is
		$\linkmaptilde(\Scal') = \gamma(t) \linkmap(\Scal')(\gamma(s))^{-1}$.
		
		However $\linkmap(\Scal)=\linkmap(\Scal')$ for any semiwalks sharing the same starting and end nodes if the original group graph $\Ggraph$ is balanced.  Hence this means that all semiwalks sharing the same starting and end nodes in the switching equivalent graph $\tilde{\Ggraph}$ also always have the same label and so, by \defref{d:bggsw} of a balanced group graph the  switching equivalent graph $\tilde{\Ggraph}$ is balanced if the original group graph is balanced.
		
		The converse can be shown as follows. Suppose two semiwalks $\Scal$ and $\Scal'$ from the same node $s$ to the same node $t$ have a different label  $\linkmap(\Scal) \neq \linkmap(\Scal')$ in the original group graph $\Ggraph=(\Dcal,G,\linkmap)$. 
		The switched labels of these two semiwalks are still $\linkmaptilde(\Scal) = \gamma(t) \linkmap(\Scal)(\gamma(s))^{-1}$ and $\linkmaptilde(\Scal') = \gamma(t) \linkmap(\Scal')(\gamma(s))^{-1}$. However, if $\linkmaptilde(\Scal)=\linkmaptilde(\Scal')$ then  $\gamma(t) \linkmap(\Scal)(\gamma(s))^{-1}=\gamma(t) \linkmap(\Scal')(\gamma(s))^{-1}$. Pre-multiplying by $(\gamma(t) )^{-1}$ and post-multiplying by $\gamma(s)$ gives $\linkmap(\Scal)=\linkmap(\Scal')$ which contradicts our assumption. So  $\linkmaptilde(\Scal) \neq \linkmaptilde(\Scal')$ if $\linkmap(\Scal) \neq \linkmap(\Scal')$. Thus unbalanced group graphs are only ever switching equivalent to other unbalanced group graphs.	
	\end{proof}

	We noted earlier that we can use the node labels of a balanced group graph to define a switching transformation. We know from \corref{c:lldecomp} that $\linkmap((v,u)) = g(v)  (g(u))^{-1}$ using the node labels $g(v)$ of a balanced group graph. But if we consider the switching transformation with $\gamma(v)=g(v)$ then $\linkmaptilde((v,u)) = \gamma(v). g(v)  (g(u))^{-1} (\gamma(u))^{-1} = e$. Formally
	\begin{lemma}[Node labels and switching equivalence]\label{l:nodelabelswitch}
		If a group graph $\Ggraph=(\Dcal,G,\linkmap)$ is balanced then the node labels $g_v$ give the switching transformation with $\gamma(v) = g(v)$ such that all link labels $\linkmaptilde(\ell)$ are the identity in the new group graph $\tilde{\Ggraph}=(\Dcal,G,\linkmaptilde)$ after the switching transformation.
	\end{lemma}
	A corollary of this lemma is the form found in the literature, for example \citep{Z89,CDD21}.
	\begin{corollary}[Balance and switching equivalence]\label{c:balswitch}
		If a group graph $\Ggraph=(\Dcal,G,\linkmap)$ is switching equivalent to a graph where all link labels are the identity, then the group graph is balanced. That is if there exists $\gamma: \Ncal \to G$ such that
		$\gamma(v)\linkmap((v,u)) (\gamma(u))^{-1} = e$ for all links $(v,u) \in \Lcal$ then $\Ggraph$ is balanced.
	\end{corollary}

	\Lemref{l:swalkswitch} tells us is that whenever we are only interested in the labels of walks or semiwalks on a group graph, then there is essentially no difference between a very large number of group graphs with distinct link labels, those $|\Gcal|^{N}$ ($N={|\Ncal|}$) group graphs related by switching transformations. In particular, for the dynamical processes discussed below, these group graphs will all have the same time evolution as that is controlled by the topology of the network, not by the group structure. Essentially this is as if each switching transformation is changing the coordinates used at each node for some dynamical process but does not changes the actual dynamics.
	
	Given that so many group graphs are related to one another by switching transformations, it is not surprising to learn that these switching transformation do destroy some structure which is of interest.  From \lemref{l:nodelabelswitch} and \corref{c:balswitch} we deduce that every switching equivalent group graph with a distinct set of link labels needs a different switching transformation to reach the group graph with all link labels equal to the identity, and so has a different set of node labels. While we have seen there is a simple relationship between $|G|$ such sets of node labels, that still leaves $|\Gcal|^{N-1}$ distinct sets of node labels and so $|\Gcal|^{N-1}$ distinct balanced group graphs. One way to see how these switching equivalent balanced group graphs are different is to look at the example in \figref{f:bggswVL}. This shows how one one switching transinformation applied to the balanced $S_3$ group graph of \figref{f:bggex}.ii produces a completely different set of blocks, nodes communities, as defined by the node labels.  Since these often correspond to meaningful groups, such as groups of mutual friends, the fact that switching transformations usually completely change these blacks shows that they may not be very relevant in some practical situations.
	
	\begin{figure}[htb]
		\begin{center}
			\includegraphics[width=0.9\textwidth]{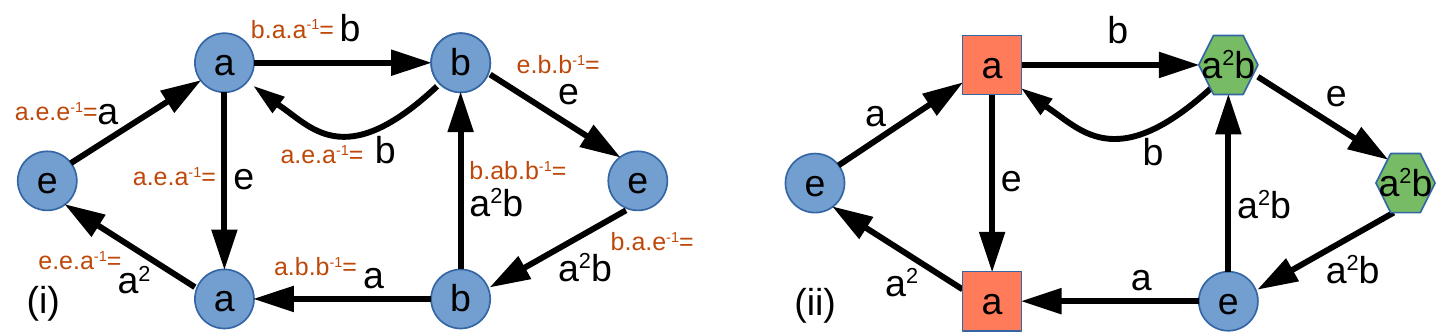}
		\end{center}
		\caption{A balanced $S_3$ group graph which is switching equivalent to that of \figref{f:bggex}.ii. On the left in (i) the symbol inside the nodes shows the group element $\gamma(v)$ used for each node in the switching transformation. The switching transformation $\linkmaptilde((v,u)) := \gamma(v)\linkmap((v,u)) (\gamma(u))^{-1}$ for each link is shown in a small red font next to the current link label $\linkmaptilde((v,u))$ shown in black beside each link. For instance $\linkmaptilde((3,5)) = \gamma(3)\linkmap((5,3)) (\gamma(u))^{-1} = b.ab.(b)^{-1} =a^2b$ using the node numbering from \figref{f:bggex}.ii. Note that this remains a balanced group graph after the switching transformation from \figref{f:bggex}.ii illustrating \lemref{l:sbalancewitch}.
			\\
			On the right in (ii) is the same group graph as (i) but now the node labels $g_v$ of \lemref{l:nodelabel} are shown inside each node symbol. The node shape and colour in (ii) indicates membership of one particular block of the node partition which respects the group structure. Note that the node partition of the switched group graph here is completely different from that found in the original group graph as shown in \figref{f:bggexVL}.ii or \figref{f:ggcnstntexVLrl}.ii. 
		}
		\label{f:bggswVL}
	\end{figure}

	\section{Stochastic matrices and monoids}\label{as:stochmatmonoid}
	
	\subsection{Monoids}\label{as:monoid}
	
	Groups are useful when we have a reversible process on a network because every group element has a unique inverse which is used to represent this reversibility in group graphs. However, many processes of interest on a graph are not reversible such as diffusion. If we drop the requirement for an inverse from the axioms that define a group, we have what is known as a monoid.
	
	\begin{definition}[A monoid]
		A \tsedef{monoid} $\Mcal$ is a set of elements and a `multiplication' law $\times : \Mcal \to \Mcal$ that obeys three axioms: closure ($g_1 \times g_2 \in \Mcal$), associativity ($g_1 \times (g_2 \times g_3) = (g_1 \times g_2) \times g_3$) and identity ($\exists \; e \in \Mcal$ s.t.\ $e \times g = g \times e =g$) where these hold for any elements $g,g_1,g_2,g_3 \in \Mcal$. 
	\end{definition}
	
	The lack of an inverse means monoids are associated with processes with no reverse, always decreasing or always moving forward. An example of a representation of a monoid is a set of Markovian matrices $\Mcal = \{ \Pmat | \sum_{i} P_{ij} =1 \}$ under matrix multiplication. The unit matrix is Markovian and it represents the identity element, $\Dmat(e) = \unitmat$. Associativity comes with matrix multiplication. 
	
	Closure comes as follows. Let $\Pmat(a)$ and $\Pmat(b)$ be Markovian matrices representing monoid elements $a,b \in \Mcal$. Then we have that
	\beq
	\sum_i [\Pmat(ab)]_{ik}
	=
	\sum_i [\Pmat(a)\Pmat(b)]_{ik}
	=
	\sum_i \sum_k [\Pmat(a)]_{ik} [\Pmat(b)]_{kj}
	=
	\sum_k  [\Pmat(b)]_{kj}
	=1
	\eeq
	Thus $\Pmat(ab)=\Pmat(a)\Pmat(b)$ is also Markovian and we have closure.
	
	The inverse matrix is not usually Markovian (the identity matrix shows some inverses are in the set) so there is no inverse in general and so these matrices do not form a group.
	
	An simple example of a monoid of three elements is a simple extension of the $Z_2$ group. This is defined by the multiplication table
	\beq
	\begin{tabular}{c|ccc}
		& $e$    & $a$    & $b$    \\ \hline
		&        &        &        \\[\dimexpr-\normalbaselineskip+2pt]
		$e$    & $e$    & $a$    & $b$    \\
		$a$    & $a$    & $e$  & $b$    \\
		$b$    & $b$    & $b$    & $b$    \\
	\end{tabular}
	\label{ae:monoid3}
	\eeq
	One two-dimensional representation of this monoid \eqref{ae:monoid3} in terms of three stochastic matrices is
	\beq
	\Dmat(e) =
	\begin{pmatrix}
		1 & 0 \\
		0 & 1
	\end{pmatrix}
	\, ,
	\quad
	\Dmat(a) =
	\begin{pmatrix}
		0 & 1 \\
		1 & 0
	\end{pmatrix}
	\, ,
	\quad
	\Dmat(b) =
	\begin{pmatrix}
		1/2 & 1/2 \\
		1/2 & 1/2 
	\end{pmatrix}
	\, .
	\label{ae:monoid3rep2d}
	\eeq
	Note in particular that $\Dmat(b)$ is not unitary and it is not even invertible.
	These types of process described by general stochastic matrices are not reversible, hence the lack of a unitary representation, but they are often useful when representing processes occurring on a network. So monoid graphs rather than group graphs might be of interest.
	
	\subsection{Stochastic matrices, monoids and groups}\label{as:stochmat}
	
	A stochastic matrix $\Smat$ is a square matrix whose entries lie between zero and one, $0 \leq S_{ij}\leq 1$ and whose columns add to one $\sum_{i=1}^d S_{ij} =1$. The set of all $d$-by-$d$ stochastic matrices $\Scal(d)$ forms a monoid under matrix multiplication. 
	
	The identity element is the unit matrix which is a stochastic matrix. Associativity is a property of matrix multiplication. This leaves closure which can be shown as follows. 
	
	If we have real matrices $\Smat$ and $\Tmat$ whose columns sum to one, such as stochastic matrices, then the sum of the columns of the product matrix $\Smat\Tmat$ are also always one, 
	\beq
	\sum_{i=1}^d [\Smat\Tmat]_{ik} 
	= \sum_{i=1}^d  \sum_{j=1}^d  S_{ij} T_{jt} 
	= \sum_{j=1}^d  T_{jt} 
	=1 \, .
	\label{ae:colsumone}
	\eeq
	At the same time, the entries of the product matrix $\Smat \Tmat$ are non-negative since they are the sum of products of non-negative entries of $\Smat$ and $\Tmat$. However, we also know that since $S_{ij} \leq 1$ then
	\beq
	[\Smat\Tmat]_{ik} 
	= \sum_{j=1}^d  S_{ij} T_{jt} 
	\leq \sum_{j=1}^d  T_{jt} =1
	\eeq
	So entries in the product matrix $\Smat\Tmat$ again lie between zero and one inclusive so $\Smat\Tmat$ is a stochastic matrix if $\Smat$ and $\Tmat$ are stochastic matrices; $\Smat\Tmat\in \Scal$ if $\Smat,\Tmat \in \Scal(d)$, (closure).
	
	For instance, the monoid $\Scal(2)$ of two dimensional stochastic matrices under matrix multiplication may be specified as the set $\{ \Smat(a,b) | \, a,b \in \Rbb, \; 0 \leq a, b \leq 1\}$  where
	\beq
	\Smat(a,b) =
	\begin{pmatrix}
		a & (1-b) \\
		(1-a) & b 
	\end{pmatrix}  \, .
	\label{e:Sabdef}
	\eeq
	Closure in this example is expressed by $\Smat(a_1,b_1)\Smat(a_2,b_2)=\Smat(a,b)$ with $a=a_1a_2+(1-b_1)(1-a_2)$ and $b=b_1b_2+(1-a_1)(1-b_2)$.

	The inverse of a stochastic matrix may not exist. Examples include would the matrix $S_{ij} =1/d$ or indeed any stochastic matrix with two identical columns. So the set of $d$-by-$d$ stochastic matrices $\Scal(d)$ does not form a group under matrix multiplication. 
	For example $\Scal(2)$ we see that
	\beq
	\big(\Smat(a,b)\big)^{-1} =
	\frac{1}{1-a-b}\begin{pmatrix}
		-b & (1-b) \\
		(1-a) & -a 
	\end{pmatrix}  \, .
	\eeq
	So any two-dimensional stochastic matrix where $a+b=1$ has no inverse (an example where two columns are identical).
	
	However, the inverse matrix of many stochastic matrices does exist. If the matrix exists, the columns will sum to one but the values of the entries while real are no longer restricted to lie between zero and one inclusive. To show the columns of the inverse of any stochastic matrix sum to one, should the inverse exist, we start from $\Smat.\Smat^{-1}=\unitmat$. Then we see that using only the property that the column sum of $\Smat$ is one we have that
	\bea
	\sum_{i=1}^d [\Smat.\Smat^{-1}]_{ik} = \sum_{i=1}^d \delta_{ik}  = 1
	\\
	\Rightarrow \quad 
	1 = \sum_{i=1}^d \sum_{j=1}^d S_{ij} [\Smat^{-1}]_{jk} 
	= \sum_{j=1}^d [\Smat^{-1}]_{jk} \, .
	\label{ae:invcolsumone}
	\eea
	A standard property of real matrices is that the inverse has real entries (this can be shown by looking at how such inverses are computed). 
	
	So we can consider a set $\Scalhat(d)$ of real $d$-by-$d$ matrices whose columns sum to one and which have an inverse which I will call general stochastic matrices. These form do form a group with the subset $\Scal(d)$ forming a monoid.
	
	Identity is the unit matrix, and is in this set. Associativity is a property of matrix multiplication. Every element has an inverse (by definition) and, as we showed above in \eqref{ae:invcolsumone}, the inverse has columns which sum to one.  Note the proof in \eqref{ae:invcolsumone} holds for the more general case where $\Smat \in \Scalhat$.
	We have closure in $\Scalhat(d)$ as the product matrix $\Smat.\Tmat$ is clearly a matrix of real values. The product matrix has an inverse $\Tmat^{-1}.\Smat^{-1}$ which exists since $\Smat^{-1}$ and $\Tmat^{-1}$ exist. The columns of the columns of the product $\Smat\Tmat$ sum to one as the proof in \eqref{ae:colsumone} only requires that the columns of $\Smat$ and $\Tmat$ each sum to one which we have in $\Scalhat(d)$. 
	
	Thus the general stochastic matrices, the set $\Scalhat(d)$ of real $d$-by-$d$ matrices whose columns sum to one and which have an inverse, do form a group.
	For instance, the group of two-dimensional generalised stochastic matrices $\Scalhat(2)$ may be specified as the set $\{ \Smat(a,b)| a,b \in \Rbb, \; a+b \neq 1  \}$ where $\Smat(a,b)$ was defined in \eqref{e:Sabdef}.

	\section{Adjacency matrices for group graphs}\label{as:adjmat}
	
	There is a close connection between a \tsedef{gain graph}, a \tsedef{voltage graph}, a \tsedef{biassed graph} and the group graphs defined here. I will examine this relationship here.
	
	A gain graph is built from an algebra $A$ over a field $F$ by adding an extra multiplication rule to a vector space. So we will define our notation by first noting the definition of these two structures.

	\subsection{Background}
	
	\begin{definition}[Vector Space]\label{d:vectorspace}
		A \tsedef{vector space} $\Vcal$ over a field $\Fcal$ is set of objects $\{\vvec\}$ known as \tsedef{vectors} with two  operations, vector addition denoted by ``$+$'' and scalar multiplication denoted by ``$.$''. These obey the following rules
		\begin{enumerate}
			\item Associativity of vector addition $\uvec + (\vvec + \wvec) = (\uvec + \vvec) + \wvec$.
			\item 
			Commutativity of vector addition $\uvec + \vvec = \vvec + \uvec$.
			\item 
			Identity element of vector addition. There exists an element denoted $\zerovec \in \Vcal$, called the \tsedef{zero vector}, such that $\vvec + \zerovec  = \vvec$ for all vectors $\vvec \in V$.
			\item \label{i:idvecadd}
			Inverse elements of vector addition. For every $\vvec \in V$, there exists an inverse element denoted $-\vvec \in V$, called the \tsedef{additive inverse} of $\vvec$, such that $\vvec + (-\vvec) = \zerovec$.
			\item 
			Compatibility of scalar multiplication with field multiplication $a.(b.\vvec) = (ab).\vvec$.
			\item \label{i:idscalarmult}
			Identity element of scalar multiplication	$1.\vvec = \vvec$, where $1 \in \Fcal$ denotes the multiplicative identity in the field $\Fcal$.
			\item
			Distributivity of scalar multiplication with respect to vector addition $a.(\uvec + \vvec) = (a.\uvec) + (a.\vvec)$
			\item
			Distributivity of scalar multiplication with respect to field addition $(a + b).\vvec = (a.\vvec) + (b.\vvec)$
		\end{enumerate}
		Here $\uvec$, $\vvec$ and $\wvec$ are all arbitrary vectors so $\uvec,\vvec,\wvec \in \Vcal$ while the elements of the field, known as \tsedef{scalars}, $a$ and $b$ are arbitrary elements so $a,b \in \Fcal$. 
	\end{definition}	
	
	Notation. There are two distinct types of multiplication here. The multiplication between any two elements $a$ and $b$ of the field is denoted with out any explicit symbol as $ab$. Formally this multiplication without a symbol is a map from $\Fcal \times \Fcal \to \Fcal$. This is strictly different from the multiplication between a scalar $a\in \Fcal$ from the field and any vector $\vvec \in \Vcal$ from the vector space to give a new vectors denoted as $a.\vvec \in \Vcal$. However, it is common to drop the explicit ``$.$'' notation and write $a\vvec \equiv a.\vvec$ when the context and other notation makes it clear that $a$ is a scalar and $\vvec$ is a vector.
	
	We do not give the full definition of a field here as we will restrict our work to example to the field of complex numbers $\Fcal=\Cbb$ which also covers the common cases where scalars are limited further to be real numbers or simply integers. We will just highlight here that these scalars, the elements of fields, include both a zero element $0 \in \Cbb = \Fcal$ and an identity element $1 \in \Cbb = \Fcal$ as needed by the axioms of a field. By restricting ourselves to examples using the field of complex numbers, these $0$ and $1$ are simply the natural integers zero and one as the notation suggests. 
	
	However, the notations 0 and 1 are often used to represent similar zero and identity elements in different constructions. For example, the $1$ in point \ref{i:idscalarmult} of the definition of an algebra is the unit element of the field, that is this \emph{scalar} is the natural number one $1 \in \Cbb = \Fcal$. However the $\zerovec$ in point \ref{i:idvecadd} is the zero \emph{vector} $\zerovec \in \Vcal$ while this can always be given in terms of the zero element of the field, the number zero $0 \in \Fcal$, that is $\zerovec = 0 \vvec$. While these two objects $0 \in\Cbb$ and $\zerovec \in \Vcal$ are formally distinct, they are closely related. Most of the time it will be clear from the context which zero or identity element is referred to but equally, on occasion it can pay dividends to be very precise about what type of object your symbols ``0'' and ``1'' refer to.    
	
	Algebra is a term that is used to describe a type of mathematics involving the manipulation of abstract objects, rather than numbers or geometrical concepts. However, here we will use the term ``algebra'' as a shorthand for a specific structure built on top of a vector space.
	\begin{definition}[Algebra]
		An algebra $\Acal$ over a field $\Fcal$ is a vector space over the field $\Fcal$ with an additional binary operation denoted here by $\times$ mapping a pair of algebra elements back onto the algebra, $\times : \Acal \times \Acal \to \Acal$. The operation $\times$ has to obey three identities for all elements, vectors, in the algebra, $\xvec, \yvec, \zvec \in \Acal$ and for all elements (the ``scalars'') $a$ and $b$ from the field, $a,b \in \Fcal$:
		\begin{enumerate}
			\item Right distributivity: $(\xvec + \yvec) \times \zvec = (\xvec \times \zvec) + (\yvec \times \zvec)$.
			\item Left distributivity:  $\zvec \times (\xvec + \yvec) = (\zvec \times \xvec) + (\zvec \times \yvec)$.
			\item Compatibility with scalars: $(a.\xvec)\times (b.\yvec) = (ab) . (\xvec \times \yvec)$.
		\end{enumerate}
	\end{definition}
	These three axioms define what is called a \tsedef{bilinear binary operation}. At this stage we now have three types of multiplication: scalar multiplication $ab\in\Fcal$, scalar multiplication $a.\vvec \in \Vcal$, and this bilinear binary operation $\times$ that is added by the definition of the algebra, $\uvec \times \vvec \in \Acal$. Again, when the context is clear, it is not uncommon to drop the explicit multiplication symbol used when two elements are combined from any combination of scalar and vector.
	
	We will focus on one particular algebra for our discussion of gain graphs which is the algebra of a group over the complex numbers, denoted as $\Cbb \Gcal$. Here the field used is that of complex numbers hence the $\Cbb$ in the notation. The group multiplication law is the one used for the multiplication $\times$ needed for the algebra. However, a group is not apriori a vector space as there is no multiplication by scalars drawn from a field and we need a second binary operation between group elements to play the role of addition of vectors. The simple solution used here is that we just define these two operation to exist such that they obey the rules required by the vector space and algebra definitions. That is the elements of our algebra are given by
	\beq
	\xvec = \sum_{g \in \Gcal} x_g . g = x_e .e + x_a.a + \ldots \in \Cbb \Gcal
	\label{ae:vecdef}
	\eeq
	where $x_g \in \Cbb = \Fcal$ are the scalars and $g \in \Gcal$ are the group elements which play the role of the vectors in the vector space. One way to think about it is that the scalars $x_g$ are playing the role of the coordinates of the vector and the group elements $g$ are the basis vectors. This is because the vector addition and scalar multiplication rules follow exactly as the notation suggests.

	First the required vector addition operation, ``$+$'' of \defref{d:vectorspace}, is defined as follows 
	\bea
	\xvec + \yvec 
	= 
	\left(\sum_{g \in \Gcal} x_g . g \right) 
	+
	\left(\sum_{h \in \Gcal} y_h . h \right)
	&:=&
	\sum_{g \in \Gcal} (x_g +y_g) . g 
	\eea
	The addition of two elements of the field $\Fcal$ in the factor $(x_g +y_g)$ is simply the addition rule that comes with the field, i.e.\ here just the standard way of adding two complex numbers. 
	
	The multiplication between algebra elements, the bilinear binary operation denoted by ``$\times$'', is simply built on the group multiplication law, which, in a slight abuse of notation, we will also denote by $\times$. Remembering that we have to include the scalars then the multiplication rule of the algebra $\Cbb \Gcal$ for two arbitrary vectors is given by 
	\bea
	\xvec \times \yvec = 
	\left(\sum_{g \in \Gcal} x_g . g \right) 
	\times 
	\left(\sum_{h \in \Gcal} y_h . h \right)
	&:=& 
	\sum_{g,h \in \Gcal} (x_h y_g) . (g \times h) 
	\label{e:CGmult}
	\eea
	The $(x_h y_g)$ is the multiplication of two scalars, two elements of a field, and uses the scalar multiplication which comes with the field, simply the standard rule for multiplying two complex numbers in our case. The $(g \times h)$ uses the group multiplication law. Finally the ``$.$'' indicates the scalar-vector multiplication defined as part of a vector space in \defref{d:vectorspace}.

	\subsection{Adjacency matrix network representation}
	
	Finally we come to the construction of networks, graphs, based on the algebra $\Cbb \Gcal$. This is similar to the approach of \citet{CDD21} but the algebra used here is different.
	
	One approach to networks is to represent them in terms of their adjacency matrix, $\Amat$, where the entry $A_{ji}$ is a complex number which represents the weight or value associated with an edge from node $v_i$ to node $v_j$. In most cases, the edge weights $A_{ji}$ are real, often non-negative integers and indeed even limited to be zero (no edge) or one (unweighted edge present). An entry of zero always represents no edge.  This representation describes any network where only one edge is allowed from any given node $v_i$ to any given node $v_j$, i.e.\ a network without multiedges. Self-loops, weighted edges and directed edges are all encoded nicely by this adjacency matrix representation. 
	
	Typical network analysis then requires that expressions such as $\sum_{j=1}^N A_{ij}A_{jk}$ or the adjacency matrix acting on a vector $\wvec$, so expressions such as $\sum_j A_{ij}w_j$, are all well behaved. When the values of the matrices and vectors are complex numbers (including reals, integers and binary values) we have traditional linear algebra expressions and such expressions are that are widely exploited in network science.
	
	These matrix and vector expressions only require that the entries can be added and multiplied in an appropriate way. So the adjacency matrix entries, the edge weights, can be taken from any field, not just the complex numbers, and much of the linear algebra results from network analysis will still hold. 
	
	However, another option exists where the multiplication and addition needed for network analysis using the adjacency matrix can be based on the bilinear operations of an algebra.  This is one way to look G-gain groups. 
	So suppose now that we work with an adjacency matrix $\Amat$ and node vectors $\wvec$ taking values from the algebra $\Cbb \Gcal$, that is $A_{ji}\in \Cbb \Gcal$ and $w_j \in \Cbb \Gcal$. 
	
	For instance, for a finite group 
	we can interpret matrix multiplication in terms of the rules for multiplication $\times$ of two algebra elements in $\Cbb \Gcal$ given in \eqref{e:CGmult}. We define 
	\beq
	A_{ji} = \sum_g A_{ji}^{(g)} g 
	\, , \quad
	w_{j}  = \sum_g  w_{j}^{(g)} g 
	\eeq
	where  $A_{ji}^{(g)}, w_{j}^{(g)} \in \Cbb$ and $g \in \Gcal$. Note that the lack of an edge is the zero element of the algebra (the identity under the addition), that is, in a slight abuse of notation, the algebra element $O = \sum_g c_g g \in \Cbb \Gcal$ where $c_g=0$ is used to represent no edge present. 
	Then we can use the $\Cbb \Gcal$ algebra rules to evaluate expressions such as
	\bea
	[\Amat\Amat]_{ik} 
	&=& 
	\sum_{j=1}^N A_{ij} \times A_{jk}
	\\
	&=& 
	\sum_{g,h} \left(\sum_{j=1}^N A_{ij}^{(g)} A_{jk}^{(h)} \right).(g \times h) 
	\in \Cbb \Gcal \, .
	\eea
	
	One feature of note with this approach is that a zero ``value'' in an adjacency matrix can play two roles. Here a zero in the adjacency matrix is a shorthand for the zero vector in the vector space, where the complex coordinates are all zero $x_g=0$ in \eqref{ae:vecdef}. Such an entry in the adjacency matrix represents the lack of any relationship between two nodes in one direction. However, we can have groups where the number zero is also representing one of the elements of the group. A good example is the group of integers under addition where the number zero represents the identity element of the group $e$. For instance, this group is used to represent the differences in status, formal and perceived, between individuals in an organisation. In this case, if there is a working relationship but it is between two people of equal status then we use the group element $e$ represented by the number zero as the link label. The entry in the adjacency matrix is the vector $\xvec = \sum_g x_g g = 1.D(e)$ so $x_e=1$ but we use $D(e)=0$, a zero, to label the link. For this group in this representation, the notation ``0'' can mean two very different things.

%
%


\bibliographystyle{abbrvnat} 

\end{document}